\documentclass[letter, 11pt]{article}

\usepackage{natbib}
\usepackage{amssymb,amsmath,graphicx,xspace,colortbl,rotating} %
\usepackage[raggedrightboxes]{ragged2e} 
\usepackage{textcomp}
\usepackage{appendix}  
\usepackage{bm}  
\usepackage{boxedminipage}  
\usepackage[usenames, dvipsnames]{xcolor} 
\usepackage{endnotes}  
\usepackage{ragged2e}  
\usepackage[onehalfspacing]{setspace}  
\usepackage{tabulary}  
\usepackage{varioref}  
\usepackage{wrapfig}  

\usepackage{multimedia}  
\usepackage{pgf,tikz,pgfplots}
\pgfplotsset{width=5cm,compat=1.15}
\usepackage{mathrsfs}
\usetikzlibrary{arrows}
\usepackage{caption,subcaption}
\usepgfplotslibrary{fillbetween}
\definecolor{zzttqq}{rgb}{0.6,0.2,0.0}

\graphicspath{{Stochastic_Path 17 new_graphics/}{Stochastic_Path 17 new_tcache/}{Stochastic_Path 17 new_gcache/}}
\DeclareGraphicsExtensions{.pdf,.eps,.ps,.png,.jpg,.jpeg}
\usepackage[margin=1in]{geometry}
\usepackage {bm}
\usepackage {tabulary}
\usepackage[normalem]{ulem}

\newtheorem {theorem}{Theorem}

\newtheorem {corollary}{Corollary}

\newtheorem {definition}{Definition}
\newtheorem {example}{Example}

\newtheorem {lemma}{Lemma}

\newtheorem {proposition}{Proposition}

\newenvironment {proof}[1][Proof]{\noindent \textbf {#1.} }{\ \rule {0.5em}{0.5em}}

\begin{document}
\title{Quality Selection in Two-Sided Markets: \\ A Constrained Price Discrimination Approach\thanks{We would like to thank Dirk Bergemann, Kostas Bimpikis, Ozan Candogan, Nicole Immorlica, Mike Ostrovsky,  Yiangos Papanastasiou, Steve Tadelis, three anonymous referees and the Associate Editor, as well as seminar participants at Simons Institute, Stanford, INFORMS RMP, YoungEC workshop, and INFORMS annual meeting for helpful comments and advice. B. Light thanks for funding from Stanford GSB since the first version of the paper was finished while being a PhD student there. G. Weintraub thanks Joseph and Laurie Lacob for the support during the 2018-2019 academic year as a Joseph and Laurie Lacob Faculty Scholar at Stanford Graduate School of Business. This work was supported by the National Science Foundation under grant 1931696. }}

\author{Bar Light\protect\footnote{  Microsoft Research, NY, USA. e-mail: \textsf{barlight@microsoft.com}} ~ Ramesh Johari\protect\footnote{  Department of Management Science and Engineering, Stanford University, Stanford, CA 94305, USA. e-mail: \textsf{rjohari@stanford.edu}} ~  Gabriel Y. Weintraub\protect\footnote{  Graduate School of Business, Stanford University, Stanford, CA 94305, USA. e-mail: \textsf{gweintra@stanford.edu}\ }~ ~}
\maketitle
\thispagestyle{empty}

\begin{abstract}
Online platforms collect rich information about participants and then share some of this information back with them to improve market outcomes. In this paper 
we study the following information disclosure problem in two-sided markets: 
If a platform wants to maximize revenue, which sellers should the platform allow to participate, and how much of its available information about participating sellers' quality should the platform share with buyers?
We study this information disclosure problem in the context of two distinct two-sided market models: one in which the platform chooses prices and the sellers choose quantities (similar to ride-sharing), and one in which the sellers choose prices (similar to e-commerce). Our main results provide conditions under which simple information structures commonly observed in practice, such as banning certain sellers from the platform while not distinguishing between  participating sellers, maximize the platform’s revenue.  The platform's information disclosure problem naturally transforms into a constrained price discrimination problem where the constraints are determined by the equilibrium outcomes of the specific two-sided market model being studied. We analyze this constrained price discrimination problem 
 to obtain our structural results. 
 


\end{abstract}

\newpage

\section{Introduction} \label{Sec: Intro}

Online platforms have an increasingly rich plethora of information available about market participants. These include rating systems, public and private written feedback, and purchase behavior, among others.  Using these sources, platforms have become increasingly sophisticated in classifying the quality of the sellers that participate in their platform  (for example, see \cite{tadelis2016reputation}, \cite{filippas2018reputation},  \cite{donaker2019designing}, and \cite{garg2018designing}). This information can be used  both to increase the platform's revenue, and to enhance the welfare of the platform's participants.  For example, cleaning services and ridesharing platforms remove low quality sellers from their platforms. Platforms can also boost the visibility of high quality sellers with certain badges, as is done by online marketplaces such as Amazon Marketplace and eBay. We refer broadly to such market design choices by platforms as \emph{quality selection}.  

In this paper, we study quality selection in two-sided markets. In particular, we investigate which sellers a two-sided market platform should allow to participate in the platform, as well as the optimal amount of information about the participating sellers' quality that the platform should share with buyers in order to maximize its own revenue.
Our results characterize conditions under which simple information structures, such as just banning a portion of low quality suppliers or giving badges to high quality suppliers, emerge as optimal designs.  


We introduce two different two-sided market models with heterogeneous buyers and heterogeneous sellers. Sellers are heterogeneous in their quality levels and buyers are heterogeneous in how they trade-off quality and price. In the first model, the platform chooses prices and the sellers choose quantities (e.g., how many hours to work).  This setting is loosely motivated by labor platforms such as ride-sharing and cleaning services. In the second model, the sellers choose prices, and quantities are determined in equilibrium. This setting is motivated by online e-commerce 
 marketplaces such as Amazon Marketplace. In both models, quality selection by the platform involves deciding  on an {\em information structure}, that is, how much of the information it has about the sellers' quality to share with buyers. The platform's goal is to choose an information structure that maximizes the platform's revenue. The information structure can consist of banning a certain portion of the sellers, and also richer structures that share more granular information with buyers about the quality of participating sellers.

The mapping from the information that the platform shares about the sellers' quality to market outcomes is generally complicated. After the platform chooses an information structure, the buyers and the sellers take strategic actions. Market outcomes such as prices and offered qualities are determined by these strategic actions and the resulting equilibrium conditions, typically including market clearing: not only must the buyers' incentive compatibility and individual rationality constraints be satisfied (as in a standard price discrimination problem, e.g., \cite{mussa1978monopoly}), but the total supply must also be compatible with the total demand.  

First, we observe that the platform’s information disclosure problem transforms into a {\em constrained price discrimination problem.} Every information structure induces a certain subset of price-expected quality pairs which we call a {\em menu}, from which the buyers can choose.  Optimization over feasible menus yields a price discrimination problem. This transformation to a price discrimination problem can be seen as a  revelation principle style argument applied to our setting. Similar observations  were noticed in previous literature (e.g., \cite{bergemann2019information}, \cite{drakopoulos2018persuading}, and \cite{lingenbrink2019optimal}).  The transformation to a constrained price discrimination problem is beneficial in our framework as it allows us to capture different market arrangements and two-sided market models. 


Note that platforms can use the information they collect about the sellers' quality to induce a menu in many different ways. For example, giving badges to high quality sellers can influence the prices such sellers charge, the quantities they sell, and their market entry decisions \citep{hui2018certification}. Similarly, banning some low quality sellers can also influence the prices, the quantities sold, and the participating sellers' quality.

Using our analysis of  the constrained price discrimination problem, we provide  a broad set of conditions under which a simple information structure in which the platform bans a certain portion of low quality sellers and does not distinguish between  participating sellers maximizes the platform's revenue. This resembles a common practice in ride-sharing and cleaning services platforms (in these cases the participating suppliers' review scores are typically so high that they do not reveal much information \cite{tadelis2016reputation}).  To obtain this result, we require two conditions.  First, we require a regularity condition on the induced set of feasible menus in the constrained price discrimination problem; as we suggest later, this regularity condition is natural and likely to be satisfied in a wide range of market models.  Given this regularity condition, our second requirement is an appropriate convexity condition on the demand; as we note, this condition reduces to the requirement that the demand elasticity is not too low. Furthermore, we show that these conditions are generally necessary for the optimality of  the simple information structure in which the platform bans a certain portion of low quality sellers and does not distinguish between  participating sellers.
We also provide results involving only local demand elasticity conditions that guide the market design decision of whether to share less or more information about sellers' quality. We provide a simple example in Section \ref{Sec: toy} that illustrates the key features of our analysis. 

We then apply the equivalence between the constrained price discrimination problem and the information disclosure problem in order to study the two different two-sided market models mentioned above. 
 In both models, the platform's decisions (the platform decides on an information structure and prices in the first model, and on an information structure in the second model) generate a game between buyers and sellers. Given the platform's decisions there are four equilibrium requirements. First, the sellers choose their actions (prices or quantities) to maximize their profits. Second, the buyers choose whether to buy the product and if so, what  (expected) quality to buy to maximize their utility. Third, given the information structure that the platform chooses, the buyers form beliefs about the sellers' qualities that are consistent with Bayesian updating and with the sellers' actions. Fourth, we require market clearing: the total supply equals the total demand. We make this assumption because we envision our setting representing long-run outcomes in which prices should naturally evolve to clear the market. However, this last requirement can be relaxed so our results also apply for other market arrangements where supply and demand can be imbalanced.  

We show that each equilibrium of the game induces a certain subset of price-quality pairs; each pair consists of a price, and the expected quality of sellers selling at that price.  The platform's goal is to choose a menu that maximizes the platform's revenue. Finding the set of equilibrium menus that the platform can choose from depends on the equilibrium outcomes of the game. Hence, this set is determined by the specific two-sided market model being studied and can be challenging to characterize. 
For our first model (in which the platform sets prices), we show that for every information structure there exists a strictly convex optimization problem whose unique solution yields the unique menu of induced price-quality pairs. For the second model, Bertrand competition between the sellers pins down the equilibrium prices, so we are able to explicitly provide the menu that each information structure induces.  In each setting, we then leverage the analysis of the constrained price discrimination problem to (1) characterize the platform’s optimal information disclosure, and in particular to find conditions under which the policy of banning low quality sellers, and not distinguishing between the remaining high quality sellers, is optimal; and (2) to provide local improvement results.
 
 The rest of the paper is organized as follows. Section \ref{Sec: Related lit} discusses related literature. In Section \ref{Sec: toy} we describe a simple example that captures the main features of our analysis. In Section \ref{Sec: price-dis} we study the general constrained price discrimination problem. 
 In Section \ref{Sec: INFORMATION} we present the platform's initial information and information structures. In Section \ref{Sec: Model 1 quant} we present our first model where the platform chooses prices and the sellers choose quantities. In Section \ref{Sec: Two-sided2 prices} we present our second model where the sellers choose prices and quantities are determined in equilibrium. 
 In Section \ref{Sec: summary} we provide concluding remarks. All proofs are provided in the Appendix that follows.

\subsection{Related Literature} \label{Sec: Related lit}

Our paper is related to several strands of literature. We discuss each of them separately below. 

\noindent \textbf{Information design}. There is a vast recent literature on how different information disclosure policies influence the decisions of strategic agents and equilibrium outcomes in different settings. Applications include Bayesian persuasion (\cite{aumann1966game} and \cite{kamenica2011bayesian}), dynamic contests \citep{bimpikis2019designing}, matching markets \citep{ostrovsky2010information}, queueing models \citep{lingenbrink2019optimal}, games with common interests \citep{lehrer2010signaling}, transportation \citep{meigs2020optimal}, inventory systems \citep{kostami2019price}, ad-auctions \citep{Ash2018}, exploration in recommendation systems (\cite{papanastasiou2017crowdsourcing} and \cite{immorlica2019bayesian}), social networks (\cite{candogan2017optimal} and \cite{candogan2019persuasion}), social services \citep{anunrojwong2020information}, the retail industry (\cite{lingenbrink2018signaling} and \cite{drakopoulos2018persuading}), warning policies \citep{alizamir2020warning}, and many more.  (See \cite{candogan2020information} for a recent review of information design in operations.)

In this paper we focus on the amount of information about the sellers’ quality that a two-sided market platform should share with buyers. 
Similar to the Bayesian persuasion literature, we reformulate the platform's optimization problem. In the Bayesian persuasion literature, it can be shown that the platform's (sender) payoffs are determined by the receivers' posterior beliefs. The standard approach is to optimize over these posterior beliefs instead of over information structures. This approach leads, at least in some cases, to sharp characterizations of the optimal information disclosure policy (see, e.g., \cite{aumann1966game} and \cite{kamenica2011bayesian}). In our setting,  the platform's payoffs are determined by the buyers' (i.e., the receivers) equilibrium posterior quality means and by the equilibrium prices. Our approach is to optimize jointly over posterior means and prices, and thus, we transform the information disclosure problem to a price discrimination problem. Similar transformations in different settings were observed in previous literature (for example,    \cite{drakopoulos2018persuading} and \cite{lingenbrink2018signaling} study the gains from personalized information provision and provide similar observations). In our setting, the transformation to  price discrimination is beneficial mainly because it allows one to capture many different two-sided market models and market arrangements. We discuss in Section \ref{Subsec:ImbalancesModel1} market arrangements where our transformation to a constrained price discrimination problem and the results we obtain using it  fail. 

Our information disclosure problem is, in principle, under the umbrella of Bayesian Persuasion problems with a continuum of receivers, however, existing results do not apply to our setting.  The conditions we provide that imply the optimality of a simple information structure where the platform bans certain sellers from the platform while not distinguishing between participating sellers inherently relate to the receivers types' distribution (their valuations in our setting), the receivers' utility functions, and to the structure of the total supply and the total demand in the specific two-sided market that is being studied. In particular, the optimal information structure depends on demand elasticities that relate to the buyers valuations' distribution and on the constraint set that depends on the two-sided market model being studied. 


\noindent \textbf{Nonlinear pricing}. Nonlinear pricing schemes are widely studied in the economics and management science literature (see \cite{wilson1993nonlinear} for a textbook treatment). The price discrimination problem that we consider in this paper is closest to the classical second-degree price discrimination problem , cf.~\citet{mussa1978monopoly} and \citet{maskin1984monopoly}.  

The problem that the platform solves in our setting differs from the previous literature on price discrimination in two major aspects. First, the costs for the platform from producing higher quality products are zero. This is because in the two-sided market models that we study, the costs of producing a higher quality product are incurred by the sellers and not by the platform. Hence, the platform's revenue maximization problem transforms into a constrained price discrimination problem with no costs.\footnote{In that sense our problem resembles the classic mechanism design problem studied in \cite{myerson1981optimal} but with constraints on the possible allocations. } Second, the platform cannot simply choose any subset of price-quality pairs (menus) that satisfies the incentive compatibility and individual rationality constraints. The set of  menus from which the platform can choose is determined by the additional equilibrium requirements described in the introduction.

These differences significantly change the analysis and the platform's optimal menu. First, a key part of our analysis is to incorporate equilibrium constraints into the price discrimination problem. 
In addition, under the regularity assumption that the virtual valuation function is increasing, \cite{mussa1978monopoly} show that the optimal menu assigns different qualities of the product to different types. In contrast, the results in our paper are drastically different: under certain regularity assumptions, the optimal menu assigns the same quality of the product to different types.\footnote{Another difference from most of the previous literature is that in our model each menu is finite (i.e., there is a finite number of price-quality pairs), and thus the standard techniques used to analyze the price discrimination problems in the previous literature cannot be used.  \cite{bergemann2011mechanism} study a price discrimination problem with a finite menu in order to study a setting with limited information. However, because the platform's costs are zero in our setting, we cannot use the Lloyd-Max optimality condition that \cite{bergemann2011mechanism} employs.}

\noindent \textbf{Two-sided market platforms}. Recent papers study how platforms can use information and other related market design levers to improve market outcomes. In the context of matching markets, \cite{arnosti2018managing} and \cite{kanoria2017facilitating} suggest different restrictions on the  agents' actions in order to mitigate inefficiencies that arise in those markets. \cite{vellodi2018ratings} studies the role of design of rating systems in shaping industry dynamics. In   \cite{romanyuk2019cream} the platform designs what buyer information the sellers should observe before the platform decides to form a match. 

The paper most closely related to ours is the contemporaneous work by \cite{Yiangos2019} that studies the interaction between information disclosure and the quantity and quality of the sellers participating in the platform. Studying a dynamic game theoretic model, \cite{Yiangos2019} focuses on how information design influences supply-side decisions, showing that  information design can be a substitute to charging lower fees when solving the ``cold start” problem. As in our paper, in the papers noted above the full disclosure policy is not necessarily optimal, and hiding information can increase the social welfare and/or the platform's revenue. 
\section{A Simple Motivating Model} \label{Sec: toy}
In this section we provide a simple model that illustrates many important features of our paper. While this model ignores important features of our more general model, it will be helpful to highlight important aspects of our analysis and main results. 

Consider a platform where heterogeneous sellers and heterogeneous buyers interact. In our simple model of this section, there are two types of sellers: high quality sellers $q_{H}$ and low quality sellers $q_{L}$ with $q_{H} > q_{L}>0$. The platform knows the sellers' quality and considers two policies. Policy $B$ is to ban the low quality sellers and keep only the high quality sellers on the platform. Policy $K$ is to keep both low quality and high quality sellers on the platform and share the information about the sellers' quality with the buyers. 

The total supply of products by sellers whose quality level is $i=H,L$ is given by the function $S_{i}(p^{j}_{i})$.  When the platform chooses policy $j=B,K$, $p^{j}_{i}$ is the price of the product sold by sellers whose quality level is $i=H,L$. We assume that the total supply is increasing in the price. The total supply can also depend on the mass of sellers whose quality level is $i=H,L$ and on the sellers' costs.
In our two-sided market models the supply function will be micro-founded, but we abstract away from these details for now.

Buyers are heterogeneous in how much they value quality relative to price. A buyer with type $m$ that decides to purchase from a seller whose quality level is $i=H,L$ has a utility $mq_{i}-p^{j}_{i}$. We normalize the utility associated to not buying to zero. The distribution of the buyers' types is described by a probability distribution function $F$. We assume that $F$ admits a density function $f$.  The buyers choose to buy or not to buy the product from sellers whose quality level is $i=H,L$ in order to maximize their own utility. The buyers' decisions generate demand for quality $i=H,L$ sellers $D^{K}_{i}(p^{K}_{L},p^{K}_{H})$ when the platform chooses policy $K$, and demand for quality $H$ sellers $D^{B}_{H}(p^{B}_{H})$ when the platform chooses policy $B$ (when the platform chooses option $B$, there is no demand for low quality sellers as they are banned). 

The platform's goal is to choose a policy that maximizes the total transaction value given that prices form an equilibrium. Equilibrium requires that the market must clear: that is, supply must equal demand. Note that if the platform charges commissions from each side of the market, maximizing the total transaction value is equivalent to maximizing the platform's revenue. For this reason, we will refer to the platform's objective as ``revenue'' or ``total transaction value'' interchangeably. If the platform chooses policy $B$, then the total transaction value is $p^{B}_{H}D^{B}_{H}(p^{B}_{H})$ and the equilibrium requirement is  $S_{H}(p^{B}_{H})=D^{B}_{H}(p^{B}_{H})$. If the platform chooses policy $K$, then the total transaction value is 
$$p^{K}_{H}D^{K}_{H}(p^{K}_{L},p^{K}_{H})+p^{K}_{L}D^{K}_{L}(p^{K}_{L},p^{K}_{H})$$ and the equilibrium requirements are \begin{equation}
    \label{Eq: Sup-de toy}
 S_{H}(p^{K}_{H})=D^{K}_{H}(p^{K}_{L},p^{K}_{H}) \text{ and } S_{L}(p^{K}_{L})=D^{K}_{L}(p^{K}_{L},p^{K}_{H}). 
 \end{equation}
For simplicity, we assume that the prices that satisfy the equilibrium requirements are unique. That is, $(p^{K}_{L},p^{K}_{H})$ are the unique prices that solve the equations in (\ref{Eq: Sup-de toy}) and $p^{B}_{H}$ is the unique price that solves $D^{B}_{H}(p^{B}_{H})=S_{H}(p^{B}_{H})$. In this case, the platform's revenue maximization problem transforms into a \emph{constrained price discrimination problem}. Choosing policy $B$ is equivalent to showing the buyers the price-quality pair $(q_{H},p^{B}_{H})$, while choosing policy $K$ is equivalent to showing the buyers the price-quality pairs $(q_{H},p^{K}_{H})$ and $(q_{L},p^{K}_{L})$. Hence, each policy is equivalent to a subset of price-quality pairs that we call a {\em menu}, and the platform's goal is to choose the menu with the higher revenue. This transformation to a price discrimination problem is useful for capturing different two-sided markets with different assumptions on the behavior of sellers and buyers, and different market arrangements. We  note that the transformation is not a reduction in the sense that the number of menus that the platform chose from is not smaller than the number of policies.

In this example, we assume that the demand matches the supply perfectly. In general two-sided market models there can be supply and demand imbalances. In Section \ref{Subsec:ImbalancesModel1}, we study these imbalances and show that our main results hold for the case where the equilibrium conditions allow  supply to be  greater than or equal demand. On the other hand, when demand is greater than supply, the transformation to a price discrimination problem fails and our techniques do not apply.


In our simple model, the set of feasible menus (denoted by $\mathcal{C}$) contains only two menus. We introduce our general model in Section \ref{Sec: price-dis}, where we study a general price discrimination problem with a rich set of possible menus $\mathcal{C}$, defined by a general constraint set.  Furthermore, in the model we consider in this section, the sellers' qualities are fixed and the prices are constrained by the equilibrium requirements. In the general two-sided market models we consider (see Sections \ref{Sec: Model 1 quant} and \ref{Sec: Two-sided2 prices}), the expected qualities are also determined in equilibrium. Hence, the set of feasible menus $\mathcal{C}$ in the corresponding price discrimination problem is determined by the specific two-sided market model that we study. When the market model is complex, characterizing the set $\mathcal{C}$ can be challenging as it requires computation of the equilibria of the two-sided market model.   

While the price discrimination problem in this example is simple, we later show that we can solve a general constrained price discrimination problem with similar arguments (see Section \ref{Sec: price-dis}).  We analyze the price discrimination problem in two stages.  In the first stage, we compare the revenue from policy $K$ (showing the price-quality pairs $(q_{H},p^{K}_{H})$ and $(q_{L},p^{K}_{L})$) to the revenue from the {\em infeasible policy} $I$: showing the price-quality pair $(q_{H},p^{K}_{H})$. Policy $I$ might be infeasible because while the pair $(q_{H},p^{K}_{H})$ and $(q_{L},p^{K}_{L})$ clears the market, only showing $(q_{H},p^{K}_{H})$ will generally not do so: demand will be higher than supply.

Note that the equilibrium requirements imply that the price of the product sold by high quality sellers is higher than the price of the product sold by low quality sellers, i.e., $p^{K}_{H}>p^{K}_{L}$. Now, if the platform were to choose policy $I$ then fewer buyers would participate in the platform compared to policy $K$, but the participating buyers would pay the higher price $p^{K}_{H}$. Policy $I$ would be better than policy $K$ if and only if the revenue gains from the participating buyers that pay a higher price when choosing $I$ instead of $K$ outweigh the revenue losses from the mass of buyers that do not participate in the platform when choosing $I$ instead of $K$. This depends on the {\em elasticity of the density function}  $\partial \ln f(m)/\partial \ln m$. Intuitively, when the density function's elasticity is not too ``low'' the mass of buyers that the platform loses is not too ``high''. We show in Theorem \ref{Theorem: Main} a general version of the following: when the density function's elasticity is bounded below by $-2$, policy $I$ yields more revenue than policy $K$ (see a detailed analysis of the elasticity condition in Section \ref{Sec: price-dis}).

In the second stage of the analysis, we compare the revenue from policy $B$ to the revenue from (potentially infeasible) policy $I$. The equilibrium requirements imply that $p^{B}_{H} \geq p^{K}_{H}$. To see this, note that $D^{B}_{H}(p^{K}_{H}) \geq D^{K}_{H}(p^{K}_{L},p^{K}_{H})=S_{H}(p^{K}_{H})$, i.e., the demand for high quality sellers in policy $B$ is greater than the demand for high quality sellers in policy $K$ when the price is $p^{K}_{H}$. This follows because for some buyers, buying from the high quality sellers yields a positive utility that is smaller than the utility from buying from the low quality sellers. Hence, in policy $B$, these buyers buy from the high quality sellers, while in policy $K$ they buy from the low quality sellers. Thus, the demand for high quality sellers under the price $p^{K}_{H}$ exceeds the supply.\footnote{We can also see this in Figure \ref{Fig: 1} in Section \ref{Sec: price-dis} that shows the demand for a given specific prices and qualities. In the figures, in the left column, the platform chooses the policy $K$. The black color represents the buyers that choose to not participate in the platform, the green color represents the buyers that choose $L$, and the red color represents the buyers that choose $H$.  In the figures in the right column, the platform chooses policy $B$, the black color represents the buyers that choose to not participate in the platform, and the orange color represents the buyers that choose  $H$. Note that buyers whose valuations are between $2$ and $2.5$ choose $L$ when the platform chooses $K$ but choose $H$ when the platform chooses $B$. Hence, the demand for  high quality sellers is greater under $B$ than under $K$. } Because the supply is increasing and the demand is decreasing in the price, we must have $p^{B}_{H} \geq p^{K}_{H}$ so that the market clears. 

Before proceeding with the second stage of the analysis, we note that for some models it is the case that $p^{B}_{H} = p^{K}_{H}$, like in the Bertrand competition model that we study in Section \ref{Sec: Two-sided2 prices}. In this model, because  
supply is perfectly elastic prices drop down all the way to marginal cost independently of whether low quality sellers participate in the platform. In this case, this second stage of the analysis is not necessary.

Now, if the platform shows the buyers the menu $(q_{H},p)$ only the buyers whose valuations satisfy $mq_{H} - p \geq 0$ buy the product from the high quality sellers. Thus, $pD^{B}_{H}(p)= p(1 - F\left (p/q_{H}) \right )$. When the density function's elasticity is bounded below by $-2$, the revenue function $R_{H}(p) := p \left (1 - F\left (p/q_{H}) \right ) \right)$ is concave in the price $p$. Thus, as  shown in Figure \ref{Fig 2: Revenue} below, policy $B$ yields more revenue than policy $I$ if the equilibrium price $p^{B}_{H}$ is lower than the {\em monopoly price} $p^{M}_H$, i.e., the unconstrained price that maximizes the platform's revenue $p^{M}_{H}$ ignoring equilibrium conditions:
\begin{equation*} 
    p^{M}_{H}=\operatorname{argmax}_{p \geq 0}p \left (1 - F\left (\frac{p}{q_{H}} \right ) \right ).
\end{equation*}

\begin{figure}[ht]
    \centering
   
\definecolor{rvwvcq}{rgb}{0.08235294117647059,0.396078431372549,0.7529411764705882}
\begin{tikzpicture}[line cap=round,line join=round,>=triangle 45,x=1cm,y=1cm]
\begin{axis}[
    axis lines = left,
    xmin=-0,
xmax=1,
ymin=0,
ymax=0.35,
    xlabel = $p$,
    ylabel = {$R_{H}(p)$},
]

\addplot [line width=2pt] [
    domain=0:1, 
    samples=100, 
    color=black,
]
{x-x^(2))};
\begin{scriptsize}
\draw [fill=rvwvcq] (0.5,0.25) circle (2.5pt);
\draw[color=rvwvcq] (0.5,0.3) node {$p_{H}^{M}$};
\draw [fill=rvwvcq] (0.3,0.21) circle (2.5pt);
\draw[color=rvwvcq] (0.3,0.26) node {$p_{H}^{B}$};
\draw [fill=rvwvcq] (0.1,0.09) circle (2.5pt);
\draw[color=rvwvcq] (0.07,0.14) node {$p_{H}^{K}$};
\end{scriptsize}

\end{axis}
\end{tikzpicture}

\caption{The platform's revenue as a function of the price.}
\label{Fig 2: Revenue}
\end{figure}

 Intuitively, the equilibrium price $p_B^H$ is lower than the price that maximizes the platform's revenue $p_H^M$ if the total supply of high quality sellers is large enough. In particular, if the total supply of high quality sellers exceeds the total demand under the price $p^{M}_{H}$, then the equilibrium price $p^{B}_{H}$ must be lower than $p^{M}_{H}$ to ensure the market clears. In many two-sided markets, competition between platforms and between sellers, platform subsidies on the supply side, penetration pricing strategies, and other factors decrease equilibrium prices considerably. Hence, in our context it is natural to assume that the monopoly price is higher than or equal to the equilibrium price, i.e., $p^{H}_{B} \leq p^{M}_{H}$. In addition, if the equilibrium price was higher than the price that maximizes the platform's revenue the platform could  introduce balanced transfers for each side of the market, i.e., paying suppliers and charging buyers in order to decrease the equilibrium price.

 In the general two-sided market models that we study in Sections \ref{Sec: Model 1 quant} and \ref{Sec: Two-sided2 prices}, the qualities are also determined in equilibrium and the set of possible menus that the platform can choose from can be very large.
We will call this set \textit{regular} if it satisfies a general version of the conditions   $p^{M}_{H} \geq p^{B}_{H} \geq p^{K}_{H}$ discussed above. That is, the set is regular if removing low quality sellers increases the equilibrium price for high quality sellers; and if, in addition,  the monopoly price is higher than this equilibrium price.  These conditions give rise to natural constraints on the equilibria that can arise in the two-sided market models that we study (see the discussion after Definition \ref{Def: marketN} in Section \ref{Sec: price-dis}).

We conclude that when the elasticity of the density function is not too low, and the monopoly price is higher than the equilibrium price, then policy $B$ yields more revenue than policy $K$. That is, banning low quality sellers and keeping only the high quality sellers yields more revenue than keeping both low quality and high quality sellers on the platform and distinguishing them for buyers. In the next sections we study this and other structural results in the context of general two-sided market models and information structures.

\section{A Constrained Price Discrimination Problem}\label{Sec: price-dis}

In the simple model of the previous section, we observed that the platform's problem of choosing how much information to share with the buyers about the sellers' quality transforms into a price discrimination problem with constraints on the menu that can be chosen by the platform. In this section,  we study a general constrained price discrimination problem; the simple model in the previous section is a special case.  In the price discrimination problem we consider, the platform chooses a subset of price-quality pairs, i.e., a {\em menu}, from a feasible space of possible menus (referred to as the {\em constraint set}).  The constraint set restricts the possible choices of menus available to the platform.

In the two-sided market models that we study in Sections \ref{Sec: Model 1 quant} and \ref{Sec: Two-sided2 prices}, the constraint set is determined by the endogenously-determined equilibrium in these markets: i.e., the price-quality pairs in the menu must form an equilibrium, in the sense that the prices and qualities agree with the buyers' and sellers' optimal actions, and supply equals demand. Different two-sided market models generate different constraint sets. In this section, we consider a general constraint set. The platform's problem is to choose a subset of price-quality pairs (the menu) that belongs to the constraint set in order to maximize the total transaction value, while knowing only the distribution of valuations of possible buyers. As previewed in the simple model of the previous section, in Sections \ref{Sec: Model 1 quant} and \ref{Sec: Two-sided2 prices} we will show that the platform's information disclosure problem in our  two-sided market models transforms into the constrained price discrimination problem that we study in this section. 

The rest of this section is organized as follows. First, we provide preliminary concepts followed by our main characterization result regarding the optimality of simple menus. Then, we discuss the necessity of the conditions required for this characterization. We finish by showing local improvement results and structural results for more general distribution functions than the ones assumed in the main result.

\subsection{Preliminaries}

In this subsection we collect together basic concepts needed for our subsequent development.

\noindent {\bf Menus}.  A {\em menu} $C$ is a finite set of price-quality pairs.

\noindent {\bf Constraint set}.  We denote by $\mathcal{C}$ the nonempty set of all possible menus from which the platform can choose. $\mathcal{C}$ is called a \emph{constraint set}. We provide in Example \ref{Example: 1-rich} examples of constraint sets.

\noindent {\bf Buyers}.  We assume a continuum of buyers.  Given a menu, the buyers choose whether to buy a unit of the product and if so, at which price-quality pair to buy it.  Each buyer has a {\em type} that determines how much they value quality relative to price. The utility of a type $m$ buyer over price $(p)$-quality $(q)$ combinations is  $mq-p$. The type  distribution is given by a continuous cumulative distribution function $F$ with a density function $f$. We assume that $F$ is supported on an interval $[a ,b] \subseteq \mathbb{R}_{ +} : =[0 ,\infty )$.\footnote{All the results in the paper can be extended to the case that the utility of a type $m$ buyer over price-quality combinations is  $z(m)q-p$ for some strictly increasing function $z$. In this case we can define the distribution function $\bar{F}:=F(z^{-1})$ and our results hold when the assumptions on $F$ are replaced by the same assumptions on $\bar{F}$. }  Our results also hold in the case that the support of $F$ is unbounded.

\noindent {\bf Sellers.} 
In this section, we will not directly model the sellers. Depending on the particular model in focus, sellers might either set prices or determine quantities. In subsequent sections, namely Sections \ref{Sec: Model 1 quant} and \ref{Sec: Two-sided2 prices}, we will delve into the specifics of sellers' decisions. While we might not explicitly discuss sellers' actions in this section, their behavior and the ensuing market structure determine the constraint set mentioned earlier, which is represented as a general set in this section. Consequently, in the two-sided market models we study, the prices that can arise are constrained by both the sellers and buyers' behavior and the specific market arrangement. We note that the two-sided market models that we study allow for imbalances between supply and demand. Our focus in this paper is on a demand-constrained market where supply meets or surpasses demand. This demand-constrained market scenario is common in real-world marketplaces, a point we elaborate on in Section 5.6. This premise also indicates that the platform's total transaction value hinges on buyer demand. We now describe the platform's optimization problem.

\noindent {\bf Platform optimization problem and optimal menus.}  Given the constraint set $\mathcal{C}$, the platform chooses a menu $C =\left \{\left (p_{1} ,q_{1}\right ) ,\ldots  ,\left (p_{k} ,q_{k}\right )\right \} \in \mathcal{C}$ to maximize the total transaction value, subject to the standard incentive compatibility and individual rationality constraints.\footnote{Because we study two-sided markets where the 
 price-quality menus are finite, we focus on finite menus in our analysis. Our results can be readily extended for infinite menus by standard arguments.}  

In other words, the platform chooses a menu $C \in \mathcal{C}$ to maximize:
  \begin{equation*}\pi \left (C\right ) : =\sum _{ (p_{i} ,q_{i})  \in C}p_{i}D_{i}(C),
\end{equation*}
where $D_{i}(C)$ is the total mass of buyers that choose the price-quality pair $(p_{i} ,q_{i})$ when the platform chooses the menu $C \in \mathcal{C}$. That is,\footnote{If there is a subset of price-quality pairs $C^{\prime}$ such that for some type $m$ buyer we have $mq_{i}-p_{i} \geq 0$ and $mq_{i} -p_{i} =\max _{i \in C}mq_{i} -p_{i}$ for all $(p_{i},q_{i}) \in C^{\prime}$ then we assume that the buyer chooses the price-quality pair with the highest index. This assumption does not change our analysis because we assume that $F$ does not have atoms.}  
\begin{equation*}D_{i}(C) : =\int _{a}^{b}1_{\{m :mq_{i} -p_{i} \geq 0\}}(m)1_{\{m :mq_{i} -p_{i} =\max _{(p_{i},q_{i}) \in C}mq_{i} -p_{i}\}}(m)F(dm),
\end{equation*}
where $1_{A}$ is the indicator function of the set $A$. A menu $C^{\prime} \in \mathcal{C}$ is called {\em optimal} if it maximizes the total transaction value, i.e., $C^{\prime} = \operatorname{argmax}_{C \in \mathcal{C}} \pi (C)$.

\noindent {\bf $k$-separating menus}.  Let $\mathcal{C}_{p}=\{C \in \mathcal{C}: D_{i}(C)>0 \text{ for all } (p_{i} ,q_{i}) \in C \}$ be the set that contains all the menus $C$ such that the mass of buyers that choose the price-quality pair $(p_{i} ,q_{i})$ is positive for every  $(p_{i} ,q_{i}) \in C$. A menu $C =\left \{\left (p_{1} ,q_{1} \right ) ,\ldots  ,\left (p_{k} ,q_{k}\right )\right \} \in \mathcal{C}_{p}$ is said to be {\em $k$-separating} for a positive integer $k$ if $C$ contains exactly $k$ different price-quality pairs. That is, a $k$-separating menu $C$ satisfies $\vert C\vert  =k$ where $\vert C\vert $ is the number of price-quality pairs on the menu $C$. We let $\mathcal{C}_{1} \subseteq \mathcal{C}_{p}$ be the set of all $1$-separating menus. For the rest of the section, we assume without loss of generality that prices are labeled so that $p_{1} \leq p_{2} \leq \ldots  \leq p_{k}$ for every $k$-separating menu $C =\left \{\left (p_{1} ,q_{1}\right ) ,\ldots  ,\left (p_{k} ,q_{k}\right )\right \}$. 

\subsection{Optimality of $1$-Separating Menus}

The main result of this section (Theorem \ref{Theorem: Main}) shows that under certain conditions, a $1$-separating menu is optimal. Translating this to the two-sided market model, it means that the platform bans a portion of the sellers and provides no further  information to buyers about the quality of the remaining sellers that participate in the platform.  

Our theorem shows that this result holds under two key conditions on the model, each of which is related to the conditions discussed in Section \ref{Sec: toy}.   The first is a regularity condition that will be satisfied by a wide range of two-sided market models, including those we consider in this paper. The second is the convexity of $F(m)m$ which relates to demand elasticities. We now discuss each condition in turn. 

\noindent {\bf Regularity}.  The first condition that we introduce is regularity. This condition imposes natural restrictions on the possible constraint set that can arise in the two-sided market models. As we discussed in Section \ref{Sec: toy}, the constraint set in the price discrimination problem describes the set of menus that is generated from buyers' and sellers' behavior and the specific market arrangement. Hence, the condition on the constraint set that we describe next relates to the properties of the two-sided market models under consideration. 

\begin{definition}\label{Def: marketN}
We say that the constraint set $\mathcal{C}$ is regular if the following two conditions hold: 

(i) If $C =\left \{\left (p_{1} ,q_{1}\right ) ,\ldots  ,\left (p_{k} ,q_{k}\right )\right \} \in \mathcal{C}_{p}$ then there exists a feasible $1$-separating menu $\{(p,q)\} \in \mathcal{C}_{1}$ such that $p \geq p_{k}$ and $q \geq q_{k}$.\footnote{Recall that we assume without loss of generality that $p_{1} \leq p_{2} \leq \ldots  \leq p_{k}$ for every  menu $C =\left \{\left (p_{1} ,q_{1}\right ) ,\ldots  ,\left (p_{k} ,q_{k}\right )\right \}$.} 

(ii) Let $\{(p,q)\} \in \mathcal{C}_{1}$ be such that $p \geq p'$ for all $\{(p',q')\} \in \mathcal{C}_{1}$. Then $p \leq p^{M}(q)$.\footnote{Recall that given some quality $q$, the monopoly price ignoring equilibrium conditions, $p^{M}(q)$ is given by 
\begin{equation*} 
    p^{M}(q)= \inf \operatorname{argmax}_{p \geq 0}p \left (1 - F\left (\frac{p}{q} \right ) \right ).
\end{equation*}
}
\end{definition}
Condition (i) in Definition \ref{Def: marketN} can be interpreted in the two-sided market models as follows: For a feasible menu (i.e., a menu that can arise in equilibrium), suppose that the platform bans  all  sellers other than the highest quality sellers in that menu, then there is a feasible menu  with just one price-quantity pair, reflecting the enhanced price and quality of these high quality sellers. 
 This is a natural condition in markets as decreasing the supply of low quality sellers increases the demand for high quality sellers (see Section \ref{Sec: toy}).  Condition (ii) in Definition \ref{Def: marketN} means that when the platform uses a $1$-separating menu, the highest price that can arise in the two-sided market model is lower than the monopoly price. As we discussed in Section \ref{Sec: toy}, this is also a natural condition because market factors such as competition and  subsidizing supply suggest that the equilibrium price should be lower than the monopoly price.

In the two-sided market models that we study, a sufficient condition that implies condition (ii) in Definition \ref{Def: marketN} is that the supply of high quality sellers is not very low. In this case, the equilibrium price is not very high and condition (ii) holds (see Section \ref{Sec: Model 1 quant}). The two conditions in Definition \ref{Def: marketN} generalize the regularity condition discussed in the simple model we presented in  Section \ref{Sec: toy}. 
We believe that regularity is a mild condition over two-sided market models; hence, we think of the demand elasticity condition that we introduce next as the primary determinant of the optimality of $1$-separating menus.

\noindent {\bf Convexity of $F(m)m$}. The second condition that we require is the convexity of $F(m) m $.  If we suppose that $F$ has a strictly positive and continuously differentiable density $f$, then an elementary calculation shows that $F(m)m$ is convex if and only if:
\begin{equation*}
\frac{ \partial f(m)}{ \partial m}\frac{m}{f(m)} =\frac{f^{ \prime }(m)m}{f(m)} \geq  -2.
\end{equation*}

In words, the {\em elasticity} of the density function must be bounded below by $-2$. 
A number of distributions satisfy this condition, e.g., power law distributions ($F(m) = d + cm^k$ for some constants $k > 0$, $c$, $d$); beta distributions ($f\left (m\right ) =\frac{\Gamma \left (\alpha  +\beta \right )}{\Gamma \left (\alpha \right )\Gamma \left (\beta \right )}m^{\alpha  -1}\left (1 -m\right )^{\beta  -1}$ with $\beta \leq 1$, where $\Gamma $ is the gamma function); and Pareto distributions ($F\left (m\right ) =1 -\genfrac{(}{)}{}{}{c}{m}^{\alpha }$ on $[c,\infty)$, where $c \geq 1$ is a constant and $\alpha \leq 1$).  It is also worth noting that the condition that $F(m)m$ is convex is distinct from monotonicity of the so-called {\em virtual value function} $r(m):= m-(1-F(m))/f(m)$, a condition that plays a key role in the price discrimination literature.\footnote{See \cite{mussa1978monopoly} and \cite{maskin1984monopoly}, and more generally the mechanism design literature (e.g., \cite{myerson1981optimal}), for use of the monotonicity of the virtual valuation function.  Convexity of $F(m)m$ can be shown to be equivalent to monotonicity of the {\em product} of the virtual valuation with the density, $r(m)f(m)$.} 

To see the dependence on the density function's elasticity, consider a simple price discrimination setting inspired by the example of Section \ref{Sec: toy}.  In particular, suppose that the platform has only two price-quality pairs available: $(p_L, q_L) = (1,1.5)$ and $(p_H, q_H) = (2,4)$, and the platform can either choose the $1$-separating menu $\{ (p_H, q_H) \}$ consisting of high quality only, or the full ($2$-separating) menu $\{ (p_L, q_L), (p_H, q_H) \}$ consisting of both qualities.  In Figure \ref{Fig: 1} we demonstrate the consequences of different elasticities of $f$.  In the figures in the left column, the platform chooses the full menu, the black color represents the buyers that choose not to  participate in the platform, the green color represents the buyers that choose $L$, and the red color represents the buyers that choose $H$.  In the figures in the right column, the platform chooses the $1$-separating high quality menu, the black color represents the buyers that choose to not participate in the platform, and the orange color represents the buyers that choose to buy the product. 

The $1$-separating high quality menu yields more revenue than the full menu if and only if the area between the points $B$ and $C$ times $p_{H}$ is greater than or equal to the area between the points $A$ and $C$ times $p_L$, that is, the revenue losses from losing the participation in the platform of buyers whose valuations are between $1.5$ and $2$ are smaller than the revenue gains from charging the participating buyers whose valuations are between $2$ and $2.5$ the higher price.  Intuitively, when the elasticity is lower, this difference is higher. In other words, when the elasticity is lower, the full menu is more attractive because the platform loses too much revenue when choosing the $1$-separating high quality menu instead.

\begin{figure}[htb]
\centering

\definecolor{rvwvcq}{rgb}{0.08235294117647059,0.396078431372549,0.7529411764705882}
\begin{subfigure}{0.25\textwidth}
\begin{tikzpicture}[scale=0.75]
\begin{axis}[
    axis lines = left,
       xmin=1,
xmax=3,
ymin=0,
ymax=1,
    xlabel = $m$,
    ylabel = {$f(m)$},
]

\addplot [line width=2pt] [
    domain=1:1.5, 
    samples=100, 
    color=black,
]
{0.5*x^(-1.5)};

\addplot [line width=2pt] [
    domain=1.5:2.5, 
    samples=100, 
    color=green,
]
{0.5*x^(-1.5)};
 
 \addplot [line width=2pt] [
    domain=2.5:3.5, 
    samples=100, 
    color=red,
]
{0.5*x^(-1.5)};

\begin{scriptsize}
\draw [fill=rvwvcq] (1.5,0.272) circle (2.5pt);
\draw[color=rvwvcq] (1.5,0.4) node {$A$};
\draw [fill=rvwvcq] (2,0.1767) circle (2.5pt);
\draw[color=rvwvcq] (2,0.1767+0.11) node {$B$};
\draw [fill=rvwvcq] (2.5,0.126) circle (2.5pt);
\draw[color=rvwvcq] (2.5,0.25) node {$C$};
\end{scriptsize}

 \end{axis}
\end{tikzpicture}
\caption{\small{$2$-separating menu. Constant elasticity of $-1.5$.}}
\end{subfigure}\hfil
\begin{subfigure}{0.25\textwidth}
\begin{tikzpicture}[scale=0.75]
\begin{axis}[
    axis lines = left,
    xmin=1,
xmax=3,
ymin=0,
ymax=1,
    xlabel = $m$,
    ylabel = {$f(m)$},
] 

\addplot [line width=2pt] [
    domain=1:2, 
    samples=100, 
    color=black,
]
{0.5*x^(-1.5)};

\addplot [line width=2pt] [
    domain=2:3.5, 
    samples=100, 
    color=orange,
]
{0.5*x^(-1.5)};
 
\begin{scriptsize}
\draw [fill=rvwvcq] (2,0.176) circle (2.5pt);
\draw[color=rvwvcq] (2,0.3) node {$B$};
\draw [fill=rvwvcq] (2.5,0.126) circle (2.5pt);
\draw[color=rvwvcq] (2.5,0.25) node {$C$};
\end{scriptsize}

\end{axis}
\end{tikzpicture}
\caption{\small{$1$-separating menu. Constant elasticity of $-1.5$.}}
\end{subfigure}

\medskip
\begin{subfigure}{0.25\textwidth}
\begin{tikzpicture}[scale=0.75]
\begin{axis}[
    axis lines = left,
     xmin=1,
xmax=3,
ymin=0,
ymax=1,
    xlabel = $m$,
    ylabel = {$f(m)$},
]

\addplot [line width=2pt] [
    domain=1.3:1.5, 
    samples=100, 
    color=black,
]
{3*x^(-4)};

\addplot [line width=2pt] [
    domain=1.5:2.5, 
    samples=100, 
    color=green,
]
{3*x^(-4)};
 
 \addplot [line width=2pt] [
    domain=2.5:3.5, 
    samples=100, 
    color=red,
]
{3*x^(-4)};

\begin{scriptsize}
\draw [fill=rvwvcq] (1.5,0.592) circle (2.5pt);
\draw[color=rvwvcq] (1.57,0.72) node {$A$};
\draw [fill=rvwvcq] (2,0.185) circle (2.5pt);
\draw[color=rvwvcq] (2,0.185+0.11) node {$B$};
\draw [fill=rvwvcq] (2.5,0.0768) circle (2.5pt);
\draw[color=rvwvcq] (2.5,0.2) node {$C$};
\end{scriptsize}

 \end{axis}
\end{tikzpicture}
\caption{\small{$2$-separating menu. Constant elasticity of $-4$.}}
\end{subfigure}\hfil
\begin{subfigure}{0.25\textwidth}
\begin{tikzpicture}[scale=0.75]
\begin{axis}[
    axis lines = left,
     xmin=1,
xmax=3,
ymin=0,
ymax=1,
    xlabel = $m$,
    ylabel = {$f(m)$},
] 

\addplot [line width=2pt] [
    domain=1.3:2, 
    samples=100, 
    color=black,
]
{3*x^(-4)};

\addplot [line width=2pt] [
    domain=2:3.5, 
    samples=100, 
    color=orange,
]
{3*x^(-4)};

\begin{scriptsize}
\draw [fill=rvwvcq] (2,0.1875) circle (2.5pt);
\draw[color=rvwvcq] (2,0.3) node {$B$};
\draw [fill=rvwvcq] (2.5,0.0768) circle (2.5pt);
\draw[color=rvwvcq] (2.5,0.2) node {$C$};
\end{scriptsize}

\end{axis}
\end{tikzpicture}
\caption{\small{$1$-separating menu. Constant elasticity of $-4$.}}
\end{subfigure}\hfil

\caption{\small The black color represents the buyers that choose not to  participate in the platform, the green color represents the buyers that choose $L$, and the red color represents the buyers that choose $H$ in the $2$-separating menu  and the orange color represents the buyers that choose to buy the product in the $1$-separating menu.} 
\label{Fig: 1}
\end{figure}



\noindent {\bf Main result}.  We can now state our main result using the previous two conditions.  The following theorem  states that our constrained price discrimination problem admits an optimal solution that is $1$-separating. All the proofs in the paper are deferred to the Appendix. 

\begin{theorem}
\label{Theorem: Main}
Suppose that $F(m)m$ is a strictly\footnote{The assumption that $F(m)m$ is {\em strictly} convex implies that the monopoly price is unique. This assumption is for mathematical convenience and does not influence the result.}  convex function on $[a ,b]$ and that $\mathcal{C}$ is regular. Assume that the set of all $1$-separating menus $\mathcal{C}_{1}$ is  a compact subset of $\mathbb{R}^2$.\footnote{In the two-sided market models that we study the constraint set is finite, and hence, $\mathcal{C}_{1}$ is compact.} Then there is an optimal $1$-separating menu.  In addition, the optimal $1$-separating menu $\{(p,q)\}$ is {\em maximal} in $\mathcal{C}_{1}$: for every $ \{(p',q') \} \in \mathcal{C}_{1}$ such that $(p',q') \neq (p,q)$ we have $p > p'$ or $q > q'$. 

\end{theorem}
In the Appendix we also show that we can slightly weaken the regularity condition. 

We note that if for every menu $C =\left \{\left (p_{1} ,q_{1}\right ) ,\ldots  ,\left (p_{k} ,q_{k}\right )\right \} $ that belongs to $\mathcal{C}$, the $1$-separating menu $C^{ \prime } =\{p_{k} ,q_{k}\}$ belongs to $\mathcal{C}$ then the second condition in Definition \ref{Def: marketN} is not needed in order to prove the optimality of a $1$-separating menu. The proof of this follows immediately from the proof of Theorem \ref{Theorem: Main}. The intuition for this result follows from the argument in Section \ref{Sec: toy} that shows that the second stage of the analysis of the example provided there is not needed when such menu $C'$ belongs to $\mathcal{C}$. As we discussed in Section \ref{Sec: toy}, this is useful for the two-sided market model where sellers compete in a Bertrand competition (see Section \ref{Sec: Two-sided2 prices}). We use the next Corollary to prove the optimality of a $1$-separating menu in that model.

\begin{corollary} \label{Corr: Bert}
Suppose that $F(m)m$ is a convex function on $[a ,b]$ and that for every menu $C =\left \{\left (p_{1} ,q_{1}\right ) ,\ldots  ,\left (p_{k} ,q_{k}\right )\right \} \in  \mathcal{C}$ we have $C^{ \prime } =\{p_{k} ,q_{k}\} \in \mathcal{C}$. Assume that the set of all $1$-separating menus $\mathcal{C}_{1} \in \mathcal{C}$ is compact. Then there is an optimal $1$-separating menu. 
\end{corollary}

Corollary \ref{Corr: Bert} can be applied for some important constraint sets as the following example shows.  

\begin{example} \label{Example: 1-rich}
(i) In this example, the platform can choose any subset of price-quality pairs from a pre-fixed set of price-quality pairs.  Suppose that there is a given set $\mathcal{P}$ of $R$ price-quality pairs, $\mathcal{P} =\left \{\left (p_{1} ,q_{1}\right ) ,\ldots  ,\left (p_{R} ,q_{R}\right )\right \}$. Then the constraint set is  $\mathcal{C}_{\mathcal{P}} =2^{\mathcal{P}}$ where $2^{\mathcal{X}}$ is the set of all subsets of a set $\mathcal{X}$.

(ii) In this example, the platform can choose any finite string $(p_{1} ,q_{1} ,\ldots  ,p_{k} ,q_{k})$ in $\mathbb{R}^{2k}$ for $k \leq N$ where $N \geq 1$, $p_{i} \in [0 ,\overline{p}]$ and $q_{i} \in [0 ,\overline{q}]$ for all $1 \leq i \leq k$. That is, the constraint set is given by 

 $\mathcal{C}_{N} =\{C :C$ is a $k$-separating menu for $k \leq N$ such that $(p ,q) \in [0 ,\overline{p}] \times [0 ,\overline{q}]$ for all $(p ,q) \in C\}$.
\end{example}

In the two-sided market model in Section \ref{Sec: Two-sided2 prices}, the constraint set that the platform faces is the same as the constraint set in Example \ref{Example: 1-rich} part (i) (see Theorem \ref{thm: Bertrand}).  
 The constraint set in Example \ref{Example: 1-rich} part (ii) is standard in the price discrimination literature (see for example \cite{bergemann2011mechanism}).

\subsection{Local Results} \label{Section: local}

In practice, because of operational considerations or other constraints, a platform might only consider a small number of options.  For example, an e-commerce platform can introduce a new top rated sellers category or remove an existing category. In this section we show that our main result holds also locally. That is, the values of the density function's  elasticity on some local region remain the key condition when deciding which option will yield more total transaction value.

For simplicity, suppose that the platform considers only two menus $C = \{(p_1,q_1),\ldots,(p_{n},q_{n}) \} \in \mathcal{C}_{p}$ and $C' = C  \setminus \{(p_1,q_1)\}$ where $p_{i} < p_{j}$, $q_{i} < q_{j}$ if $i<j$. (Below, we provide a more general version of the local result.)
In our two sided-market model where sellers choose prices, the menu $C'$ is feasible and can be obtained from the menu $C$ by banning some low quality sellers (see Section \ref{Sec: Two-sided2 prices}).  Intuitively, when the market is competitive, removing  some sellers should not have a major impact on the qualities and the prices of the other sellers. The platform does not seek to find the optimal menu across all menus but only to determine which menu yields more total transaction value: $C$ or $C'$. In Proposition \ref{Prop: localA} we show that the menu $C$ yields lower (resp., higher) total transaction value than the menu $C'$ if the density function's elasticity is bounded below (resp., above) by $-2$ on the interval $A := \left [p_{1}/q_{1},(p_{2}-p_{1})/(q_{2}-q_{1}) \right ]$. 


 
We can obtain some intuition for this result as follows.  A type $m$ buyer chooses the price-quality pair $(p_{1},q_{1})$ under the menu $C$ if and only if $m \in A$. Thus, in order to compare $C$ and $C'$, the density function's elasticity must be bounded below or above $-2$ on the set of buyers' types that choose the price-quality pair $(p_{1},q_{1})$. Further,  the elasticity of many standard density functions is decreasing. In such a case, we can check the density function's elasticity at just one point (which can be typically done by price experimentation) to determine which menu yields more total transaction value: $C$ or $C'$. Thus, the local results in this section can be used to guide the market design decision of whether to introduce a new category of sellers, or to remove an existing category.  

We actually prove a more general version of the result discussed above. We compare any two menus $C$ and $C'$ such that $C' \in 2^{C}$ where $2^{C}$ is the power set of $C$. In the two-sided market model the menu $C'$ can be obtained by removing some sellers from the platform (not necessarily the lowest quality sellers). We show that $C'$ yields more (resp., less) total transaction value than $C$ under convexity (resp., concavity) of $F(m)m$ on a certain relevant local region.

\begin{definition}
For a menu $C=\{(p_{1},q_{1}),\ldots,(p_{n},q_{n})\} \in \mathcal{C}_{p}$, we define $m_{i}(C) =\left (p_{i} -p_{i-1}\right )/\left (q_{i} -q_{i-1}\right )$ for $i =1,\ldots ,n$ where $p_{0} = q_{0}=0$.
\end{definition}
 
\begin{proposition} \label{Prop: localA}
 Let $C = \{(p_1,q_1),\ldots,(p_{n},q_{n}) \} \in \mathcal{C}_{p}$ and let  $C'=\{(p_{\mu(1)} , q_{\mu(1)}),\ldots,(p_{\mu(k)} , q_{\mu(k)} )\}   \in 2^{C}$ so $\mu(i) \in \{1,\ldots,n\}$ for all $i$. Assume without loss of generality that $p_{i} < p_{j}$ and $\mu(i) < \mu(j)$ whenever $i<j$. Define $ \mu (0) = 0$. 
 
Assume that $\mu(k) = n$.\footnote{We show in the proof of Theorem \ref{Theorem: Main} (see Step 3) that if $\mu(k)<n$ then $C'$ is not optimal, so we only consider the case that $C'$ is such that $\mu(k)=n$. The intuition for this fact is that removing the highest quality sellers can only decrease the revenues.} 
 Then, $\pi(C) \leq \pi(C') $ if $F(m)m$ is convex on $ [m_{\mu(j-1)+1}(C),m_{\mu(j)}(C) ]$ for all $j$ such that $\mu(j) - \mu (j-1) > 1$. Further, $\pi(C) \geq \pi(C') $ if $F(m)m$ is concave on $ [m_{\mu(j-1)+1}(C),m_{\mu(j)}(C) ]$ for all $j$ such that $\mu(j) - \mu (j-1) > 1$. 
 \end{proposition} 
 
 Proposition \ref{Prop: localA} applies directly to the two-sided market model where sellers choose prices we introduce in Section \ref{Sec: Buyers prices}. In this model the market is competitive as sellers compete in Bertrand competition, and hence, equilibrium prices of high quality sellers do not change when low quality sellers are removed from the platform. For less competitive markets, we may expect that equilibrium prices of high quality sellers increase when low quality sellers are removed from the platform from the same logic discussed in Section 2: some demand shifts from the banned low quality sellers to high quality sellers (see also the discussion in the end of Section \ref{Sec: toy}).   A similar ``local'' analysis to the one in Proposition \ref{Prop: localA} can be applied to this case also, but this requires an 
 additional condition on the menus under consideration that  can be seen as a local version of the regularity condition (see Definition \ref{Def: marketN} and the discussion after that definition). To simplify the conditions on the constraint set we analyze the case where the platform considers to ban low quality sellers. In Section \ref{Sec: Model 1 quant}, we use this  analysis to provide local results  for the model where sellers choose quantities (see  Proposition \ref{prop:localModel1}). 
 
 \begin{proposition} \label{Prop: localB}
 Let $C = \{(p_1,q_1),\ldots,(p_{n},q_{n}) \} \in \mathcal{C}_{p}$ and let  $C'=\{(p_{n}' , q_{n} )\} \in \mathcal{C}_{p}$ be a $1$-separating menu where $p_{n}' \geq p_{n}$.  Assume that $m(1-F(m))$ is strictly quasi-concave on $[a,b]$ and that $p^{M}(1) \in (a,b)$.\footnote{Recall that $p^{M}(1)$ is the monopoly price when the quality is $1$ (see Section 3.2). The strict quasi-concavity of $m(1-F(m))$  implies that $p^{M}(1)$ is uniquely defined.} We have $\pi(C) \leq \pi(C')$ if  $ m_{n}(C') \leq p^{M} (1)$ and $F(m)m$ is convex on $[m_{1}(C),m_{n}(C)]$ and $\pi(C) \geq \pi(C')$ if  $m_{n}(C) \geq p^{M} (1)$ and $F(m)m$ is concave on $[m_{1}(C),m_{n}(C)]$. 
 \end{proposition}


We end this section with a corollary that shows that it is enough to assume that the function $F(m)m$ is convex on a subset of $[a,b]$ in order to prove that there exists a $1$-separating menu that yields more total transaction value than any other menu $C$. The proof of Corollary \ref{Coro: convex on int} follows immediately from the proof of Theorem \ref{Theorem: Main}. 

\begin{corollary} \label{Coro: convex on int}
Let $C =\{(p_{1},q_{1}) ,\ldots ,(p_{k} ,q_{k})\} \in \mathcal{C}_{p}$ be a $k$-separating menu where $p_{i} <p_{j}$ if $i <j$. Suppose that  $F(m)m$ is convex on\footnote{Note that $C \in \mathcal{C}_{p}$ implies $m_{i}(C) < m_{j}(C)$ for $i<j$ and that $[m_{1}(C) ,m_{k}(C)] \subseteq [a ,b]$ (see the proof of Theorem \ref{Theorem: Main}).} $[m_{1}(C) ,m_{k}(C)]$ and that $\mathcal{C}$ is regular. Then there exists a $1$-separating menu $C^{ \ast }$ that yields more revenue than $C$, i.e., $\pi (C) \leq \pi (C^{ \ast })$.

In addition, if $F(m)m$ is convex on $[m_{1}(C) ,m_{k}(C)]$ for every menu and $\mathcal{C}_{1}$ is compact, then there is a $1$-separating menu that maximizes the total transaction value.
\end{corollary}

\subsection{Necessity of the Conditions For the Optimality of 1-separating Menus} \label{Sec:Necessity}

Theorem \ref{Theorem: Main} provides sufficient conditions for the optimality of 1-separating menus. In this section, we show that these conditions are necessary for the optimality of 1-separating menus in the sense that if one of the conditions is violated we can find a constrained price discrimination problem where $1$-separating menus are not optimal. These constrained price discrimination problems arise in the two-sided market models we study in Sections \ref{Sec: Model 1 quant} and \ref{Sec: Two-sided2 prices} for specific reasonable models' parameters (e.g., sellers' costs).    
We first show that when the function $F(m)m$ is not convex, we can always find a a regular constraint set  $\mathcal{C}$ such that no $1$-separating menu exists that maximizes the total transaction value. In particular, we can find a simple regular constraint set $\mathcal{C} =2^{C}$ where  $C =\left \{\left (p_{1} ,q_{1}\right ) ,\left (p_{2} ,q_{2}\right )\right \}$ (see Example \ref{Example: 1-rich} part (i)), for which the $2$-separating menu is optimal. In a similar manner we can provide more complicated constraint sets where a $k$-separating menu is optimal for $k \geq 2$ when $F(m)m$ is not convex.   

\begin{proposition} \label{Prop: Converse}
 Suppose that $F(m)m$ is not convex on $(a ,b)$. Then for any $k \geq 2$, there exists a menu
 $C =\left \{\left (p_{1} ,q_{1}\right ) ,\ldots, \left (p_{k} ,q_{k}\right )\right \}$  and a regular constraint set $\mathcal{C}$ such that the menu $C \in \mathcal{C}$ maximizes the  revenues and yields strictly more revenue than any $1$-separating menu in $\mathcal{C}$.  
\end{proposition}

When $F(m)m$ is not convex on $(a,b)$, Proposition \ref{Prop: Converse} shows that  we  can construct a constraint set where a $2$-separating menu yields more total transaction value than any  $1$-separating menu. In terms of the two-sided markets we study this result means  that there is a range of  parameters (e.g., different sellers' costs or qualities) in which a $1$-separating menu is not optimal when $F(m)m$ is not convex on $(a,b)$. Hence, the convexity of $F(m)m$ on $(a,b)$ is a necessary condition for the optimality of $1$-separating menus. The intuition for this result is similar to the intuition provided in Sections \ref{Sec: toy} and 3.2. In the case that the function $mF(m)$ is not convex, $1$-separating menus are less attractive as they cause the platform to lose a relatively significant amount of participants when moving from a $2$-separating menu to a $1$-separating menu (see Figure \ref{Fig: 1}). 

When $F(m)m$ is convex on $(a,b)$ we can find a constraint set $\mathcal{C}$ that satisfies condition (i)  but not condition (ii) or satisfies condition (ii) but not condition (i) in the definition of regularity (see Definition \ref{Def: marketN}) such that no $1$-separating menu exists that maximizes the total transaction value. Hence, the regularity conditions are also necessary for the optimality of $1$-separating menus in the sense  that if one of these 
 conditions do not hold then a $1$-separating menu is not necessarily optimal.

\begin{proposition} \label{prop:converse-regular}
 Suppose that $F(m)m$ is convex on $(a,b)$ and $\mathcal{C}_{1}$ is non-empty. 
 
1. There exists a constraint set $\mathcal{C}$ that satisfies condition (i)  in Definition \ref{Def: marketN} such that there is a menu $C \in \mathcal{C}$  that is not $1$-separating and yields strictly more revenue than any $1$-separating menu in $\mathcal{C}$.

2. There exists a constraint set $\mathcal{C}'$ that satisfies condition (ii)  in Definition \ref{Def: marketN} such that there is a menu $C' \in \mathcal{C}'$ that is not $1$-separating and yields strictly more revenue than any $1$-separating menu in $\mathcal{C}$. 

\end{proposition}

The proof of the last proposition consists of simple  examples and the following  observations. For part (1), if condition (ii) of Definition  \ref{Def: marketN} does not hold, $1$-separating menus can be arbitrarily bad as these menus can have very high prices compared to the optimal ones. For part (ii), if condition (i) of Definition  \ref{Def: marketN} does not hold, the $1$-separating menus can have lower prices compared to other menus, and hence, lower revenues. These violations of regularity are not reasonable in typical two-sided market models as we explain after Definition \ref{Def: marketN} and show in both of the two-sided market models we study. 


\subsection{Convex-concave distributions} \label{sec:convex-concave}

In this section we expand the class of customer type distributions for which we provide structural results by analyzing convex-concave distributions. That is, we study distributions for which $F(m)m$ is convex on some interval $[a,m^{*}]$ and concave on $[m^{*},b]$ for some $m^{*} \in [a,b]$. This condition is satisfied for many distributions of interest, in particular, it typically holds for distributions from the exponential family, such as the exponential distribution,  chi-squared distribution, and log-normal distribution. For example, the exponential distribution with a parameter $\lambda$ has a density function $\lambda \exp (-\lambda x)$ and the density function's elasticity is above $-2$ on $[0,2/\lambda]$ and below $-2$ on $[2/ \lambda , \infty)$. In this example, the density function's elasticity is decreasing which captures settings where the total mass of buyers that are willing to buy the product is decreasing with the product's quality.  As opposed to the convex case, the convex-concave case allows the demand to have a thin-tailed distribution as the decrease in the mass of buyers when the quality increases can be arbitrarily high. 

The following proposition shows that we can characterize the size of the optimal menu as a function of the point where the density function's elasticity crosses $-2$ for the important constraint set $\mathcal{C} = 2^{C}$ that directly applies to the two-sided market model where sellers choose prices.\footnote{Using Proposition \ref{Prop: localB}, a similar analysis can be provided for the case where the prices of high quality sellers increase when the platform removes low quality sellers under an additional condition on the constraint set that resembles the regularity condition. Such a result can be applied for the model where sellers choose quantities we study in Section \ref{Sec: Model 1 quant}. We omit the details for the sake of brevity.} 

\begin{proposition}  \label{Corr:convex-concave}
 Let $C = \{(p_1,q_1),\ldots,(p_{n},q_{n}) \} \in \mathcal{C}_{p}$ and  $\mathcal{C} = 2^{C}$. Assume that $F(m)m$ is strictly convex on  $[a,m^{*}]$ and strictly concave on $[m^{*},b]$ for some $m^{*} \in [a,b]$. Let $1 \leq k \leq n-1$ be an integer. Then the optimal menu is $k$-separating or $k+1$-separating when $m^{*} \in [m_{n-k}(C),m_{n-k+1}(C)]$. 
\end{proposition}

Proposition \ref{Corr:convex-concave} expands the set of distributions  for which we can apply our techniques and shows when a $k$-separating menus can be optimal for a general $k$. Intuitively, when $m^{*}$ is higher, i.e., $F(m)m$ is convex on a larger interval, smaller menus are optimal, and when $m^{*}$ is lower, i.e., $F(m)m$ is concave on a larger interval, bigger menus are optimal. 
\section{Information Structures \label{Sec: INFORMATION}}

Having described our constrained price discrimination problem, we are now in a position to describe how we apply that framework to design information disclosure policies in two-sided markets.  We begin in this section by describing the information the platform has about the sellers' quality levels and the set of information structures from which the platform can choose.

\noindent {\bf Seller quality}.  Let $X$ be the set of possible sellers' quality levels.  We assume that $X$ is the  interval\protect\footnote{All our results can be easily generalized for the case that $X$ is any compact set in $\mathbb{R}^{n}_{+}$.} $[0 ,\overline{x}]$ for some $\overline{x} > 0$. We denote by $\mathcal{B}(X)$ the Borel sigma-algebra on $X$ and by $\mathcal{P}(X)$ the space of all Borel probability measures on $X$.  The distribution of the sellers' quality levels is described by a probability measure $\phi \in \mathcal{P}(X)$.

\noindent {\bf Platform's information.}  The platform's information is summarized by a finite (measurable) partition $I_{o} =\{A_{1}\ldots  ,A_{l}\}$ of $X$. We assume that $\phi (A_{i}) > 0$ for all $A_{i} \in I_{o}$. The platform has no information about the sellers' quality levels if $\vert I_{o}\vert =1$ where $\vert I_{o}\vert$ is the number of elements in the partition $I_{o}$. 

\noindent {\bf Information structures.}  Given the platform's information $I_{o}$, the platform chooses an information structure to share with buyers. We now define an information structure.
\begin{definition}
 An \emph{information structure} $I$ is a family of disjoint sets such that every set in $I$ is a union of sets in $I_{o}$, i.e., $B \in I$ implies $ \cup _{i}A_{i} =B$ for some sets $A_{i} \in I_{o}$.  
 \end{definition}
While the class of information structures we study is relatively simple, it provides enough richness for our analysis. An interesting direction for future work is to expand  our analysis to other information structures.
We now provide examples of information structures.\footnote{Note that equilibrium conditions will be required to fully specify buyers' beliefs on seller quality within each element of the information structure.}

\begin{example}\label{ex: information structure}
Suppose that $X =[0,1]$, $I_{o} =\{A_{1} ,A_{2} ,A_{3} ,A_{4}\}$, $A_{j} =[0.25(j -1) ,0.25j)$, $j =1 ,\ldots ,4$.

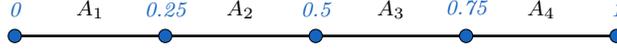
\begin{center}
    \definecolor{rvwvcq}{rgb}{0.08235294117647059,0.396078431372549,0.7529411764705882}
\begin{tikzpicture}[line cap=round,line join=round,>=triangle 45,x=1cm,y=1cm]
\clip(-11,-4.5) rectangle (-1.5,-3.2);
\draw [line width=1pt,color=black] (-9.88,-4.15)-- (-7.88,-4.15);
\draw [line width=1pt,color=black] (-7.88,-4.15)-- (-5.88,-4.15);
\draw [line width=1pt,color=black] (-5.88,-4.15)-- (-3.88,-4.15);
\draw [line width=1pt,color=black] (-3.88,-4.15)-- (-1.88,-4.15);
\begin{scriptsize}
\draw [fill=rvwvcq] (-9.88,-4.15) circle (2.5pt);
\draw[color=rvwvcq] (-9.88,-3.8) node {0};
\draw [fill=rvwvcq] (-7.88,-4.15) circle (2.5pt);
\draw[color=rvwvcq] (-7.88,-3.8) node {0.25};
\draw[color=black] (-8.88,-3.8) node {$A_{1}$};
\draw [fill=rvwvcq] (-5.88,-4.15) circle (2.5pt);
\draw[color=rvwvcq] (-5.88,-3.8) node {0.5};
\draw[color=black] (-6.88,-3.8) node {$A_{2}$};
\draw [fill=rvwvcq] (-3.88,-4.15) circle (2.5pt);
\draw[color=rvwvcq] (-3.88,-3.8) node {0.75};
\draw[color=black] (-4.88,-3.8) node {$A_{3}$};
\draw [fill=rvwvcq] (-1.88,-4.15) circle (2.5pt);
\draw[color=rvwvcq] (-1.88,-3.8) node {1};
\draw[color=black] (-2.88,-3.8) node {$A_{4}$};
\end{scriptsize}
\end{tikzpicture}
\end{center}

Two examples of information structures are the information structure $I_{1} =\{A_{3} ,A_{4}\}$

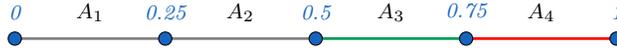
\begin{center}
    \definecolor{rvwvcq}{rgb}{0.08235294117647059,0.396078431372549,0.7529411764705882}
\begin{tikzpicture}[line cap=round,line join=round,>=triangle 45,x=1cm,y=1cm]
\clip(-11,-4.5) rectangle (-1.5,-3.2);
\draw [line width=1pt,color=gray] (-9.88,-4.15)-- (-7.88,-4.15);
\draw [line width=1pt,color=gray] (-7.88,-4.15)-- (-5.88,-4.15);
\draw [line width=1pt,color=ForestGreen] (-5.88,-4.15)-- (-3.88,-4.15);
\draw [line width=1pt,color=red] (-3.88,-4.15)-- (-1.88,-4.15);
\begin{scriptsize}
\draw [fill=rvwvcq] (-9.88,-4.15) circle (2.5pt);
\draw[color=rvwvcq] (-9.88,-3.8) node {0};
\draw [fill=rvwvcq] (-7.88,-4.15) circle (2.5pt);
\draw[color=rvwvcq] (-7.88,-3.8) node {0.25};
\draw[color=black] (-8.88,-3.8) node {$A_{1}$};
\draw [fill=rvwvcq] (-5.88,-4.15) circle (2.5pt);
\draw[color=rvwvcq] (-5.88,-3.8) node {0.5};
\draw[color=black] (-6.88,-3.8) node {$A_{2}$};
\draw [fill=rvwvcq] (-3.88,-4.15) circle (2.5pt);
\draw[color=rvwvcq] (-3.88,-3.8) node {0.75};
\draw[color=black] (-4.88,-3.8) node {$A_{3}$};
\draw [fill=rvwvcq] (-1.88,-4.15) circle (2.5pt);
\draw[color=rvwvcq] (-1.88,-3.8) node {1};
\draw[color=black] (-2.88,-3.8) node {$A_{4}$};
\end{scriptsize}
\end{tikzpicture}

\end{center}
and the information structure $I_{2} =\{A_{3} \cup A_{4}\}$

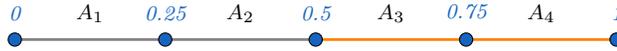
\begin{center}
\definecolor{rvwvcq}{rgb}{0.08235294117647059,0.396078431372549,0.7529411764705882}
\begin{tikzpicture}[line cap=round,line join=round,>=triangle 45,x=1cm,y=1cm]
\clip(-11,-4.5) rectangle (-1.5,-3.2);
\draw [line width=1pt,color=gray] (-9.88,-4.15)-- (-7.88,-4.15);
\draw [line width=1pt,color=gray] (-7.88,-4.15)-- (-5.88,-4.15);
\draw [line width=1pt,color=orange] (-5.88,-4.15)-- (-3.88,-4.15);
\draw [line width=1pt,color=orange] (-3.88,-4.15)-- (-1.88,-4.15);
\begin{scriptsize}
\draw [fill=rvwvcq] (-9.88,-4.15) circle (2.5pt);
\draw[color=rvwvcq] (-9.88,-3.8) node {0};
\draw [fill=rvwvcq] (-7.88,-4.15) circle (2.5pt);
\draw[color=rvwvcq] (-7.88,-3.8) node {0.25};
\draw[color=black] (-8.88,-3.8) node {$A_{1}$};
\draw [fill=rvwvcq] (-5.88,-4.15) circle (2.5pt);
\draw[color=rvwvcq] (-5.88,-3.8) node {0.5};
\draw[color=black] (-6.88,-3.8) node {$A_{2}$};
\draw [fill=rvwvcq] (-3.88,-4.15) circle (2.5pt);
\draw[color=rvwvcq] (-3.88,-3.8) node {0.75};
\draw[color=black] (-4.88,-3.8) node {$A_{3}$};
\draw [fill=rvwvcq] (-1.88,-4.15) circle (2.5pt);
\draw[color=rvwvcq] (-1.88,-3.8) node {1};
\draw[color=black] (-2.88,-3.8) node {$A_{4}$};
\end{scriptsize}
\end{tikzpicture}

\end{center}

In the information structure $I_{1}$, the sellers whose quality levels belong to the sets $A_{1}$ and $A_{2}$ are ``banned'' from the platform, and the sellers whose quality levels belong to the sets $A_{3}$ and $A_{4}$ can participate in the platform. The platform shares the information it has about the sellers whose quality levels belong to the sets $A_{3}$ and $A_{4}$, i.e., the buyers know that the quality level of a seller in the set $A_{4}$ is between $0.75$ and $1$, and the quality level of a seller in the set $A_{3}$ is between $0.5$ and $0.75$. In the information structure $I_{2}$, the sellers whose quality levels belong to the sets $A_{1}$ and $A_{2}$ are again banned from the platform and the platform does not share the information it has about the other sellers. Hence, buyers cannot distinguish between sellers in $A_3$ and $A_4$.    
\end{example}

Note that the platform's information structure $I=\{B_{1},\ldots,B_{n}\}$ determines both which sellers are banned from the platform (in particular, sellers in $X \setminus \cup_{B_{i} \in I} B_{i}$ are banned from the platform), as well as the amount of information that the platform shares with buyers regarding the sellers that participate in the platform.  

Given an information structure $I$, we define the measure space $\Omega _{I} =(X ,\sigma (I))$ where $\sigma (I)$ is the sigma-algebra generated by $I$. Recall that a function $p :(X ,\sigma (I)) \rightarrow \mathbb{R}$ is $\sigma (I)$ measurable if and only if $p$ is constant on each element of $I$, i.e., $x_{1} ,x_{2} \in B$ and $B \in I$ imply that $p(x_{1}) =p(x_{2}) : =p(B)$. 

Given the platform's initial information on the sellers' quality levels  $I_{o}$, we denote by $\mathbb{I}(I_{o})$ the set of all possible information structures.

\noindent {\bf $k$-separating information structures.}  We say that an information structure $I$ is {\em $k$-separating} if $I$ contains exactly $k$ elements, i.e., $\vert I\vert  =k$. For example, the information structure $I_{1}$ described in Example \ref{ex: information structure} is $2$-separating and the information structure $I_{2}$ is $1$-separating.

\subsection{Remarks On The Assumptions}

We now provide a few remarks on our assumptions.

\noindent \textbf{Exogenous quality.} In our two-sided market models we assume that sellers choose quantities or prices while their qualities are their types. In some platforms sellers can choose or improve their quality. In those cases, the sellers' types can be their opportunity cost, investment cost, or another feature. In principle, we could incorporate this into our model and the transformation to a constrained price discrimination problem would still hold. However, the set of feasible menus (equilibrium menus) is
determined by the specific two-sided market model we study and by the market arrangement. Hence, the set of feasible menus would be different and harder to characterize when sellers can also choose their quality.

\noindent \textbf{The platform's initial information.} As we discussed in the introduction, platforms collect information about the sellers' quality from many sources. In this paper we abstract away from the data collection process and assume that the platform has already collected some information about the sellers' quality and classified the sellers' quality (the partition $I_{o}$ represents this classification). We focus on how much of this information the platform should share with buyers to maximize its revenues. An interesting future research direction is to incorporate dynamic considerations that are related to learning, such as learning the sellers' quality over time, into our framework. 

\noindent \textbf{Constrained information structures.} The information structures available to the platform in our model are more limited than the information structures available to the platform (sender) in the standard information design literature. For example, we do not allow the platform to use a mixed strategy (i.e., mix over sets in the platform's initial information $I_{o}$). Allowing for mixed strategies would actually simplify our analysis as is typically the case in the information design literature. However, in our context of quality selection we think that platform's pure strategies are more realistic.  Also, because we assume that the platform's information about the sellers' quality is partial and is given by a finite partition, every information structure that the platform can choose as well as the set of possible information structures that the platform can choose from are finite. The analysis of the constrained price discrimination problem in Section \ref{Sec: price-dis} shows that our framework can be generalized to the case of uncountable information structures.     

\section{Two-Sided Market Model 1: Sellers Choose Quantities} \label{Sec: Model 1 quant}

In this section we consider a model in which the platform chooses the prices, and the sellers choose the quantities.

The platform chooses an information structure $I \in \mathbb{I}(I_{o})$ and a $\sigma (I)$ measurable pricing function $p$. The measurability of the pricing function means that if the platform does not reveal any information about the quality of two sellers, i.e., the two sellers belong to the same set $B$ in the information structure $I$, then these sellers are given the same price under the platform's pricing function. The  measurability condition is natural because the buyers do not have any information on the sellers' quality except the information provided by the platform, so any rational buyer will not buy from a seller $x$ whose price is higher than a seller $y$ when $x$ and $y$ have the same expected quality. 

With slight abuse of notation, for an information structure $I =\{B_{1} ,\ldots  ,B_{n}\}$, we denote a $\sigma (I)$ measurable pricing function by $\boldsymbol{p}=(p(B_{1}),\ldots,p(B_{n}))$ where $p(B_{i})$ is the price that every seller $x$ in $B_{i}$ charges. A pricing function  $\boldsymbol{p}=(p(B_{1}),\ldots,p(B_{n}))$ is said to be {\em positive} if $p(B_{i})>0$ for all $B_{i} \in I$.

An information structure $I =\{B_{1} ,\ldots  ,B_{n}\}$ and a pricing function $\boldsymbol{p}$ generate a game between the sellers and the buyers. The platform's decisions and the structure of the game are common knowledge at the start of the game. In the game, the sellers choose quantities,\footnote{Here quantities can correspond, for example, to how many hours the sellers choose to work.} and the buyers choose whether to buy a product and if so, from which set of sellers $B_{i} \in I$ to buy it. Each equilibrium of the game induces a certain revenue for the platform. The platform's goal is to choose an information structure and prices that maximize the platform's equilibrium revenue. We now describe the buyers' and sellers' decisions in detail.

\subsection{\label{Sec: Buyers-quantity}Buyers}
Buyers are heterogeneous in how much they value the quality of the product relative to its price; in particular, every buyer has a {\em type} in $[a,b] \subseteq \mathbb{R}_{ +} : =[0 ,\infty )$, with buyers' types distributed according to the probability distribution function $F$ on $[a,b]$, with continuous probability density function $f$.  The buyers do not know the sellers' quality levels, but they know the information structure $I =\{B_{1} ,\ldots  ,B_{n}\}$ and the pricing function $\boldsymbol{p}$ that the platform has chosen. 

The buyers choose whether to buy a single product and if so, from which set of sellers $B_{i} \in I$  to buy it. A type $m \in [a ,b]$ buyer's utility from buying a product from a type $x \in B_{i}$ seller is given by
\begin{equation*}
Z(m ,B_{i} ,p(B_{i})) =m\mathbb{E}_{\lambda _{B_{i}}}(X) - p(B_{i}).
\end{equation*} 
The probability measure $\lambda _{B_{i}}$ describes the buyers' beliefs about the quality levels of sellers in the set $B_{i}$, and $\mathbb{E}_{\lambda _{B_{i}}}(X)$ is the seller's expected quality given the buyers' beliefs  $\lambda _{B_{i}}$.\footnote{All of our results hold if a type $m \in [a ,b]$ buyer's utility is given by $Z(m ,B_{i} ,p(B_{i})) =mv(\lambda_{B_{i}}) - p(B_{i})$ for some function $v:\mathcal{P}(X) \rightarrow \mathbb{R}_{+}$ that is increasing with respect to stochastic dominance. For example, the function $v$ can capture buyers' risk aversion.} In equilibrium, the buyers' beliefs are consistent with the sellers' quantity decisions and with Bayesian updating.  

A type $m$ buyer buys a product from a type $x \in B_{i}$ seller if $Z(m ,B_{i} ,p(B_{i})) \geq 0$ and $Z(m ,B_{i} ,p(B_{i})) =\max _{B \in I}Z(m ,B ,p(B))$, and does not buy it  otherwise.\footnote{
If there are multiple sets $\{B_{i}\}_{B_{i} \in \overline{P}}$ such that for some type $m$ buyer we have $Z(m ,B_{i} ,p(B_{i})) \geq 0$ and $Z(m ,B_{i} ,p(B_{i})) =\max _{B \in I}Z(m ,B ,p(B))$, then we break ties by assuming that the buyer chooses to buy from the set of sellers with the highest index, i.e., $\max _{i \in \{i :B_{i} \in \bar{P}\}}i$.} The total demand in the market for products sold by types $x \in B_{i}$ sellers given the information structure $I$ and the pricing function $\boldsymbol{p}$,  $D_{I}(B_{i} ,\boldsymbol{p})$ is given by  
\begin{equation*}D_{I}(B_{i} ,\boldsymbol{p}) =\int _{a}^{b}1_{\{Z(m ,B_{i} ,p(B_{i})) \geq 0\}}1_{\{Z(m ,B_{i} ,p(B_{i})) =\max _{B \in I}Z(m ,B ,p(B))\}}F(dm). 
\end{equation*}

\subsection{Sellers} \label{Sec: Sellers-quantity}
Given the information structure $I$ and the pricing function $\boldsymbol{p}$, a type $x \in B_{i} \subseteq X$ seller's utility is given by 
\begin{equation*}U(x ,h ,p(B_{i})) = hp(B_{i}) -\frac{k(x)h^{\alpha +1}}{\alpha+1}.
\end{equation*} 
Each seller chooses a quantity $h \in \mathbb{R}_{+}$ in order to maximize their utility. For a type $x$ seller, the cost of producing $h$ units is given by $k(x)h^{\alpha  +1}/(\alpha + 1)$. The seller's cost function depends on their type and on the quantity that they sell. We assume that $k$ is measurable and is bounded below by a positive number. We also assume that the cost of producing $h$ units is strictly convex in the quantity, i.e., $\alpha  > 0$. This cost structure is quite general and simplifies the characterization of the constraint set, i.e., the set of equilibrium menus (see Proposition \ref{Prop: unique equilibrium} and Lemma \ref{Lem: expected sellers} in the Appendix) but showing that the constraint set is regular can be done under more general cost structures. 

Let $g(x ,p(B_{i})) =\ensuremath{\operatorname*{argmax}}_{h \in \mathbb{R}_{+}}U(x ,h ,p(B_{i}))$ 
be the quantity that a type $x \in B_{i}$ seller chooses when the pricing function is $\boldsymbol{p}=(p(B_{1}),\ldots,p(B_{n}))$. Note that $g$ is single-valued because $U$ is strictly convex in $h$. Let $$S_{I}(B_{i} ,p(B_{i})) =\int _{B_{i}}g(x ,p(B_{i}))\phi (dx)$$ be the total supply in the market of sellers with types $x \in B_{i}$.

\subsection{Equilibrium} \label{Sec: eq quantities}
Given the information structure and the pricing function that the platform chooses, there are four equilibrium requirements. First, the sellers choose quantities in order to maximize their utility. Second, the buyers choose whether to buy a product and if so, from which set of sellers to buy it in order to maximize their own utility. Third, the buyers' beliefs about the sellers' quality are consistent with Bayesian updating and with the sellers' actions. Fourth, demand equals supply for each set $B_{i}$ that belongs to the information structure.
We now define an equilibrium formally.

\begin{definition} \label{Def:EqModel1}
Given an information structure $I =\{B_{1},\ldots  ,B_{n}\}$ and a positive pricing function $\boldsymbol{p}=(p(B_{1}),\ldots,p(B_{n}))$, an equilibrium is given by the buyers' demand $\{D_{I}(B_{i},\boldsymbol{p})\}_{i =1}^{n}$, sellers' supply $\{S_{I}(B_{i},p(B_{i}))\}_{i =1}^{n}$, and buyers' beliefs $\{\lambda _{B_{i}}\}_{i =1}^{n}$ that satisfy the following conditions: 

(i) Sellers' optimality: The sellers' decisions are optimal. That is, 
\begin{equation*}g(x ,p(B_{i})) =\ensuremath{\operatorname*{argmax}}_{h \in \mathbb{R}_+} ~ U(x ,h ,p(B_{i}))
\end{equation*} 
is the optimal quantity for each seller $x \in B_{i} \in I$. 

(ii) Buyers' optimality: The buyers' decisions are optimal.  That is, for each buyer $m \in [a ,b]$ that buys from type $x \in B_{i}$ sellers, we have $Z(m ,B_{i} ,p(B_{i})) \geq 0$ and $Z(m ,B_{i} ,p(B_{i})) =\max _{B \in I}Z(m ,B ,p(B))$. 

(iii) Rational expectations: $\lambda _{B_{i}}(A)$ is the probability that a buyer is matched to sellers whose quality levels belong to the set $A$ given the sellers' optimal decisions, i.e.,  
\begin{equation} \label{Eq: lambda}
\lambda _{B_{i}}(A) =\frac{\int _{A}g(x ,p(B_{i}))\phi (dx)}{\int _{B_{i}}g(x ,p(B_{i}))\phi (dx)}
\end{equation}
for all $B_{i} \in I$ and for all measurable sets $A \subseteq B_{i}$.\protect\footnote{We assume uniform matching within each set $B_i$. Further,
If $\int _{B_{i}}g(x ,p(B_{i}))\phi (dx) =0$ then we define $\lambda _{B_{i}}$ to be the Dirac measure on the point $0 =\min X$.    
} 

(iv) Market clearing: For all $B_{i} \in I$ the total supply equals the total demand, i.e.,   
\begin{equation*} S_{I}(B_{i} ,p(B_{i})) =D_{I}(B_{i} ,\boldsymbol{p}) \ ,
\end{equation*}
 where $D_{I}(B_{i} ,\boldsymbol{p})$ and $S_{I}(B_{i} ,p(B_{i}))$ are defined in Sections \ref{Sec: Buyers-quantity} and \ref{Sec: Sellers-quantity} respectively.  
\end{definition}

 The equilibrium requirements limit the platform's ability to design the market. The buyers' beliefs about the expected sellers' quality depends on the sellers' quantity decisions, which the platform cannot control. Thus, the platform's ability to influence the buyers' beliefs by choosing an information structure is constrained. Furthermore, the prices and the expected sellers' qualities must form an equilibrium (i.e., supply equals demand)\footnote{In Section \ref{Subsec:ImbalancesModel1} we discuss the case that there are supply and demand imbalances.} in each set of the information structure. This equilibrium requirement is in addition to the more standard requirement in the market design literature that the buyers' and sellers' decisions are optimal. Hence, the platform cannot implement every pair of an information structure and pricing function.  This motivates the following definition.
  
 \begin{definition}
An information structure and pricing function pair $(I,\boldsymbol{p})$ is called \emph{implementable} if there exists an equilibrium $(D,S,\lambda )$ under $(I ,\boldsymbol{p})$ where $D =\{D_{I}(B_{i} ,\boldsymbol{p})\}_{B_{i} \in I}$, $S =\{S(B_{i},p(B_{i})\}_{B_{i} \in I}$, and $\lambda  =\{\lambda _{B_{i}}\}_{B_{i} \in I}$. We say that $(D,S,\lambda )$ implements $(I,\boldsymbol{p})$ if $(D ,S ,\lambda )$ is an equilibrium under $(I,\boldsymbol{p})$. 
 \end{definition}

We denote by $\mathcal{W}^{Q}$ the set of all implementable pairs of an information structure and pricing function $(I,\boldsymbol{p})$. In Section \ref{subsec:results-quant} we provide a convex program to find the equilibrium prices given an information structure if these prices exist. The platform's goal is to choose an information structure $I =\{B_{1},\ldots ,B_{n}\}$ and a pricing function $\boldsymbol{p}$ that maximize the total transaction value $\pi ^{Q}$ given by 
\begin{equation*}\pi ^{Q} (I,\boldsymbol{p}) :=\sum _{B_{i} \in I}p(B_{i})\min \{D_{I}(B_{i} ,\boldsymbol{p}) ,S_{I}(B_{i} ,p(B_{i}))\}
\end{equation*}
under the constraint that $(I,\boldsymbol{p})$ is implementable. That is, the platform's revenue maximization problem is given by  $\max_{(I,\boldsymbol{p}) \in \mathcal{W}^{Q}} \pi ^{Q}(I,\boldsymbol{p})$.\footnote{We can easily incorporate into the model commissions $\gamma _{1} ,\gamma _{2}$ on each side of the market. In this case the platform's revenue is given by $\sum _{B_{i} \in I}p(B_{i})\min \{D_{I}(B_{i} ,\boldsymbol{p}) ,S_{I}(B_{i} ,p(B_{i}))\}(\gamma _{1} +\gamma _{2}) $. The commissions may change the demand, supply and equilibrium prices. Nonetheless, for fixed commissions, the platform's revenue maximization problem is equivalent to maximizing the total transaction value. Hence, our analysis still holds for the case where the platform introduces fixed commissions.  }  

\subsection{Equivalence with Constrained Price Discrimination} \label{Subsec: equiv Quantity}
 
The main motivation for studying the constrained price discrimination problem that we analyzed in Section \ref{Sec: price-dis} is that the platform's revenue maximization problem described above transforms into this constrained price discrimination problem. To see this, let $(I,\boldsymbol{p})$ be an information structure-pricing function pair where $I =\{B_{1} ,B_{2} ,\ldots  ,B_{n}\}$ and $\boldsymbol{p}=(p(B_{1}),\ldots,p(B_{n}))$. Let $D =\{D_{I}(B_{i} ,\boldsymbol{p})\}_{B_{i} \in I}$, $S =\{S(B_{i},p(B_{i})\}_{B_{i} \in I}$, and $\lambda  =\{\lambda _{B_{i}}\}_{B_{i} \in I}$ be an equilibrium under $(I,\boldsymbol{p})$. Then $(I,\boldsymbol{p})$ \emph{induces} a subset of price-expected quality pairs $C$. The menu $C$ is given by $C =\{(p(B_{1}) ,\mathbb{E}_{\lambda _{B_{1}}}(X)) ,\ldots  ,(p(B_{n}) ,\mathbb{E}_{\lambda _{B_{n}}}(X))\}$ where $\mathbb{E}_{\lambda _{B_{i}}}(X)$ is the equilibrium expected quality of the sellers that belong to the set $B_{i}$.

Denoting, $q_{i}:=\mathbb{E}_{\lambda _{B_{i}}}(X)$, the menu $C$ yields the total transaction value 
\begin{align*}\pi \left (C\right ) &: =\sum _{ (p_{i} ,q_{i})  \in C}p_{i}D_{i}(C) \\ 
& =  \sum _{B_{i} \in I}p(B_{i})D_{I}(B_{i},\boldsymbol{p}) \\ 
&  = \sum _{B_{i} \in I}p(B_{i})\min \{D_{I}(B_{i},\boldsymbol{p}),S_{I}(B_{i} ,p(B_{i}))\} \\
& = \pi ^{Q} (I,\boldsymbol{p}). \end{align*}
 The first equality follows from the definition of $\pi$ (see Section \ref{Sec: price-dis}). The third equality follows from the fact that $(I ,\boldsymbol{p})$ is implementable. We conclude that the implementable information structure-pricing function pair $(I, \boldsymbol{p})$ yields the same revenue as the menu $C$ that it induces. 
 
 We denote by $\mathcal{C}^{Q}$ the set of all menus $C$ that are induced by some implementable $(I, \boldsymbol{p}) \in \mathcal{W}^{Q}$. With this notation, the platform's revenue maximization problem is equivalent to the constrained price discrimination problem of choosing a menu $C \in \mathcal{C}^{Q}$ to maximize $\sum p_{i}D_{i}(C)$ that we studied in Section \ref{Sec: price-dis}. That is, we have $\ensuremath{\operatorname*{max}}_{(I ,\boldsymbol{p}) \in \mathcal{W}^{Q}}\pi ^{Q} (I,\boldsymbol{p}) = \ensuremath{\operatorname*{max}}_{C \in \mathcal{C}^{Q}}\pi (C)$.

An information structure is {\em optimal} if it induces a menu that maximizes the platform's revenue.  The next subsection studies optimal information structures in this model, leveraging the equivalence with the constrained price discrimination problem.

 \subsection{Results} \label{subsec:results-quant}

 In this section we present our main results regarding the two-sided market model where the sellers choose quantities and the platform choose prices. 

Note that if $(I,\boldsymbol{p})$ induces the menu $C$ and $I$ is a $k$-separating information structure, then $C$ is a $k$-separating menu. We let $\mathcal{C}^{Q}_{k} \subseteq \mathcal{C}^{Q}$ be the set of $k$-separating menus. From the fact that the platform's revenue maximization problem transforms into the constrained price discrimination problem, Theorem  \ref{Theorem: Main} implies that if $\mathcal{C}^{Q}$ is regular and $F(m)m$ is convex, then the optimal information structure is $1$-separating, i.e., the optimal information structure consists of one element. In this subsection, we  establish certain natural conditions on the market model primitives that ensure regularity; these conditions then imply that if in addition $mF(m)$ is convex, then a $1$-separating information structure is optimal.
 
 Let $\varphi ^{Q}:\mathbb{I}(I_{o}) \rightrightarrows \mathcal{C}^{Q}$ be the set-valued mapping from the set $\mathbb{I}(I_{o})$ of all possible information structures to the set of  menus $\mathcal{C}^{Q}$ such that $C \in \varphi ^{Q}(I)$ if and only if $C$ is a menu that is induced by some implementable $(I,\boldsymbol{p})$. That is, $\varphi ^{Q}(I)$ contains all the menus that can be induced when the platform uses the information structure $I$. 
We note that the mapping $\varphi ^{Q}$ is generally complicated and there is no simple characterization of this mapping. However, we make substantial progress via the following proposition.  In particular, it can be shown that associated to every information structure $I$ there is a strictly convex program over the space of pricing functions $\boldsymbol{p}$, such that $(I,\boldsymbol{p})$ is implementable if and only if the solution to the program is $\boldsymbol{p}$.  Since every strictly convex program has at most one solution, this result also implies that the cardinality of $\varphi ^{Q}(I)$ is at most one; in other words, there is at most one menu $C$ that is induced when the platform uses the information structure $I$.  In addition, if the convex program does not have a solution for an information structure $I$, then there are no equilibrium prices associated with that information structure.

\begin{proposition} \label{Prop: unique equilibrium}
For every information structure $I \in \mathbb{I}(I_{o})$, there exists a strictly convex program over pricing functions such that $(I,\boldsymbol{p})$ is implementable if and only if the solution to the program is $\boldsymbol{p}$.  Therefore, there is at most one menu $C$ such that $C \in \varphi ^{Q}(I)$. 
\end{proposition}

To construct the claimed convex program in the preceding proposition, for every information structure $I=\{B_{1},\ldots,B_{n}\}$ we define an associated excess supply function. We show that the excess supply function satisfies the law of supply, i.e., the excess supply function is strictly monotone\footnote{A function $\zeta:\boldsymbol{P} \rightarrow \mathbb{R}^{n}$ is strictly monotone on $\boldsymbol{P}$ if for all $\boldsymbol{p} =(p_{1} ,\ldots  ,p_{n})$ and $\boldsymbol{p}^{\prime } =(p^{ \prime }_{1} ,\ldots ,p^{ \prime }_{n})$ that belong to $\boldsymbol{P}$ and satisfy $\boldsymbol{p} \neq \boldsymbol{p}^{ \prime }$, we have 
\begin{equation*}\left \langle  \zeta (\boldsymbol{p}) - \zeta (\boldsymbol{p}^{ \prime }) ,\boldsymbol{p} -\boldsymbol{p}^{ \prime }\right \rangle  > 0
\end{equation*} 
where $\left \langle \boldsymbol{x} ,\boldsymbol{y}\right \rangle  : =\sum _{i =1}^{n}x_{i}y_{i}$ denotes the standard inner product between two vectors $\boldsymbol{x}$ and $\boldsymbol{y}$ in $\mathbb{R}^{n}$.}  on a convex and open set $\boldsymbol{P} \subseteq \mathbb{R}^{n}$ such that if $\boldsymbol{p}$ is an equilibrium price vector then $\boldsymbol{p} \in \boldsymbol{P}$. The excess supply function is the gradient of some function $\psi$. Thus, minimizing $\psi$ over $\boldsymbol{P}$ is a strictly convex program that has a solution (minimizer) if and only if the solution is a zero of the excess supply function, i.e., an equilibrium price vector.  The result is helpful because it introduces a tractable convex program that for a given information structure provides an implementable price vector as its solution.

In the remainder of this subsection, we establish conditions for regularity of the space of menus induced under $\varphi ^{Q}$; these conditions are analogous to those discussed for the simple model in Section \ref{Sec: toy}.  First, note that in the Appendix we prove Lemma \ref{Lem: expected sellers} that states that given an information structure, the sellers' expected qualities do not depend on the prices as long as the prices are positive.  
This follows from the sellers' cost functions which imply that the sellers' optimal quantity decisions are homogeneous in the prices.  We assume for the rest of the section that $\mathbb{E}_{\lambda _{A_{1}}}(X)  < \ldots < \mathbb{E}_{\lambda _{A_{l}}}(X)$.   Let $\mathbb{E}_{\lambda _{B}}(X)$ be the resulting sellers' expected quality under the $1$-separating information structure $\{{B}\}$; then 
\begin{equation*}
p^{M}(B) =\ensuremath{\operatorname*{argmax}}_{p \geq 0}p \left (1 -F \left (\frac{p}{\mathbb{E}_{\lambda _{B}}(X)} \right ) \right ) 
\end{equation*} 
is the price that maximizes the platform's revenue under the $1$-separating information structure $\{{B}\}$ ignoring equilibrium conditions.


We denote by $\{B^{H}\} \in \{\{A_{1}\},\{A_{2}\},\ldots,\{A_{l}\}\}$ the information structure that generates the highest equilibrium price among the $1$-separating information structures $\{A_{1}\},\ldots,\{A_{l}\}$. That is, $\{(p(B^{H}),\mathbb{E}_{\lambda _{B^{H}}}(X))\} \in \varphi ^{Q}(\{{B^{H}}\})$ and $\{(p(B),\mathbb{E}_{\lambda _{B}}(X))\} \in \varphi ^{Q}(\{{B}\})$ imply $p(B^{H}) \geq p(B)$ for every $1$-separating information structure $\{{B}\}$ such that $\{{B}\} \in \{\{A_{1}\},\ldots,\{A_{l}\}\}$.

Theorem \ref{Thm: info structure1} shows that if 
\begin{equation} \label{Ineq: supply}
S_{\{B^{H}\}}(B^{H} ,p^{M}(B^{H})) \geq D_{\{B^{H}\}}(B^{H} ,p^{M}(B^{H}))  
\end{equation} 
and $F(m)m$ is strictly convex, then the optimal information structure is $1$-separating.  Inequality (\ref{Ineq: supply}) says that under the information structure $\{B^{H}\}$ and the price $p^{M}(B^{H})$, the supply exceeds the demand. This implies that under the information structure $\{B^{H}\}$, the equilibrium price is lower than the  optimal monopoly price that maximizes the platform's revenue, similarly to the condition discussed in section \ref{Sec: toy}. Hence, inequality (\ref{Ineq: supply}) implies  condition (ii) of the regularity definition (see Definition \ref{Def: marketN}) holds. 
 In order to prove that the optimal information structure is $1$-separating we show that condition (i) of the regularity definition also holds, and hence, the set of equilibrium menus $\mathcal{C}^{Q}$ is regular. As we discussed in section \ref{Sec: price-dis}, condition (i) means that removing low quality sellers increases the equilibrium price for high quality sellers. This is a natural condition in the context of two-sided market models. In the two-sided market model that we study in this Section we show that condition (i) holds without any further assumptions on the model's primitives. Thus, under the mild condition that ensures that the supply of high quality sellers is not too low (inequality (\ref{Ineq: supply})), we can apply Theorem \ref{Theorem: Main} to prove that the optimal information structure is $1$-separating under the convexity of $F(m)m$.

\begin{theorem} \label{Thm: info structure1}
Assume that $F(m)m$ is strictly convex on $[a,b]$. Assume that inequality (\ref{Ineq: supply}) holds.  Then, 

(i) The set $C^{Q}$ is regular. 

(ii) There exists a $1$-separating information structure $I^{ \ast }$ such that
\begin{equation*}(I^{ \ast } ,\boldsymbol{p}^{ \ast }) =\ensuremath{\operatorname*{argmax}}_{(I ,\boldsymbol{p}) \in \mathcal{W}^{Q}}\pi ^{Q} (I,\boldsymbol{p}) .
\end{equation*}
That is, there exists a $1$-separating information structure $I^{ \ast }$ that maximizes the platform's revenue.  

(iii) The pair $(I^{ \ast },p^{\ast})$ induces a menu that is maximal in $\mathcal{C}^{Q}_{1}$ and $ B^{\ast} \in I_{o}=\{A_{1},\ldots,A_{l}\}$ where $I^{\ast} = \{B^{\ast} \}$ is the optimal  information structure.\footnote{Recall that a menu $\{(p,q)\} \in \mathcal{C}_{1}^{Q}$ is maximal in $\mathcal{C}^{Q}_{1}$ if for every menu $\{(p',q')\} \in \mathcal{C}_{1}^{Q}$ such that $(p',q') \neq (p,q)$ we have $p > p'$ or $q >q'$ }  
\end{theorem}

 Theorem \ref{Thm: info structure1} shows that there exists a unique equilibrium price  $p^{eq}(A_{j})$  that the platform can induce  when it chooses the $1$-separating information structure $I=\{A_{j}\}$, i.e., $\varphi ^{Q}(I)$ is single-valued when $I$ is a $1$-separating information structure.  Further, under the natural condition that the equilibrium price is increasing in the sellers' quality, i.e., $p^{eq}(A_{j}) \leq p^{eq}(A_{k})$ whenever $\mathbb{E}_{\lambda _{A_{j}}}(X)  <  \mathbb{E}_{\lambda _{A_{k}}}(X)$, it is simple to show that there exists only one information structure-price pair $(\{A_{l}\},p^{eq}(A_{l}) ) $ that induces a maximal menu in $\mathcal{C}^{Q}_{1}$. Hence, in this case, Theorem \ref{Thm: info structure1} implies that the optimal $1$-separating information structure is $\{A_{l}\}$. That is, banning all sellers except the highest quality sellers is optimal for the platform. The necessity conditions developed in Section \ref{Sec:Necessity} show that when $F(m)m$ is not convex on $[a,b]$, then there exist examples where a $1$-separating menu is not be optimal.

Checking if inequality (\ref{Ineq: supply}) holds is straightforward given the model's primitives.  The following example illustrates that inequality (\ref{Ineq: supply}) holds if the sellers' costs in $B^{H}$ are low enough and/or the size of the supplier set $B^{H}$ is large enough. We note that if we introduce transfers or subsidies for each side of the market then the platform can always charge buyers and pay sellers in a way that inequality (\ref{Ineq: supply}) holds and the subsidies do not influence the platform's revenue.

\begin{example} \label{Ex: check inequality}
Suppose that $F(m)$ is the uniform distribution on $[0,1]$, i.e., $F(m)=m$ on $[0,1]$. Assume also that $\alpha=1$. A direct calculation shows that $p^{M}(B)=\mathbb{E}_{\lambda _{B}}(X)/2$. Hence, inequality (\ref{Ineq: supply}) holds if and only if 
\begin{equation} \label{Ineq: Ex1}
    1- \frac{p^{M}(B^{H})}{\mathbb{E}_{\lambda _{B^{H}}}(X)} \leq p^{M}(B^{H})\int _{B^{H}} k(x)^{-1} \phi (dx) \quad \Leftrightarrow \quad 1 \leq \int _{B^{H}} x k(x)^{-1} \phi (dx)
\end{equation}
where we use the fact that $\mathbb{E}_{\lambda _{B^{H}}}(X) \int _{B^{H}} k(x)^{-1} \phi (dx) = \int _{B^{H}} x k(x)^{-1} \phi (dx)$ (see Lemma \ref{Lem: expected sellers} in the Appendix). Thus, the size of the set $B^{H}$, the sellers' qualities in $B^{H}$, and the sellers' costs in $B^{H}$ determine whether inequality (\ref{Ineq: supply}) holds. In order to determine the information structure  $\{B^{H}\}$ with the highest equilibrium price we can solve for the equilibrium price:   
\begin{equation} \label{Ineq: Example2}
    1- \frac{p^{eq}(B)}{\mathbb{E}_{\lambda _{B}}(X)} = p^{eq}(B)\int _{B} k(x)^{-1} \phi (dx) \Leftrightarrow p^{eq}(B) =  \frac{\int _{B} xk(x)^{-1} \phi (dx)}{\int _{B} k(x)^{-1} \phi (dx)(1+\int _{B} xk(x)^{-1} \phi (dx))}
\end{equation}
and choose the set $B \in \{A_{1},\ldots,A_{l}\}$ with the highest equilibrium price. 

\end{example}

When the support of $F$ is unbounded it can be the case that inequality (\ref{Ineq: supply}) trivially holds because the supply under the price that maximizes the platform's revenue tends to infinity. For example, suppose that $F$ has the Pareto distribution, i.e.,  $F(m)=1-1/m^{\beta}$ on $[1,\infty)$. Then $F(m)m$ is convex for $\beta<1$. In this case, the support of $F$ is unbounded so $p^{M}$ is not necessarily well defined. Indeed, for every $q>0$ we have 
\begin{equation*}
    \lim _{p \to \infty} p \left (1-F\left ( \frac{p}{q}\right ) \right ) =  \lim _{p \to \infty} p \left ( \frac{q^{\beta}}{p^{\beta}}\right ) = \infty. 
\end{equation*}
Thus, the price that maximizes the platform's revenue tends to infinity which means that the supply under this price tends to infinity and inequality (\ref{Ineq: supply}) trivially holds.

As we discussed in Section \ref{Section: local}, in practice, the platform might consider only a few menus. In the model presented in this section, local results are not trivial to obtain as removing a set of sellers impacts the equilibrium prices in a complex way. Despite this, we compare any menu to the menu that is obtained by removing all the sellers except the sellers with the highest quality for the convex and concave distribution cases. The $1$-separating menu where the platform keeps only the highest quality sellers yields more revenue under local convexity and a local version of inequality (\ref{Ineq: supply}) while keeping the low quality sellers yields more revenue under local concavity and a local reverse version of inequality (\ref{Ineq: supply}). The proof follows from Proposition \ref{Prop: localB} and the fact that the equilibrium price for high quality sellers increases when the platform removes low quality sellers. 

\begin{proposition} \label{prop:localModel1}
 Let $(I,\boldsymbol{p}$) be an implementable information structure where $I =\{B_{1} ,B_{2} ,\ldots  ,B_{n}\}$ and assume that $\mathbb{E}_{\lambda _{B_{n}}}(X) > \mathbb{E}_{\lambda _{B_{i}}}(X)$ for all $i \neq n$. Suppose that $m(1-F(m))$ is strictly quasi-concave.\footnote{The assumption that  $m(1-F(m))$ is strictly quasi-concave guarantees that $p^{M}$ is well defined. }
 
 (i) If $F(m)m$ is convex on
  \begin {align*} 
 \left [\frac{p(B_{1})}{\mathbb{E}_{\lambda _{B_{1}}}(X)}, \frac{p(B_{n}) - p(B_{n-1})}{\mathbb{E}_{\lambda _{B_{n}}}(X) - \mathbb{E}_{\lambda _{B_{n-1}}}(X)} \right ] 
\end{align*}
and inequality (\ref{Ineq: supply}) holds for $B_{n}$ then $\pi ^{Q} (I) \leq \pi ^{Q} ( \{B_{n} \})$.

(ii) If $F(m)m$ is concave on
  \begin {align*} 
 \left [\frac{p(B_{1})}{\mathbb{E}_{\lambda _{B_{1}}}(X)}, \frac{p(B_{n}) - p(B_{n-1})}{\mathbb{E}_{\lambda _{B_{n}}}(X) - \mathbb{E}_{\lambda _{B_{n-1}}}(X)} \right ] 
\end{align*}
and $p(B_{n}) \geq p^{M}(B_{n})$ then $\pi ^{Q} (I) \geq \pi ^{Q} ( \{B_{n} \})$.
 
\end{proposition}

In the model we analyze in this section where sellers choose quantities, we can generally provide less local results than for the model where sellers choose prices we analyze in Section \ref{Sec: Two-sided2 prices}. The reason is that the constraint set generated by the buyers and sellers' equilibrium behavior, is harder to characterize in the current model compared to the model where sellers choose prices because of the Bertrand model's competitiveness nature. Despite this complexity, Proposition \ref{prop:localModel1} provides local results that show when a $1$-separating menu is better than other menus based on local properties.

\subsection{Supply and Demand Imbalances} \label{Subsec:ImbalancesModel1}

In the definition of equilibrium we assumed that  supply equals demand for each set in the information structure (see Definition \ref{Def:EqModel1}). This assumption restricts the prices that the platform can implement to market-clearing prices, and in fact, Proposition \ref{Prop: unique equilibrium} shows that there is a unique price vector that clears the market for a given information structure.  We assume this because we study quality selection decisions that are typically chosen for a long horizon (and do not change during that horizon), and hence, to guarantee that the market functions well, it is natural that  the prices chosen by the platform would  approximately make  total supply to equal total demand.  However, in general, a platform can set prices that may 
result in consistent imbalances between supply and demand;  specially, in the short-run sellers might have unsold products and buyers
can experience unsatisfied demand. 

In this section we analyze  the case where supply and demand are not equal. In this case, the set of feasible menus that the platform can choose from might be very large and the platform can vary prices depending on the specific information structure. We show that for the case where supply is greater than or equal to demand for each set in the information structure our results extend naturally. This is the observed case in many online marketplaces where a consumer sees available supply when interacting with the platform. In particular, ride-sharing platforms that choose prices use surge pricing mechanisms  to try to ensure that there is available  supply of drivers when a rider is interacting with the platform and searches for a ride \citep{hall2015effects}.  We now introduce formally the definition of equilibrium for a market that is demand constrained. 

\begin{definition} \label{Def:EqDemandConstModel1}
Given an information structure $I =\{B_{1},\ldots  ,B_{n}\}$ and a positive pricing function $\boldsymbol{p}=(p(B_{1}),\ldots,p(B_{n}))$, a demand constrained equilibrium is given by the buyers' demand $\{D_{I}(B_{i},\boldsymbol{p})\}_{i =1}^{n}$, sellers' supply $\{S_{I}(B_{i},p(B_{i}))\}_{i =1}^{n}$, and buyers' beliefs $\{\lambda _{B_{i}}\}_{i =1}^{n}$ that satisfy conditions (i), (ii), and (iii) in Definition \ref{Def:EqModel1} and the following condition: 

For all $B_{i} \in I$ the total supply is greater or equal to the total demand, i.e.,   
\begin{equation}\label{Eq:DemandConstrain}  S_{I}(B_{i} ,p(B_{i})) \geq D_{I}(B_{i} ,\boldsymbol{p}).
\end{equation}
\end{definition}

Note that the transformation to a constrained price discrimination problem we described in Section \ref{Subsec: equiv Quantity} still holds when supply is greater or equal to  demand for each set in the information structure. Hence, we can  leverage the tools we developed in Section \ref{Sec: price-dis} to analyze the demand constrained market too. In this case, the set of feasible menus changes and is typically much larger. In particular, Proposition \ref{Prop: unique equilibrium} does not hold when the demand is constrained as the platform has many menus to choose from for a given information structure.  In this case, the constraint set is typically not regular because very high prices keep supply above demand, and hence, are generally implementable. Despite the irregularity of the constraint set,  we show that  there exists a $1$-separating information structure that maximizes the platform's revenue under similar conditions to the conditions of Theorem \ref{Thm: info structure1}. 

\begin{proposition}\label{Prop:SupplyImbalnce1}
Consider the constrained demand market model. Assume that $F(m)m$ is strictly convex on $[a,b]$. Assume that inequality (\ref{Ineq: supply}) holds for $B^{H} = \{A_{l}\}$.  Then, 
there exists a $1$-separating information structure $I^{ \ast }$ that maximizes the platform's revenue.  \end{proposition}

The assumption that
inequality (\ref{Ineq: supply}) holds for $B^{H} = \{A_{l}\}$ has a similar interpretation the one we provide in the discussion before Theorem \ref{Thm: info structure1}. This assumption  implies that the equilibrium price of the highest quality sellers is lower than the monopoly price. 

We can also define an equilibrium for a supply constrained market. The definition is the same as Definition \ref{Def:EqDemandConstModel1} except that inequality (\ref{Eq:DemandConstrain}) is reversed. In this case the transformation to a constrained price discrimination described in Section \ref{Subsec: equiv Quantity} does not hold as the platform's revenues depend on the supply function.

\section{Two-Sided Market Model 2: Sellers Choose Prices} \label{Sec: Two-sided2 prices}

In this section we consider a model in which the sellers choose the prices and the quantities are determined in 
equilibrium. 

The platform chooses an information structure $I \in \mathbb{I}_{o}$ (see Section \ref{Sec: INFORMATION}). 
An information structure induces a game between buyers and sellers. In this game, sellers make entry decisions first. After the entry decisions, in each set of sellers that belongs to the information structure, the participating sellers engage in Bertrand competition. Buyers form beliefs about the sellers' quality and choose whether to buy a product and if so, from which set of sellers to buy it. 

Each equilibrium of the game induces a certain revenue for the platform. The platform's goal is to choose the information structure that maximizes the platform's equilibrium revenue. We now describe the sellers' and buyers' decisions in detail.

\subsection{\label{Sec: Buyers prices}Buyers}
In this section we describe the buyers' decisions. The buyers make their decisions after the sellers' entry and pricing decisions have been made. We denote by $H(B_{i}) \subseteq B_{i}$ the set of quality $x \in B_{i}$ sellers that participate in the platform and by $p_{x}$ the price that a quality $x \in  \cup _{B_{i} \in I}H(B_{i})$ seller charges.     

 As in Section \ref{Sec: Buyers-quantity},  the buyers' heterogeneity is described by a type space $[a,b] \subset \mathbb{R}_+$, and buyers' types are distributed according to a probability distribution function $F$ on $[a,b]$.  The buyers do not know the sellers' quality levels, but they know the information structure $I=\{B_{1},\ldots,B_{n}\}$ that the platform has chosen. Because the buyers do not have any information about the sellers' quality  aside from the information structure $I$, and there are no search costs or frictions, the buyers that decide to buy a product from quality $x \in B_{i}$ sellers buy it from the seller (or one of the sellers) with the lowest price in $B_{i}$. 

The preceding requirement implies that sellers cannot use prices in order to signal quality. That is, two sellers with quality levels $x_{1}$, $x_{2}$ such that  $x_{1} \in B_{i}$, $x_{2} \in B_{i}$ for some set $B_{i}$ in the information structure $I$ cannot disclose information about their quality level to the buyers.  Because the main focus of this section is examining the platform's quality selection decisions, we abstract away from information that sellers can disclose to buyers. In particular, our model abstracts away from the possibility that the sellers signal their quality through higher prices.  This may be an interesting avenue for future research.  

Given the information structure $I =\{B_{1} ,\ldots  ,B_{n}\}$ and the sets of sellers that participate in the platform $\{H(B_{i})\}_{B_{i} \in I}$, $H(B_{i}) \subseteq B_{i}$, the buyers form beliefs $\lambda _{B_{i}} \in \mathcal{P}(X)$ about the quality level of type $x \in B_{i}$ sellers.\footnote{With slight abuse of notations we use similar notations to those of Section \ref{Sec: Buyers-quantity}.} In equilibrium, the buyers' beliefs are consistent with the sellers' entry decisions   and with Bayesian updating. That is, $\lambda _{B_{i}}$ describes the conditional distribution of $\phi$ given $H(B_{i})$, i.e., $\lambda _{B_{i}}(A) =\phi (A\vert H(B_{i}))$ where $\phi (A\vert H(B_{i})) := \frac{\phi(A \cap H(B_{i}))}{\phi (H(B_{i}))}$ for every (measurable) set $A$ and all $B_{i} \in I$ such that $\phi (H(B_{i}))>0$.

We denote by $p(B_{i}) =\inf _{x \in H(B_{i})}p_{x}$ the lowest price among the sellers in the set $B_{i}$. A type $m \in [a,b]$ buyer's utility from buying a product from quality $x \in B_{i}$ sellers is given by
\begin{equation*}Z(m ,B_{i} ,p(B_{i})) = m\mathbb{E}_{\lambda _{B_{i}}}(X) - p(B_{i}).
\end{equation*} 
 $\mathbb{E}_{\lambda _{B_{i}}}(X)$ is the sellers' expected quality given the buyers' beliefs $\lambda _{B_{i}}$. A type $m$ buyer buys a product from a quality $x \in B_{i}$ seller if $Z(m ,B_{i} ,p(B_{i})) \geq 0$ and $Z(m ,B_{i} ,p(B_{i})) =\max _{B \in I} Z(m ,B ,p(B))$, and does not buy a product otherwise.

The total demand in the market for products that are sold by type $x \in B_{i}$ sellers $D_{I}(B_{i} ,p(B_{1}) ,\ldots  ,p(B_{n}))$ who charge the lowest price in $B_{i}$ is given by  
\begin{equation}D_{I}(B_{i} ,p(B_{1}) ,\ldots  ,p(B_{n})) =\int _{a}^{b}1_{\{Z(m ,B_{i} ,p(B_{i})) \geq 0\}}1_{\{Z(m ,B_{i} ,p(B_{i})) =\max _{B \in P}Z(m ,B ,p(B))\}}F(dm). \label{Demand  buyers}
\end{equation}
The total demand in the market for products that are sold by type $x \in B_{i}$   sellers that do not charge the lowest price in $B_{i}$ is zero.  

\subsection{\label{Sec: Sellers prices}Sellers}
In this section we describe the sellers' decisions. 
Sellers first choose whether to participate in the platform or not. In each set $B_{i} \in I$ that belongs to the information structure, participating sellers price their products simultaneously and engage in price competition with other sellers whose quality levels belong to the set $B_{i} \in I$. Because a buyer that decides to buy a product from a quality $x \in B_{i}$ seller buys it from the seller (or one of the sellers) who charges the lowest price in the set $B_{i}$, the price competition between the sellers resembles Bertrand competition. 

A quality $x \in B_{i} \subseteq X$  seller that participates in the platform sells a quantity given by \\  $h_{I}(B_{i},H(B_{i}) ,p_{x} ,p(B_{1}) ,\ldots  ,p(B_{n}))$ units if the set of participating sellers is $H(B_{i})$, the price that $x$ charges is $p_{x} \in \mathbb{R}_{ +}$, and $p(B_{i}) =\inf _{x \in H(B_{i}) \backslash \{x\} }p_{x}$ is the lowest price among the other sellers in the set $H(B_{i})$. We denote by $M_{I}(B_{i} , p(B_{1}) ,\ldots  ,p(B_{n}))$ the total mass of sellers whose quality levels belong to $B_{i}$ and who charge the price $p(B_{i})$. The quantity allocation function $h_{I}$ is determined in equilibrium and is given by  
\begin{align}h_{I} (B_{i},H(B_{i}) ,p_{x} ,\boldsymbol{p}) =\left \{\begin{array}{c}\infty \text{  if  }p_{x} < p(B_{i}) \text{,  \ }D_{I}(B_{i} ,\boldsymbol{p}) >0 \\
\frac{D_{I}(B_{i} ,\boldsymbol{p})}{M_I(B_{i} ,\boldsymbol{p})}\text{ if  }p_{x} =p(B_{i})\text{, \  }D_{I}(B_{i} ,\boldsymbol{p}) >0 \\
0 \text{ if } p_{x} > p(B_{i})\text{, or  }D_{I}(B_{i} ,\boldsymbol{p}) =0\end{array}\right.
\end{align}
where $\boldsymbol{p} : =(p\left (B_{1}\right ) ,\ldots  ,p(B_{n}))$ and we define $D_{I}(B_{i} ,\boldsymbol{p})/M_{I}(B_{i},\boldsymbol{p}) =\infty $ if $M_{I}(B_{i} ,\boldsymbol{p}) =0$ and $D_{I}(B_{i} ,\boldsymbol{p}) >0$. This quantity allocation resembles the quantity allocation in the standard Bertrand competition model with a continuum of sellers. In particular, when multiple active sellers' charge the same lowest price within a set, the buyers' demand splits evenly between those sellers. 

A quality $x \in B_{i} \subseteq X$ seller's utility from participating in the platform is given by 
\begin{equation*}\overline{U}(x ,H(B_{i}) ,p_{x} ,p(B_{1}) ,\ldots  ,p(B_{n})) =h_{I}(B_{i},H(B_{i}) ,p_{x} ,p(B_{1}) ,\ldots  ,p(B_{n}))( p_{x} -c(x)).
\end{equation*}
We assume that the cost function $c$ is positive and constant on each element of the partition $I_{o}$, i.e., $x_{1} ,x_{2} \in A_{i}$ and $A_{i} \in I_{o}$ imply $c(x_{1}) =c(x_{2}) =c(A_{i})$. The assumption that the cost function $c$ is constant on each element of the partition $I_{o}$ means that the cost function is measurable with respect to the platform's information, i.e., the platform knows the sellers' costs but not the sellers' quality levels. This assumption simplifies the analysis but is not essential to our results. We also assume that the cost function is increasing, i.e., $c(A_{i}) < c(A_{j})$ for $i<j$. This assumption means that producing higher quality products costs more. 
A quality $x \in X$ seller's utility from not participating in the platform is normalized to $0$.

\subsection{\label{Equilibrium prices}Equilibrium}

In this section we define the equilibrium concept that we use for the game described above. For simplicity, we focus on a symmetric equilibrium in the sense that for all $B_{i} \in I$, all the sellers that participate in the platform charge the same price. With slight abuse of notation, we denote this price by $p(B_{i})$, i.e., $p_{x} =p(B_{i})$ for all $x \in H(B_{i})$, $B_{i} \in I$.

\begin{definition}
Given an information structure $I =\{B_{1} ,\ldots  ,B_{n}\}$, an equilibrium consists of a vector of positive prices $\boldsymbol{p}=(p(B_{1}) ,\ldots  ,p(B_{n})) \in \mathbb{R}^{\vert I\vert }$, positive masses of sellers that participate in the platform $\{M_{I}(B_{i},\boldsymbol{p})\}_{B_{i} \in I}$, positive masses of demand $\{D_{I}(B_{i} ,\boldsymbol{p})\}_{B_{i} \in I}$, and buyers' beliefs $\lambda  =(\lambda _{B_{i}})_{B_{i} \in I}$ such that  

(i) Sellers' optimality: The sellers' decision are optimal. That is, 
\begin{equation*}
p(B_{i}) =\ensuremath{\operatorname*{argmax}}_{p_{x} \in \mathbb{R}_{ +}}\overline{U} (x ,H(B_{i}) ,p_{x} ,\boldsymbol{p}) \end{equation*}  
is the price that seller $x \in H(B_{i})$ charges. 
In addition, seller $x \in B_{i}$ enters the market, i.e., $x \in H(B_{i})$, if and only if 
$\overline{U}(x ,H(B_{i}) ,p(B_{i}) ,\boldsymbol{p}) \geq 0$. 

(ii) Buyers' optimality: The buyers' decisions are optimal. That is, for each buyer $m \in [a ,b]$ that buys from type $x \in B_{i}$ sellers, we have $Z(m ,B_{i} ,p(B_{i})) \geq 0$ and $Z(m ,B_{i} ,p(B_{i})) =\max _{B \in I}Z(m ,B ,p(B))$. 

(iii) Rational expectations: $\lambda _{B_{i}}(A)$ is the probability that a buyer is matched to sellers whose quality levels belong to the set $A$ given the sellers' entry decisions,  
i.e., 
\begin{equation*}\lambda _{B_{i}}(A) =\phi (A\vert H(B_{i})) = \frac{\phi(A \cap H(B_{i}))}{\phi (H(B_{i}))}
\end{equation*}
for every (measurable) set $A$ and for all $B_{i} \in I$. 

(iv) Market clearing: For all $B_{i} \in I$ we have 
\begin{equation*}
    M_{I}(B_{i},\boldsymbol{p})h_{I}(B_{i},H(B_{i}) ,p(B_{i}) ,\boldsymbol{p}) = D_{I}(B_{i} ,\boldsymbol{p}) \ ,
\end{equation*}  
where $M_{I}(B_{i},\boldsymbol{p})=\phi(H(B_{i}))$ 
 is the mass of sellers in $B_{i}$ that participate in the platform; $D_{I}(B_{i} ,\boldsymbol{p})$ and $h_{I}(B_{i},H(B_{i}) ,p(B_{i}) ,\boldsymbol{p})$ are defined in Sections \ref{Sec: Buyers prices} and \ref{Sec: Sellers prices}, respectively.  
\end{definition}

We say that an information structure $I$ is {\em implementable} if there exists an equilibrium $(\boldsymbol{p}, D ,M ,\lambda )$ under $I$ where $D =\{D_{I}(B_{i} ,\boldsymbol{p})\}_{B_{i} \in I}$, $M =\{M(B_{i},\boldsymbol{p}\}_{B_{i} \in I}$, and $\lambda  =\{\lambda _{B_{i}}\}_{B_{i} \in I}$. We denote by $\mathcal{W}^{P}$ the set of all implementable information structures.  

The platform's goal is to choose an implementable information structure to maximize the total transaction value $\pi ^{P}$  given by 
\begin{equation*}\pi ^{P} (I) :=\sum _{B_{i} \in I}p(B_{i})\min \{D_{I}(B_{i} ,\boldsymbol{p}) , M_{I}(B_{i},\boldsymbol{p})h_{I}(B_{i},H(B_{i}) ,p(B_{i}) ,\boldsymbol{p})\}.
\end{equation*}

\subsection{Equivalence with Constrained Price Discrimination}

As in Section \ref{Subsec: equiv Quantity}, the platform's revenue maximization problem described above transforms into the constrained price discrimination problem that we analyzed in Section \ref{Sec: price-dis}. To see this, note that an implementable information structure $I =\{B_{1} ,B_{2} ,\ldots  ,B_{n}\}$ and an associated equilibrium price vector $\boldsymbol{p}$ induce a menu $C$ that is given by 
$$C =\{(p(B_{1}) ,\mathbb{E}_{\lambda _{B_{1}}}(X)) ,\ldots ,(p(B_{n}) ,\mathbb{E}_{\lambda _{B_{n}}}(X))\}$$ where $\mathbb{E}_{\lambda _{B_{i}}}(X)$ is the equilibrium expected quality of the sellers that belong to the set $B_{i}$ and $\boldsymbol{p}=(p(B_{1}),\ldots,p(B_{n}))$ is the vector of equilibrium prices. The implementable information structure $I$ yields the same revenue as the menu $C$ that it induces (see Section \ref{Subsec: equiv Quantity}).
 We denote by $\mathcal{C}^{P}$ the set of all menus $C$ that are induced by some implementable information structure $I \in \mathcal{W}^{P}$. With this notation, the platform's revenue maximization problem is equivalent to the constrained price discrimination problem of choosing a menu $C \in \mathcal{C}^{P}$ to maximize $\sum p_{i}D_{i}(C)$ that we studied in Section \ref{Sec: price-dis}. 

 \subsection{Results}
 In this section we present our main results regarding the two-sided market model in which the sellers choose the prices. 
 
 Let $\varphi ^{P}:\mathbb{I}(I_{o}) \rightrightarrows \mathcal{C}^{P}$ be the set-valued mapping from the set $\mathbb{I}(I_{o})$ of all possible information structures to the set of  menus $\mathcal{C}^{P}$ such that $C \in \varphi ^{P}(I)$ if and only if $C$ is a menu that is induced by the  information structure $I$. As opposed to the two-sided market model that we study in Section \ref{Sec: Model 1 quant}, the mapping $\varphi ^{P}$ can be explicitly characterized in the current setting. This is because  Bertrand competition pins down the equilibrium prices (to the lowest marginal costs within a set in the information structure).

 For an information structure $I=\{B_{1},\ldots,B_{n}\}$ let $L(I)=\{G_{1},\ldots,G_{n}\}$ be an information structure such that $G_{j} \in I_{o}$ for all $G_{j} \in L(I)$ and $G_{j}$ is the set with the lowest index among the blocks of $B_{j}$, i.e., among the sets  $\{A_{k}\}$ such that $B_{j}=\cup_{k} A_{k}$. For example, if $B_1=A_1 \cup A_2$, then $G_1=A_1$. We assume without loss of generality that $c(G_{1}) < \ldots < c(G_{n})$ for every information structure $I$. The following theorem shows that for every implementable information  structure $I$ and for every set $B_{i} \in I$, the equilibrium price for sellers in $B_{i}$ equals $c(G_{i})$. This fact follows directly from our Bertrand competition assumption. Further, using this characterization of the equilibrium prices it  follows that $\mathcal{C}^{P}$ satisfies the condition of Corollary \ref{Corr: Bert}. Hence, Theorem \ref{thm: Bertrand} provides a full characterization of the implementable information structures. 

\begin{theorem} \label{thm: Bertrand}
Let $I$ be any information structure. Suppose that $C \in \varphi ^{P}(I)$. 

(i) We have 
\begin{equation*}
C =\{(c(G_{1}) ,\mathbb{E}_{\lambda _{G_{1}}}(X)) ,\ldots  ,(c(G_{n}) ,\mathbb{E}_{\lambda _{G_{n}}}(X))\} 
\end{equation*}
where $L(I)=\{G_{1},\ldots,G_{n}\}$ and $\lambda _{G_{i}}(A) =  \phi(A \cap G_{i})/ \phi (G_{i})$ for every measurable set $A$.

(ii) We have  $\{(c(G_{n}), \mathbb{E}_{\lambda _{G_{n}}}(X))\}  \in \varphi ^{P}(\{B_{n}\})$.

(iii) Suppose that $I_{o}$ is implementable and $C_{o} \in \varphi ^{P}(I_{o})$. Then $\mathcal{C}^{P} = 2^{C_{o}}$. 

\end{theorem}

The proof of the following Corollary follows immediately from Theorem \ref{thm: Bertrand} and Corollary \ref{Corr: Bert}.

\begin{corollary} \label{Corr:Bertrand-1-optimal}
Assume that $F(m)m$ is convex on $[a,b]$. Then there exists a $1$-separating information structure that maximizes the platform's revenue. 
\end{corollary}

Note that the only $1$-separating information structure that induces a menu that is maximal in $\mathcal{C}^{P}_{1}$ is $\{A_{l}\}$. Thus, when $I_{o}$ is implementable and the constraint set $\mathcal{C}^{P} = 2^{C_{o}}$ is regular (i.e., the equilibrium price is lower than the monopoly price under the information structure $\{A_{l}\}$), Theorem \ref{Theorem: Main} implies that the optimal information structure is $\{A_{l}\}$. That is, the optimal information structure bans all sellers except the highest quality sellers.

 As we discussed in Section \ref{Section: local}, in practice, a platform might consider only a small number of options, e.g., removing the lowest quality sellers or keeping them. In order to determine whether banning these low quality sellers is beneficial, the platform needs to measure the density function's elasticity only locally. If the density function's elasticity is bounded below by $-2$  on some local region that depends on the prices and qualities of the low quality sellers, then it is beneficial to ban these sellers. Conversely, if the density function's elasticity is bounded above by $-2$  on this local region, then it is beneficial to keep these sellers  (see Corollary \ref{Coro: Prices locally}). For many distribution functions the density function's elasticity is decreasing. In this case Corollary \ref{Coro: Prices locally} implies that the platform needs to check the density function's elasticity only at one point. For example, if at the highest point of the relevant interval (this point depends on the equilibrium prices and qualities) the density function's elasticity is greater than $-2$, then it is greater than $-2$ over the relevant interval. In practice, the platform might be able to estimate this elasticity with price experimentation. 

\begin{corollary} \label{Coro: Prices locally}

 Let $I=\{B_{1} , \ldots , B_{n} \} $ be an implementable information structure. 

Let $C = \{(p(G_{1}) ,\mathbb{E}_{\lambda _{G_{1}}}(X)) ,\ldots  ,(p(G_{n}),\mathbb{E}_{\lambda _{G_{n}}}(X))) \} \in \varphi ^{P}(I)$  where  $L(I)=\{G_{1},\ldots,G_{n}\}$. Consider the (implementable) information structure $I'=\{B_{2} , \ldots , B_{n} \} $. Then  \begin {align*} 
\pi ^{P} (I) \leq \pi ^{P} (I') \text{ if } F(m)m \text{ is convex on } \left [\frac{p(G_{1})}{\mathbb{E}_{\lambda _{G_{1}}}(X)}, \frac{p(G_{2}) - p(G_{1})}{\mathbb{E}_{\lambda _{G_{2}}}(X) - \mathbb{E}_{\lambda _{G_{1}}}(X)} \right ]  \\
\pi ^{P} (I) \geq \pi ^{P} (I') \text{ if } F(m)m \text{ is concave on } \left [\frac{p(G_{1})}{\mathbb{E}_{\lambda _{G_{1}}}(X)}, \frac{p(G_{2}) - p(G_{1})}{\mathbb{E}_{\lambda _{G_{2}}}(X) - \mathbb{E}_{\lambda _{G_{1}}}(X)} \right ]
\end{align*}

\end{corollary}

We also show that when $F(m)m$ is concave and $I_{o}$ is implementable, the optimal information structure is $I_{o}$, i.e., the platform reveals all the information it has about the sellers' quality. The proof of the following Corollary follows from Theorem \ref{thm: Bertrand} and Proposition \ref{Prop: localA}. Other local results that compare only two specific information structures can be obtained by applying Proposition \ref{Prop: localA}.

\begin{corollary} \label{Coro: Prices locally2}

 Assume that $I_{o}$ is implementable. Let $C_{o} = \{(p(A_{1}) ,\mathbb{E}_{\lambda _{A_{1}}}(X)) ,\ldots  ,(p(A_{l}),\mathbb{E}_{\lambda _{A_{l}}}(X)) \} \in \varphi ^{P}(I_{o})$. 
Suppose that $F(m)m$ is concave on \begin{equation*}
    \left [\frac{p(A_{1})}{\mathbb{E}_{\lambda _{A_{1}}}(X)}, \frac{p(A_{l}) - p(A_{l-1})}{\mathbb{E}_{\lambda _{A_{l}}}(X) - \mathbb{E}_{\lambda _{A_{l-1}}}(X)} \right ].
\end{equation*}
Then the optimal information structure is $I_{o}$.

\end{corollary}

We note that the discussion in Section \ref{Subsec:ImbalancesModel1} applies also to the model where sellers choose prices, i.e., our results hold also for  demand constrained markets. We omit the details for brevity. 

The results we provide in Section \ref{Sec: price-dis} show that the conditions of Corollary \ref{Corr:Bertrand-1-optimal} are necessary for the optimality of a $1$-separating menu. When $F(m)m$ is not convex, we can find parameters of the model such that a $1$-separating menu is not optimal. In addition, Proposition \ref{Corr:convex-concave} implies that we can analyze the important case where $F(m)m$ is convex on an interval $[a,m^{*}]$ and concave on the complementary interval $[m^{*},b]$ for some $m^{*} \in [a,b]$ (see Section \ref{sec:convex-concave} for more details). The following Corollary follows directly from Proposition  \ref{Corr:convex-concave}. 
\begin{corollary}
    Suppose that $I_{o}$ is implementable and let $C_{o} = \{(p(A_{1}) ,\mathbb{E}_{\lambda _{A_{1}}}(X)) ,\ldots  ,(p(A_{l}),\mathbb{E}_{\lambda _{A_{l}}}(X)) \} \in \varphi ^{P}(I_{o})$. Assume that $F(m)m$ is strictly convex on  $[a,m^{*}]$ and strictly concave on $[m^{*},b]$ for some $m^{*} \in [a,b]$. Let $1 \leq k \leq n-1$ be an integer. Then the optimal menu is $k$-separating or $k+1$-separating when 
    $$m^{*} \in  \left [\frac{p(A_{n-k}) - p(A_{n-k-1})}{\mathbb{E}_{\lambda _{A_{n-k}}}(X) - \mathbb{E}_{\lambda _{A_{n-k-1}}}(X)} , \frac{p(A_{n-k+1}) - p(A_{n-k})}{\mathbb{E}_{\lambda _{A_{n-k+1}}}(X) - \mathbb{E}_{\lambda _{A_{n-k}}}(X)} \right ].$$ 
\end{corollary}

\section{Conclusions} \label{Sec: summary}

In this paper we study optimal information disclosure policies for online platforms. We introduce two distinct two-sided market models. In the first model the sellers choose quantities, and in the second model the sellers make entry and pricing decisions. The platform's information disclosure  problem transforms into a constrained price discrimination problem, where the constraints are given by the equilibrium requirements and depend on the specific two-sided market model being studied. We use this equivalence to provide conditions on model parameters, under which a simple information structure  where the platform removes a certain portion of low quality sellers 
and does not share any information about the other sellers 
is revenue-optimal for the platform.

There are some interesting potential extensions for future work. For example, in practice, the platform and the buyers learn the sellers'  quality as they make their decisions. One possible extension of our work would be to incorporate learning into our setting. Another direction for future work is to introduce competition between platforms. In many industries, fierce competition between platforms has a first order effect on the market design choices made by platforms. Another possible extension is to introduce signaling by sellers. In some platforms sellers have tools to signal their own quality that can impact equilibrium outcomes.

Finally, a fourth interesting direction for future research is to incorporate search frictions in our setting. In some platforms (e.g., e-commerce platforms) search frictions play a significant role. In some of these platforms, because of rating inflation, the sellers' star rating does not provide substantial information about the sellers' quality (see, e.g., \cite{tadelis2016reputation}). In this case, the menu observed in practice sometimes looks similar to a $2$-separating menu: certified sellers, sellers that are not certified, and sellers that are banned. While the results in this paper show that a $1$-separating menu is optimal under an appropriate condition on demand elasticity, we conjecture that extending our setting to incorporate search costs would change the optimal menu. In particular, in order to mitigate the impact of search, a $2$-separating menu might be more attractive.

\newpage
\bibliographystyle{ecta}
\bibliography{Quality}

\newpage 

\appendix

\section{Appendix}

\subsection{Proofs of Section \ref{Sec: price-dis}}

We first introduce some definitions. A menu $C \in \mathcal{C}_{p}$ is called \textit{price-M} if for all $(p,q) \in [0,\infty) \times [0,\infty)$ such that $C \cup \{ (p,q) \} \in \mathcal{C}_{p}$, we have $p \leq p'$ for some $(p',q') \in C$. In words, a menu $C$ is price-M if it is not feasible to add a price-quality pair to $C$ with positive demand and a higher price than all the other prices in the menu $C$.

Step 3 in the proof of Theorem \ref{Theorem: Main} shows that the optimal menu (if it exists) is price-M. This also shows that Theorem \ref{Theorem: Main} holds under the following weaker version of the first condition of Definition \ref{Def: marketN} (the  regularity condition): For every price-M menu 
 $C =\left \{\left (p_{1} ,q_{1}\right ) ,\ldots  ,\left (p_{k} ,q_{k}\right )\right \}$ there exists a $1$-separating menu $\{ (p,q)\} \in \mathcal{C}_{1}$ such that $p \geq p_{k}$ and $q \geq q_{k}$.   
 
Recall that given some quality $q >0$, the price that maximizes the platform's revenue $p^{M}(q)$ is given by 
\begin{equation*} 
    p^{M}(q)=  \operatorname{argmax}_{p \geq 0}p \left (1 - F\left (\frac{p}{q} \right ) \right ).
    \end{equation*}
  Note that $p^{M}(q)$ is single-valued under the assumptions of Theorem \ref{Theorem: Main}. A $1$-separating menu $\{(p,q)\}$ is maximal in $\mathcal{C}_{1}$ if for every $ \{(p',q') \} \in \mathcal{C}_{1}$ such that $(p',q') \neq (p,q)$ we have $p > p'$ or $q > q'$. 

\begin{proof}[Proof of Theorem \ref{Theorem: Main}]
Let $C =\{(p_{i} ,q_{i})_{i =1}^{n}\} \in \mathcal{C}$ be a menu such that $p_{k} \leq p_{j}$ for all $k <j$ and $n>1$. We can assume\protect\footnote{
If for some $\left (p_{k} ,q_{k}\right )$ in $C$ we have $D_{k}(C) =0$, then the menu $C\backslash \{\left (p_{k} ,q_{k}\right )\}$ has the same total transaction value as the menu $C$. Thus, we can consider the menu $C\backslash \{\left (p_{k} ,q_{k}\right )\}$ instead of the menu $C$. 
} that the demand for each price-quality pair in $C$ has a positive mass. That is
\begin{equation*}D_{i}(C) =\int _{a}^{b}1_{\{mq_{i} -p_{i} \geq 0\}}1_{\{mq_{i} -p_{i} =\max _{i =1 ,\ldots  ,n}mq_{i} -p_{i}\}}F(dm) >0
\end{equation*}for all $1 \leq i \leq n$. Note that $D_{i}(C) >0$ for all $1 \leq i \leq n$ implies that $q_{k} <q_{j}$ for all $k <j$. 

\textbf{Step 1}. The total transaction value from the menu $C$ is given by \begin{equation*}\pi \left (C\right ) =\sum _{i =1}^{n}p_{i}\left (F(m_{i +1}) -F\left (m_{i}\right )\right )
\end{equation*}
where $m_{n +1} =b$ and the numbers $\{m_{i}\}_{i =2}^{n}$ satisfy $m_{i} \in [a ,b]$ for all $2 \leq i \leq n$ and
\begin{equation*}m_{i}q_{i} -p_{i} =m_{i}q_{i-1} -p_{i -1}
\end{equation*}where $q_{0} =p_{0} =0$. 
The number $m_{1}$ satisfies $m_{1}= \max \{a,p_{1}/q_{1}\}$.

\textbf{Proof of Step 1}. The proof of Step 1 is standard (see \cite{maskin1984monopoly}). We provide it here for completeness. 

Because $q_{n} >q_{j}$ for all $1 \leq j \leq n -1$, if for some $1 \leq j \leq n -1$ and $m \in [a ,b]$ we have 
\begin{equation*}m\left (q_{n} -q_{j}\right ) \geq p_{n} -p_{j}
\end{equation*}
then 
\begin{equation*}m^{ \prime }\left (q_{n} -q_{j}\right ) \geq p_{n} -p_{j}
\end{equation*}
for all $m^{ \prime } \in [m ,b]$. Thus, if for some $m \in [a ,b]$ we have 
\begin{equation} \label{Step 2.1} mq_{n} -p_{n} \geq \max \{\max _{1 \leq j \leq n -1}mq_{j} -p_{j} ,0\} 
\end{equation}
then inequality (\ref{Step 2.1}) holds for all $m^{ \prime } \in [m ,b]$. In other words, if a type $m$ chooses the price-quality pair $(p_{n} ,q_{n})$, then every type $m^{ \prime }$ with $m \leq m^{ \prime } \leq b$ chooses the price-quality pair $(p_{n} ,q_{n})$.  

Let 
\begin{equation*}W_{n} : =\{m \in [a ,b] :mq_{n} -p_{n} \geq \max \{\max _{1 \leq j \leq n -1}mq_{j} -p_{j} ,0\}\}
\end{equation*}
be the set of types that choose the price-quality pair $(p_{n} ,q_{n})$. Define $m_{n} =\min W_{n}$. $D_{n}(C) >0$ implies that the set $W_{n}$ is not empty. From the fact that $m \in W_{n}$ implies $m^{ \prime } \in W_{n}$ for all $m \leq m^{ \prime } \leq b$, $W_{n}$ equals the interval $[m_{n} ,b]$. Thus,
\begin{equation*}D_{n}(C) =\int _{a}^{b}1_{W_{n}}(m)F(dm) =F\left (b\right ) -F\left (m_{n}\right ) =F\left (m_{n +1}\right ) -F\left (m_{n}\right )
\end{equation*}
where $m_{n +1} : =b$ so \ $F(m_{n +1}) =1$. 

Define $m_{i} =\min W_{i}$ where we define the sets \begin{equation*}W_{i} : =\{m \in [a ,m_{i +1}] :mq_{i} -p_{i} \geq \max \{\sup _{1 \leq j \leq i -1}mq_{j} -p_{j} ,0\}
\end{equation*} 
for all $1 \leq i \leq n -1$. $D_{i}(C)>0$ implies that $W_{i}$ is not empty. Thus, $m_{i}$ is well defined. From the same argument as the argument above, if a type $m \in W_{i}$ chooses the price-quality pair $(p_{i} ,q_{i})$, then every type $m^{ \prime }$ with $m \leq m^{ \prime } \leq m_{i +1}$   chooses the price-quality pair $(p_{i} ,q_{i})$. Thus, $W_{i}$ equals the interval $[m_{i},m_{i +1}]$ and \begin{equation*}
    D_{i}(C)= \int_{a}^{b} 1_{W_{i}} (m) F(dm) = F(m_{i +1}) -F(m_{i}) > 0
 \end{equation*} for all $1 \leq i \leq n$.

Note that $W_{1} =\{m \in [a ,m_{2}] :mq_{1} -p_{1} \geq 0\}$. The continuity of the function $mq_{1} -p_{1}$ implies that $m_{1} =\min W_{1}$ satisfies $m_{1}= \max \{a,p_{1}/q_{1}\}$. Using continuity again and the definition of $m_{2}$ we conclude that $m_{2}q_{2} -p_{2} =m_{2}q_{1} -p_{1}$. Similarly, $m_{i}q_{i} -p_{i} =m_{i}q_{i-1} -p_{i -1}$ for all $2 \leq i \leq n$.

Thus, the total transaction value from the menu $C$ is given by \begin{equation*}\pi (C) =\sum _{i =1}^{n}p_{i}D_{i}(C) =\sum _{i =1}^{n}p_{i}\left (F\left (m_{i +1}\right ) -F\left (m_{i}\right )\right )
\end{equation*}
where $m_{n +1} =b$ and the numbers $\{m_{i}\}_{i =1}^{n}$ satisfy $m_{i} \in [a ,b]$ for all $1 \leq i \leq n$ and $m_{i}q_{i} -p_{i} =m_{i}q_{i-1} -p_{i -1}$, $q_{0} =p_{0} =0$.



\textbf{Step 2}. Let $0 =d_{0} <d_{1} <\ldots  <d_{k}$ and $0 =z_{0} <\ldots  <z_{k}$. Assume that $\left (z_{i} -z_{i -1}\right )/\left (d_{i} -d_{i -1}\right ) \in [a ,b]$ for all $1 \leq i \leq k$. Then \begin{equation} \label{Fundem ineq} z_{k}F\genfrac{(}{)}{}{}{z_{k}}{d_{k}} \leq \sum _{i =1}^{k}\left (z_{i} -z_{i -1}\right )F\genfrac{(}{)}{}{}{z_{i} -z_{i -1}}{d_{i} -d_{i -1}} . 
\end{equation}

\textbf{Proof of Step 2. }
Note that the function $f(x ,y) =xF \left (\frac{x}{y} \right )$ is convex on $E =\{\left (x ,y\right ) :x/y \in [a ,b] ,$ $y >0\}$. To see this, recall that the perspective function 
\begin{equation*}yF\genfrac{(}{)}{}{}{x}{y}\frac{x}{y} =xF\genfrac{(}{)}{}{}{x}{y} =f(x ,y)
\end{equation*}
is convex on $E$ when $F(x)x$ is convex on $[a,b]$. 

From Jensen's inequality we have \begin{equation*}k^{ -1}\sum _{i =1}^{k}x_{i}F\genfrac{(}{)}{}{}{k^{ -1}\sum _{i =1}^{k}x_{i}}{k^{ -1}\sum _{i =1}^{k}y_{i}} =f\left (k^{ -1}\sum _{i =1}^{k}\left (x_{i} ,y_{i}\right )\right ) \leq k^{ -1}\sum _{i =1}^{k}f\left (x_{i} ,y_{i}\right ) =k^{ -1}\sum _{i =1}^{k}x_{i}F\genfrac{(}{)}{}{}{x_{i}}{y_{i}}
\end{equation*}
for all $(x_{1} ,\ldots  ,x_{k})$ and $(y_{1} ,\ldots  ,y_{k})$ such that $(x_{i} ,y_{i}) \in E$ for all $i =1 ,\ldots  ,k$. Thus, \begin{equation*}\sum _{i =1}^{k}x_{i}F\genfrac{(}{)}{}{}{\sum _{i =1}^{k}x_{i}}{\sum _{i =1}^{k}y_{i}} \leq \sum _{i=1}^{k}x_{i}F\genfrac{(}{)}{}{}{x_{i}}{y_{i}}.
\end{equation*}
Let $z_{i} -z_{i -1} =x_{i} \geq 0$ and $d_{i} -d_{i -1} =y_{i} >0$. Note that $\sum _{i =1}^{k}x_{i} =z_{k}$ and $\sum _{i =1}^{k}y_{i} =d_{k}$ to conclude that inequality (\ref{Fundem ineq}) holds.

\textbf{Step 3.} The menu that maximizes the total transaction value is price-M (see the beginning of this section for a definition of price-M menus).

\textbf{Proof of Step 3.} Assume that $C$ is not price-M. Then there exists a price-quality pair $\{(p_{n+1},q_{n+1}) \}$ such that $p_{n+1} > p_{n}$ and $C \cup \{(p_{n+1},q_{n+1}) \}$ belongs to $\mathcal{C}_{p}$, i.e., $D_{i}(C)>0$ for all $1\leq i \leq n+1$. From Step 1, we have $m_{i}q_{i} -p_{i} =m_{i}q_{i-1} -p_{i -1}$ for all $i$ (recall that $q_{0} =p_{0} =0$). This implies that  

\begin{equation*}m_{i} =\frac{p_{i} -p_{i -1}}{q_{i} -q_{i-1}}.
\end{equation*}
for all $i$. We have 
\begin{align*}\pi(C \cup \{(p_{n+1},q_{n+1}) \}) - \pi (C) &  = \sum \limits _{i =1}^{n}p_{i}\left (F\left (m_{i +1}\right ) -F(m_{i})\right ) +   p_{n+1}(1-F(m_{n+1})) \\
 & - \sum \limits _{i =1}^{n-1}p_{i}\left (F\left (m_{i +1}\right ) -F(m_{i})\right ) - p_{n}(1-F(m_{n})) \\
 &  = p_{n} \left (F\left (\frac{p_{n+1}-p_{n}}{q_{n+1}-q_{n}} \right ) - F\left (\frac{p_{n}-p_{n-1}}{q_{n}-q_{n-1}} \right ) \right ) + p_{n+1} \left (1 - F\left (\frac{p_{n+1}-p_{n}}{q_{n+1}-q_{n}} \right ) \right ) \\
 &  - p_{n} \left (1 - F\left (\frac{p_{n}-p_{n-1}}{q_{n}-q_{n-1}} \right ) \right )  > 0.\end{align*}
 Thus, $C$ is not optimal. The inequality follows from the facts that $p_{n+1} > p_{n}$ and $D_{n+1} = 1 - F\left ((p_{n+1}-p_{n})/(q_{n+1}-q_{n}) \right ) > 0$. We conclude that the menu that maximizes the total transaction value (if it exists) is price-M.

\textbf{Step 4.} Let $C^{ \ast } =\{(p_{n} ,q_{n})\}$. We have\begin{equation*}\pi (C) \leq \pi (C^{ \ast }).
\end{equation*}  

\textbf{Proof of Step 4.} From Step 1 we have 
\begin{align*}\pi (C) &  =\sum \limits _{i =1}^{n}p_{i}\left (F\left (m_{i +1}\right ) -F(m_{i})\right ) \\
 &  =\sum \limits _{i =1}^{n -1}p_{i}\left (F\left ( \frac{p_{i +1} -p_{i}}{q_{i +1} -q_{i}} \right ) -F \left ( \frac {p_{i} -p_{i -1}}{q_{i} -q_{i-1}} \right ) \right ) +p_{n}\left (1 -F\genfrac{(}{)}{}{}{p_{n} -p_{n -1}}{q_{n} -q_{n-1}}\right ) \\
 &  =p_{n} -\sum \limits _{i =1}^{n} (p_{i} -p_{i -1}) F \left ( \frac{p_{i} -p_{i -1}}{q_{i} -q_{i-1}} \right ) .\end{align*}
 The first equality follows from Step 1. In the second equality we use the fact that $F(m_{n +1}) =F(b) =1$.

 Let $C^{ \ast } =\{(p_{n} ,q_{n})\}$. Using Step $1$ again we have 
 \begin{equation*}\pi \left (C^{ \ast }\right ) =p_{n}\left (1 -F\genfrac{(}{)}{}{}{p_{n}}{q_{n}}\right )
\end{equation*}
Thus, we have $\pi (C) \leq \pi (C^{ \ast })$ if and only if 
\begin{equation} \label{Ineq:pfT1}
p_{n}F\genfrac{(}{)}{}{}{p_{n}}{q_{n}} \leq \sum _{i =1}^{n}\left (p_{i} -p_{i -1}\right )F\genfrac{(}{)}{}{}{p_{i} -p_{i -1}}{q_{i} -q_{i-1}}.
\end{equation}
From Step 1, $m_{i} =\left (p_{i} -p_{i -1}\right )/\left (q_{i} -q_{i-1}\right ) \in [a ,b]$ for all $1 \leq i \leq n$. Thus, from Step 2, inequality (\ref{Ineq:pfT1}) holds. We conclude that $\pi (C) \leq \pi (C^{ \ast })$.

\textbf{Step 5.} We have $p^{M}(q) \geq p$
for every $1$-separating menu $\{ (p,q) \}$ that is maximal in $\mathcal{C}_{1}$. 

\textbf{Proof of Step 5.}
We first show that for any two $1$-separating menus $\{(p,q)\}$ and $\{(p',q')\}$ we have $p^{M}(q) \geq p^{M}(q')$ whenever $q \geq q' >0 $. 

 Because $F(m)m$ is strictly convex on $[a,b]$, $p^{M}(q)$ is single-valued. In addition, we clearly have $a \leq p^{M}(q)/q < b$. Hence, we have 
 \begin{equation*}
 \max _{p \geq 0} p\left (1-F \left ( \frac{p}{q} \right ) \right )  = \max _{ qa \leq p } p\left (1-F \left ( \frac{p}{q} \right ) \right ).  
 \end{equation*}

Assume in contradiction that $p^{M}(q) < p^{M}(q')$ and $q \geq q'$. Then $ p^{M}(q)/q< p^{M}(q')/q'$. The first order conditions for the optimality of $p^{M}$  and the fact that the strict convexity of $F(m)m$ on $[a,b]$ implies that the function $F(m)+mf(m)$ is strictly increasing on $[a,b]$ yield 
\begin{align*}
    0 & \geq 1 - \left (F \left (\frac{p^{M}(q)}{q} \right ) + \frac{p^{M}(q)}{q} f \left (\frac{p^{M}(q)}{q} \right )   \right )  \\
    & > 1 - \left (F \left (\frac{p^{M}(q')}{q'} \right ) + \frac{p^{M}(q')}{q'} f \left (\frac{p^{M}(q')}{q'} \right )   \right ) = 0 
\end{align*}
which is a contradiction. We conclude that $p^{M}(q) \geq p^{M}(q')$ whenever $q \geq q' >0 $.

 Let $\{(p^{H},q^{H})\} \in \mathcal{C}_{1}$ be such that $p^{H} \geq p'$ for all $\{(p',q')\} \in \mathcal{C}_{1}$. 
and let $\{(p,q)\}$ be a maximal element in $\mathcal{C}_{1}$. From the definition of $p^{H}$ we have $p^{H} \geq p$. Because $\{(p,q)\}$ is maximal in $\mathcal{C}_{1}$ we have $q \geq q^{H}$. Thus, we have $p^{M}(q) \geq p^{M}(q^{H})$.  Because $\mathcal{C}$ is regular we have  $p^{M}(q^{H})  \geq  p^{H}$. We conclude that 
\begin{equation*}
    p \leq p^{H} \leq p^{M}(q^{H}) \leq p^{M}(q)
\end{equation*}
which proves Step 5.

\textbf{Step 6.} There exists a $1$-separating menu $C^{ \prime } \in \mathcal{C}$ such that $\pi (C^{ \ast }) \leq \pi (C^{ \prime })$ where $C^{ \ast } =\{(p_{n} ,q_{n})\}$. 

\textbf{Proof of Step 6.} Because $\mathcal{C}$ is regular and $C =\{(p_{i} ,q_{i})_{i =1}^{n}\} \in \mathcal{C}_{p}$, there exists a $1$-separating menu $\{(p',q')\} \in \mathcal{C}_{1}$ such that $p' \geq p_{n}$ and $q' \geq q_{n}$. We consider two cases.

\textbf{Case 1}. $\{(p',q')\}$ is maximal in $\mathcal{C}_{1}$. 

From Step 5 we have $p' \leq p^{M}(q')$. We conclude that $p_{n} \leq p' \leq p^{M}(q')$. The convexity of $F(m)m$ on $[a,b]$ implies that $p \left (1 -F \left (\frac{p}{q' } \right ) \right )$ is increasing in $p$ on $[p_{n},p^{M}(q')]$. Thus,  
\begin{align*} p_{n} \left (1 -F \left (\frac{p_{n}}{q_{n}} \right ) \right ) 
 &  \leq  p_{n} \left (1 -F \left (\frac{p_{n}}{q'} \right ) \right ) \leq   p' \left (1 -F \left (\frac{p'}{q'} \right ) \right ). 
 \end{align*}
Thus, the menu $\{(p',q')\} \in \mathcal{C}_{1}$ yields more total transaction value than the menu  $\{(p_{n},q_{n})\}$.

 \textbf{Case 2}. $\{(p',q')\}$ is not maximal in $\mathcal{C}_{1}$.

In this case, because $\mathcal{C}_{1}$ is compact, there exists a menu $\{(p,q)\} \in \mathcal{C}_{1}$ such that $p \geq p'$ and $q \geq q'$, and  $\{(p,q)\}$ is maximal in $\mathcal{C}_{1}$. From Step 5 we have $p \leq p^{M}(q)$. 
 
 Hence, we have $p_{n} \leq p \leq p^{M}(q)$ which implies  
 \begin{align*} p_{n} \left (1 -F \left (\frac{p_{n}}{q_{n}} \right ) \right ) 
 &  \leq  p_{n} \left (1 -F \left (\frac{p_{n}}{q} \right ) \right ) \leq   p \left (1 -F \left (\frac{p}{q} \right ) \right ). 
 \end{align*}
That is, the menu $\{(p,q)\} \in \mathcal{C}_{1}$ yields more total transaction value than the menu  $\{(p_{n},q_{n})\}$. This proves Step 6.

Step 4 and Step 6 prove that for any menu $C =\{(p_{i} ,q_{i})_{i =1}^{n}\} \in \mathcal{C}$ there exists  a $1$-separating menu $C^{ \prime } \in \mathcal{C}$ such that $\pi (C) \leq \pi (C')$.  Thus, 
\begin{equation*}\sup _{C \in \mathcal{C}}\pi \left (C\right ) \leq \max _{C \in \mathcal{C}_{1}}\pi \left (C\right )
\end{equation*}
which proves the Theorem. The maximum on the right side of the last inequality is attained because the distribution function $F$ is continuous and $\mathcal{C}_{1}$ is a compact set.

From Case 2 in Step 6, for every $1$-separating menu $C$ that is not maximal in $\mathcal{C}_{1}$ there exists a $1$-separating menu that is maximal in $\mathcal{C}_{1}$ that yields more total transaction value than $C$. We conclude that the optimal $1$-separating menu is maximal in $\mathcal{C}_{1}$. 
\end{proof}

\begin{proof}[Proof of Proposition \ref{Prop: localA}]
Clearly $C \in \mathcal{C}_{p}$ implies $C' \in \mathcal{C}_{p}$.

Assume that $\mu(k) = n$.  From Step 1 in the proof of Theorem \ref{Theorem: Main} and using the fact that $\mu(k)=n$ we have 
\begin{align*}
    \pi(C) - \pi(C') & = p_{n} -\sum \limits _{i =1}^{n} (p_{i} -p_{i -1}) F \left ( \frac{p_{i} -p_{i -1}}{q_{i} -q_{i-1}} \right ) \\
    & - \left( p_{\mu(k)} -\sum \limits _{i = i}^{k} (p_{\mu(i)} -p_{\mu(i -1)}) F \left ( \frac{p_{\mu(i)} -p_{\mu(i-1)}}{q_{\mu(i)} -q_{\mu(i-1)}} \right )\right) \\
    & = -\sum \limits _{i =1}^{n} (p_{i} -p_{i -1}) F \left ( \frac{p_{i} -p_{i -1}}{q_{i} -q_{i-1}} \right ) + \sum \limits _{i = 1}^{k} (p_{\mu(i)} -p_{\mu(i -1)}) F \left ( \frac{p_{\mu(i)} -p_{\mu(i-1)}}{q_{\mu(i)} -q_{\mu(i-1)}} \right ).
    \end{align*}
    Let $j$ be such that $\mu(j) - \mu(j-1) = d > 1$ and assume that $F(m)m$ is convex on $[m_{\mu(j-1)+1}(C),m_{\mu(j)}(C)] $.  From Step 2 in the proof of Theorem \ref{Theorem: Main} the function $f(x ,y) =xF \left (\frac{x}{y} \right )$ is convex on $E=\{\left (x ,y\right ) :x/y \in [m_{\mu(j-1)+1}(C),m_{\mu(j)}(C)] ,$ $y >0\}$. Hence, using Jensen's inequality with the points $(x_{i},y_{i})=(p_{i} - p_{i-1},q_{i}-q_{i-1}) \in E$ for $i=\mu(j-1)+1,\ldots,\mu(j)$ yields 
    $$\sum _{i=\mu (j-1)+1} ^{\mu (j)} d^{-1}f(x_{i},y_{i}) \geq f \left ( d^{-1}\sum _{i=\mu (j-1)+1} ^{\mu (j)} x_{i}, d^{-1} \sum _{i=\mu (j-1)+1} ^{\mu (j)} y_{i} \right ),$$ 
    i.e., 
    $$   \sum \limits _{i =\mu(j-1)+1}^{\mu(j)} (p_{i} -p_{i -1}) F \left ( \frac{p_{i} -p_{i -1}}{q_{i} -q_{i-1}} \right ) \geq  
    (p_{\mu(j)} -p_{\mu(j-1)}) F \left ( \frac{p_{\mu(j)} -p_{\mu(j-1)}}{q_{\mu(j)} -q_{\mu(j-1)}} \right ). $$
   Summing the last inequality over all $j$ such that $\mu(j) - \mu(j-1) > 1$ shows that $\pi(C') \geq \pi(C)$. The case where $F(m)m$ is concave is proven by an analogous argument. 
    \end{proof}

    \begin{proof}[Proof of Proposition \ref{Prop: localB}]\label{Proof of Prop: localB}
    Assume that $F(m)m$ is convex on $[m_{1}(C),m_{n}(C)]$ and $m_{n}(C') \leq p^{M} (1)$. From Proposition \ref{Prop: localA} the menu $C'' = \{(p_{n} , q_{n} )\}$ yields more revenue than the menu $C$. We will now show that the menu $C'$ yields more revenue than the menu $C''$.
    
    Differentiating $\pi (C'')$ with respect to $p_{n}$ yields 
 $1 - f(m_{n}(C''))m_{n}(C'')-F(m_{n}(C'')$. Note that $m_{n}(C'') \leq m_{n}(C') \leq  p^{M}(1)$. 
 
 From the first order condition we have $1 - f(p^{M}(1))p^{M}(1) - F(p^{M}(1)) = 0$. Quasi-concavity of $m(1-F(m))$ implies that $1 - f(m_{n}(C''))m_{n}(C'')-F(m_{n}(C'')) \geq 0$ as $m_{n}(C'') \leq p^{M}(1)$. Because $m_{n}(C') \leq  p^{M}(1)$ we conclude that the partial derivative of $\pi(C'')$ with respect to $p_{n}$ is positive on  $[p_{n},p_{n}']$.  Hence, $\pi(C') \geq \pi(C'')$ which concludes the proof. The concave case is proven in a similar manner and is therefore omitted. 
    \end{proof}

\begin{proof}[Proof of Proposition \ref{Prop: Converse}]
We prove the result for $k=2$. The proof follows from the same arguments for $k > 2$. Suppose that $g\left (z\right ) =F(z)z$ is not convex on $(a ,b)$. Then there exist non-negative numbers $z_{1} \in (a ,b)$, $z_{2} \in (a ,b)$ and $0 <\lambda  <1$ such that \begin{equation*}g\left (\lambda z_{1} +\left (1 -\lambda \right )z_{2}\right ) >\lambda g\left (z_{1}\right ) +\left (1 -\lambda \right )g\left (z_{2}\right ) .
\end{equation*}Let $k_{1} ,k_{2} ,d_{1} ,d_{2}$, and $0 <\theta  <1$ be such that $k_{1} \geq 0$, $k_{2} \geq 0$, $d_{1} >0$, $d_{2} >0$, $d_{1}z_{1} =k_{1}$, $d_{2}z_{2} =k_{2}$, and $\theta d_{1} =\lambda \left (\theta d_{1} +\left (1 -\theta \right )d_{2}\right )$. 

Note that $1 -\lambda  =\left (1 -\theta \right )d_{2}/\left (\theta d_{1} +\left (1 -\theta \right )d_{2}\right )$. 

Denote $d_{\theta } : =\theta d_{1} +\left (1 -\theta \right )d_{2}$ and $k_{\theta } : =\theta k_{1} +\left (1 -\theta \right )k_{2}$. Note that \begin{equation*}\lambda z_{1} +\left (1 -\lambda \right )z_{2} =\frac{\theta d_{1}}{d_{\theta }}\frac{k_{1}}{d_{1}} +\frac{\left (1 -\theta \right )d_{2}}{d_{\theta }}\frac{k_{2}}{d_{2}} =\frac{k_{\theta }}{d_{\theta }} .
\end{equation*}

We have 
\begin{equation*}\theta d_{1}g\genfrac{(}{)}{}{}{k_{1}}{d_{1}} +\left (1 -\theta \right )d_{2}g\genfrac{(}{)}{}{}{k_{2}}{d_{2}} =d_{\theta }\left (\frac{\theta d_{1}}{d_{\theta }}g\genfrac{(}{)}{}{}{k_{1}}{d_{1}} +\frac{\left (1 -\theta \right )d_{2}}{d_{\theta }}g\genfrac{(}{)}{}{}{k_{2}}{d_{2}}\right ) <d_{\theta }g\genfrac{(}{)}{}{}{k_{\theta }}{d_{\theta }} .
\end{equation*}We conclude that the function $f(x ,y) : =yg\genfrac{(}{)}{}{}{x}{y} =xF\genfrac{(}{)}{}{}{x}{y}$ is not convex on $E^{ \ast } =\{\left (x ,y\right ) :x/y \in (a ,b) ,y >0\}$.

 Since $f$ is continuous and not convex it is not midpoint convex.\protect\footnote{
Recall that the function $f :E^{ \ast } \rightarrow \mathbb{R}$ is midpoint convex if for all $e_{1} ,e_{2} \in E^{ \ast }$ we have $f\left (\left (e_{1} +e_{2}\right )/2\right ) \leq \left (f(e_{1}) +f(e_{2})\right )/2$. A continuous midpoint convex function is convex. We conclude that $f$ is not midpoint convex.
}

 Thus, there exists $\left (x_{1} ,y_{1}\right ) \in E^{ \ast }$ and $\left (x_{2} ,y_{2}\right ) \in E^{ \ast }$ such that 
 \begin{equation} \label{Ineq3}
 f\left (\frac{\left (x_{1} ,y_{1}\right )}{2} +\frac{\left (x_{2} ,y_{2}\right )}{2}\right ) >\frac{f\left (x_{1} ,y_{1}\right )}{2} +\frac{f\left (x_{2} ,y_{2}\right )}{2}. 
\end{equation} 
If $x_{1} =x_{2} =0$ then the left-hand-side and the right-hand-side of the last inequality equal $0$ which is a contradiction, so we have $x_{1} +x_{2} >0$. 

Assume in contradiction that $\frac{x_{2}}{y_{2}} =\frac{x_{1}}{y_{1}}$. We have
\begin{align*}f\left (\frac{\left (x_{1} ,y_{1}\right )}{2} +\frac{\left (x_{2} ,y_{2}\right )}{2}\right ) >\frac{f\left (x_{1} ,y_{1}\right )}{2} +\frac{f\left (x_{2} ,y_{2}\right )}{2} \\
 \Leftrightarrow \left (x_{1} +x_{2}\right )F\genfrac{(}{)}{}{}{x_{1} +x_{2}}{y_{1} +y_{2}} >x_{1}F\genfrac{(}{)}{}{}{x_{1}}{y_{1}} +x_{2}F\genfrac{(}{)}{}{}{x_{2}}{y_{2}} \\
 \Leftrightarrow F\genfrac{(}{)}{}{}{x_{1} +x_{2}}{y_{1} +y_{2}} >F\genfrac{(}{)}{}{}{x_{1}}{y_{1}} \\
 \Rightarrow \frac{x_{1} +x_{2}}{y_{1} +y_{2}} >\frac{x_{1}}{y_{1}} \Leftrightarrow \frac{x_{2}}{y_{2}} >\frac{x_{1}}{y_{1}} ,\end{align*}which is a contradiction. Thus, $\frac{x_{2}}{y_{2}} \neq \frac{x_{1}}{y_{1}}$.

Assume without loss of generality that $\frac{x_{2}}{y_{2}} >\frac{x_{1}}{y_{1}}$. Then $x_{2} >0$. 

 Let $p_{2} >p_{1}$ and $q_{2} >q_{1}$ be such that $p_{2} - p_{1}=x_{2} >0$, $p_{1}=x_{1}$, $q_{2} - q_{1} = y_{2}$ and  $y_{1} =q_{1}$. Define the menus $C =\{\left (p_{1} ,q_{1}\right ) ,\left (p_{2} ,q_{2}\right )\}$,  $C^{ \ast } =\{\left (p_{1} ,q_{1}\right )\}$, and $C^{ \ast  \ast } =\{\left (p_{2} ,q_{2}\right )\}$. Let $\mathcal{C} =\{C ,C^{ \ast } ,C^{ \ast  \ast }\}$. We now show that $D_{1}(C) >0$, $D_{2}\left (C\right ) >0$ and that $C$ yields more total transaction value than the $1$-separating menus $C^{ \ast }$ and $C^{ \ast  \ast }$.

Note that $\frac{x_{2}}{y_{2}} >\frac{x_{1}}{y_{1}}$ implies
\begin{equation*} m_{2} =\frac{p_{2} -p_{1}}{q_{2} -q_{1}} >\frac{p_{1}}{q_{1}} =m_{1}
\end{equation*}
where $m_{1}$ and $m_{2}$ are defined in Step 1 in the proof of Theorem \ref{Theorem: Main}. 

Since $F$ is supported on $[a ,b]$, $F$ is strictly increasing on $[a ,b]$. Note that $m_{1}$ and $m_{2}$ belong to $(a ,b)$ so $m_{2} >m_{1}$ implies that $F(m_{2}) >F(m_{1})$. We have $D_{1}(C) =F(m_{2}) -F(m_{1}) >0$. In addition, because $m_{2} =x_{2}/y_{2}$ and $\left (x_{2} ,y_{2}\right ) \in E^{ \ast }$ we have $m_{2} <b$, so $D_{2}\left (C\right ) =1 -F(m_{2}) >0$.

Inequality (\ref{Ineq3}) implies that
\begin{equation*}p_{2}F\genfrac{(}{)}{}{}{p_{2}}{q_{2}} >\left (p_{2} -p_{1}\right )F\genfrac{(}{)}{}{}{p_{2} -p_{1}}{q_{2} -q_{1}} +p_{1}F\genfrac{(}{)}{}{}{p_{1}}{q_{1}} .
\end{equation*}
Because $D_{1}(C) >0$ and $D_{2}\left (C\right ) >0$, from Step 4 in the proof of Theorem \ref{Theorem: Main}, the last inequality implies $\pi (C) >\pi (C^{ \ast \ast })$ where $C^{ \ast  \ast } =\{\left (p_{2} ,q_{2}\right )\}$. 

The menu $C^{ \ast } =\{\left (p_{1} ,q_{1}\right )\}$ does not maximize the total transaction value  because\begin{equation*}\pi (C^{  \ast }) =p_{1}\left (1 -F\genfrac{(}{)}{}{}{p_{1}}{q_{1}}\right ) <p_{2}\left (1 -F\genfrac{(}{)}{}{}{p_{2} -p_{1}}{q_{2} -q_{1}}\right ) +p_{1}\left (F\genfrac{(}{)}{}{}{p_{2} -p_{1}}{q_{2} -q_{1}} -F\genfrac{(}{)}{}{}{p_{1}}{q_{1}}\right ) =\pi \left (C\right )
\end{equation*}
where the equalities follow from Step 1 in the proof of Theorem \ref{Theorem: Main}. 

We conclude that the $2$-separating menu $C$ yields more total transaction value than the $1$-separating menus $C^{ \ast }$ and $C^{ \ast  \ast }$.   
\end{proof}

\begin{proof}[Proof of Proposition \ref{prop:converse-regular}]
(i) Let $q_{2} >q_{1}$, $ bq_{2} > p_{2} >p_{1}$,  and $C =\{\left (p_{1} ,q_{1}\right ) ,\left (p_{2} ,q_{2}\right )\} \in \mathcal{C}_{2}$,  $C^{ \ast } =\{\left (p_{3} ,q_{2}\right )\} \in \mathcal{C}_{1}$ where $p_{3} = (b  - \epsilon)q_{2} $ for some $\epsilon>0$ . Let $\mathcal{C} =\{C ,C^{ \ast } \}$. Then we can find a small $\epsilon>0$ such that  condition (i) of Definition \ref{Def: marketN} is satisfied and
\begin{equation*}\pi (C^{ \ast   }) =p_{3}\left (1 -F\genfrac{(}{)}{}{}{p_{3}}{q_{2}}\right ) <p_{2}\left (1 -F\genfrac{(}{)}{}{}{p_{2} -p_{1}}{q_{2} -q_{1}}\right ) +p_{1}\left (F\genfrac{(}{)}{}{}{p_{2} -p_{1}}{q_{2} -q_{1}} -F\genfrac{(}{)}{}{}{p_{1}}{q_{1}}\right ) =\pi \left (C\right ).
\end{equation*}

(ii) Let $p_{2} >p_{1}$ and $q_{2} >q_{1}$ by such that $p^{M}_{1} > p_{1}$ and  $C =\{\left (p_{1} ,q_{1}\right ) ,\left (p_{2} ,q_{2}\right )\} \in \mathcal{C}_{1}$, $C^{ \ast } =\{\left (p_{1} ,q_{1}\right )\} \in \mathcal{C}_{2}$. Let $\mathcal{C} =\{C ,C^{ \ast } \}$. Then condition (ii) of Definition \ref{Def: marketN} is satisfied and
 \begin{equation*}\pi (C^{ \ast   }) =p_{1}\left (1 -F\genfrac{(}{)}{}{}{p_{1}}{q_{1}}\right ) <p_{2}\left (1 -F\genfrac{(}{)}{}{}{p_{2} -p_{1}}{q_{2} -q_{1}}\right ) +p_{1}\left (F\genfrac{(}{)}{}{}{p_{2} -p_{1}}{q_{2} -q_{1}} -F\genfrac{(}{)}{}{}{p_{1}}{q_{1}}\right ) =\pi \left (C\right )
\end{equation*}
which completes the proof. 
\end{proof}

   \begin{proof} [Proof of Proposition \ref{Corr:convex-concave}]
 Suppose that $m^{*} \in [m_{n-k}(C),m_{n-k+1}(C)].$  Let $C^{*}$ be the optimal menu. Then $(p_{n},q_{n}) \in C^{*}$ (see Step 3 in the proof of Theorem \ref{Theorem: Main}). Assume in contradiction that $C^{*}$ is a $(k+d)$-separating menu for $d \geq 2$.  Consider the $k+1$-separating menu $C'$ that consists of the $k+1$ highest price-quality pairs in $C^{*}$. Then Proposition \ref{Prop: localA} implies that $C'$ yields more revenue than $C^{*}$ if $F(m)m$ is convex on $[a,m_{d}(C^{*})]$. 
 
 We will now show that $m_{d}(C^{*}) \leq m_{n-k}(C)$. First note that $d \leq n-k$ so it is enough to show that $m_{d}(C^{*}) \leq m_{d}(C)$. Let $x_{i}/y_{i}$ be an increasing sequence of numbers. Then $x_{k}/y_{k} \geq \sum _{j \in W_{k}} x_{j}/ \sum _{j \in W_{k}} y_{j}$ for any $k$ and $W_{k} \subseteq \{1,\ldots,k\}$. Letting $x_{i} = p_{i}-p_{i-1}$, $y_{i} = q_{i}-q_{i-1}$ we note that $x_{i}/y_{i}$ is increasing and we can find a set $W_{d} \subseteq \{1,\ldots,d\}$ such that $ \sum _{j \in W_{d}} x_{j}/ \sum _{j \in W_{d}} y_{j} = m_{d}(C^{*})$. Hence,  $m_{d}(C) = x_{d}/y_{d} \geq  \sum _{j \in W_{d}} x_{j}/ \sum _{j \in W_{d}} y_{j} = m_{d}(C^{*})$ for a suitable $W_{d}$. 
 
 An analogous argument shows that if $C^{*}$ is $(k-d)$-separating for some $d \geq 1$ we can find a $k$-separating menu $C''$ that yields more revenue than $C^{*}$. We conclude that the optimal menu is $k$-separating or $k+1$-separating. 
\end{proof}

\subsection{Proofs of Section \ref{Sec: Model 1 quant}}

We first prove the following Lemma: 
 \begin{lemma} \label{Lem: expected sellers}
 Fix an information structure $I =\{B_{1} ,B_{2} ,\ldots  ,B_{n}\}$ in $\mathbb{I}(I_{o})$. Then, for every positive pricing function $\boldsymbol{p}$ we have \begin{equation*}\mathbb{E}_{\lambda _{B_{i}}}(X)  =\frac{\int _{B_{i}}x(k(x))^{ -1/\alpha}\phi (dx)}{\int _{B_{i}}(k(x))^{ -1/\alpha}\phi (dx)}.
\end{equation*} 
The probability measure $\lambda_{B_{i}}$ is given in Equation (\ref{Eq: lambda}) in Section \ref{Sec: Model 1 quant}. 
That means the expected sellers' qualities do not depend on the prices. 
  \end{lemma}
 
\begin{proof}[Proof of Lemma \ref{Lem: expected sellers}]
Fix an information structure $I =\{B_{1} ,B_{2} ,\ldots  ,B_{n}\}$ in $\mathbb{I}(I_{o})$. 

Given a positive pricing function $\boldsymbol{p}$, the optimal quantity of a seller $x$ in $B_{i}$, $g(x ,p(B_{i})) =\ensuremath{\operatorname*{argmax}}_{h \in \mathbb{R}_+}U(x ,h ,p(B_{i}))$ is given by
\begin{equation} \label{supply function}  
g(x,p(B_{i})) =\left (\frac{p(B_{i})}{k(x)}\right )^{1/\alpha }.
\end{equation}
Hence, we have
\begin{equation*}\mathbb{E}_{\lambda _{B_{i}}}(X) =\int _{B_{i}}x\lambda _{B_{i}}(dx) =\frac{\int _{B_{i}}xg(x ,p(B_{i}))\phi (dx)}{\int _{B_{i}}g(x ,p(B_{i}))\phi (dx)} =\frac{\int _{B_{i}}x(k(x))^{ -1/\alpha}\phi (dx)}{\int _{B_{i}}(k(x))^{ -1/\alpha}\phi (dx)}.
\end{equation*}
Thus, the expected sellers' quality $\mathbb{E}_{\lambda _{B_{i}}}(X)$ does not depend on the prices when the pricing function is positive.  
\end{proof}
 
 \begin{proof}[Proof of Proposition \ref{Prop: unique equilibrium}] 
For the rest of the proof except for Step 3, we fix an information structure $I =\{B_{1} ,B_{2} ,\ldots  ,B_{n}\}$ in $\mathbb{I}(I_{o})$ and assume that $\mathbb{E}_{\lambda _{B_{1}}}(X) < \ldots  < \mathbb{E}_{\lambda _{B_{n}}}(X)$ where the expected sellers' quality $\mathbb{E}_{\lambda _{B_{i}}}(X)$ is given in Lemma \ref{Lem: expected sellers}. 

Let $\boldsymbol{P}$ be the set of all pricing functions such that the demand for each set $B_{i} \in I$, $D_{I}(B_{i},\boldsymbol{p})$ is greater than $0$, each price is greater than $0$, and the prices are ordered according to an ascending order.
That is, 
\begin{equation*}\boldsymbol{P} =\{\boldsymbol{p} \in \mathbb{R}_{+}^{n} :D_{I}(B_{i} ,\boldsymbol{p}) > 0 \text{ for all }  i =1 ,\ldots  ,n, 0 < p(B_{1})<\ldots<p(B_{n})\} .
\end{equation*}

To simplify notation, for the rest of the proof we denote $p_{i} =p(B_{i})$, $p^{\prime}_{i} =p^{\prime}(B_{i})$, $s_{i}(p_{i}) =S_{I}(B_{i} ,p(B_{i}))$, $\mathbb{E}_{\lambda _{B_{i}}}(X)=q_{i}$, and $d_{i}(\boldsymbol{p}) =D_{I}(B_{i} ,\boldsymbol{p})$. Note that $\boldsymbol{p} \in \boldsymbol{P}$ implies $0<q_{1}<\ldots < q_{n}$ (recall that Lemma \ref{Lem: expected sellers} implies that the expected sellers' quality $q_{i}$ does not depend on the prices).  


Define the function $\psi :\boldsymbol{P} \rightarrow \mathbb{R}$ by  
\begin{equation}\psi (\boldsymbol{p}) =\sum _{i =1}^{n}\frac{p_{i}^{\frac{\alpha  +1}{\alpha }}\int _{B_{i}}k(x)^{-1/\alpha }\phi (dx)}{(1 +1/\alpha )} - p_{n} +\sum _{i =0}^{n -1}F_{2}\genfrac{(}{)}{}{}{p_{i +1} -p_{i}}{q_{i +1} -q_{i}}(q_{i +1} -q_{i})
\end{equation}
where $F_{2}(x) =\int _{a}^{x}F(m)dm$ is the antiderivative of $F$ and $q_{0}=p_{0}=0$. Note that $\boldsymbol{p} \in \boldsymbol{P}$ implies that for every $1 \leq i \leq n-1$ we have $a \leq (p_{i+1}-p_{i})/(q_{i+1}-q_{i}) \leq b$ (see Step 1 in the proof of Theorem \ref{Theorem: Main}). Because the function $F$ is continuous, the fundamental theorem of calculus implies that the function $F_{2}$ is differentiable and $F_{2}^{ \prime } =F$.  Thus, $\psi $ is continuously differentiable. 

Let $\nabla \psi$ be the gradient of $\psi$ and let $\nabla _{i}\psi$ be the $i$th element of the gradient. A direct calculation shows that for $1 \leq i \leq n-1$ we have
\begin{align*} \nabla _{i}\psi (\boldsymbol{p}) &  = p_{i}^{1/\alpha }\int _{B_{i}}k(x)^{-1/\alpha} \phi (dx) -F_{2}^{ \prime }\genfrac{(}{)}{}{}{p_{i +1} -p_{i}}{q_{i +1} -q_{i}} +F_{2}^{ \prime }\genfrac{(}{)}{}{}{p_{i} -p_{i -1}}{q_{i} -q_{i -1}} \\
 & =  p_{i}^{1/\alpha }\int _{B_{i}}k(x)^{-1/\alpha} \phi (dx) -F\genfrac{(}{)}{}{}{p_{i +1} -p_{i}}{q_{i +1} -q_{i}} +F\genfrac{(}{)}{}{}{p_{i} -p_{i -1}}{q_{i} -q_{i -1}} \\
 &  =s_{i}(p_{i}) -d_{i}(\boldsymbol{p}) .\end{align*} 
 The last equality follows from Step 1 and Step 4 in the proof of Theorem 1, the fact that $\boldsymbol{p} \in \boldsymbol{P}$, and Equation (\ref{supply function}) (see the proof of Lemma \ref{Lem: expected sellers}).  Similarly, 
 \begin{equation*} \nabla _{n}\psi (\boldsymbol{p}) = p_{n}^{1/\alpha }\int _{B_{n}}k(x)^{-1/\alpha} \phi (dx) -1 + F\genfrac{(}{)}{}{}{p_{n} - p_{n-1}}{q_{n} - q_{n -1}} =s_{n}(p_{i}) -d_{n}(\boldsymbol{p}).
\end{equation*}
Thus, the excess supply function is given by $\nabla \psi (\boldsymbol{p}) =\left (\nabla _{1}\psi(\boldsymbol{p}) ,\ldots  ,\nabla _{n}\psi(\boldsymbol{p})\right )$ where $\nabla _{i}\psi(\boldsymbol{p})=s_{i}(p_{i}) -d_{i}(\boldsymbol{p})$
for all $i$ from $1$ to $n$. Note that $\nabla \psi (\boldsymbol{p})=0$ implies that $(I,\boldsymbol{p})$ is implementable.  

Our goal is to prove that $(I,\boldsymbol{p})$ is implementable if and only if $\boldsymbol{p}$ is the unique minimizer of $\psi$.    
 To show that $\psi$ has at most one minimizer we prove that $\psi$ is strictly convex on the convex set $\boldsymbol{P}$. We proceed with the following steps:
 
 \textbf{Step 1.} The set $\boldsymbol{P}$ is bounded, convex and open in $\mathbb{R}^{n}$.   
 
 \textbf{Proof of Step 1.}
We first show that $\boldsymbol{P}$ is bounded. Let $\overline{p}=q_{n}b$ and let  $\boldsymbol{p}=(p_{1},\ldots,p_{n})$ be a vector such that $p_{i} > \overline{p}$ for some $1 \leq i \leq n$. Then
\begin{equation*}
mq_{i}-p_{i} \leq bq_{n} - p_{i} < bq_{n} - \overline{p}. 
\end{equation*}
Hence $d_{i}(\boldsymbol{p})=0$. That is, $\boldsymbol{p}$ does not belong to $\boldsymbol{P}$. We conclude that $(\overline{p},\ldots, \overline{p})$ is an upper bound of $\boldsymbol{P}$ under the standard product order on $\mathbb{R}^{n}$. Clearly, $\boldsymbol{P}$ is bounded from below. Hence, $\boldsymbol{P}$ is bounded.   

We now show that the set $\boldsymbol{P}$ is convex in $\mathbb{R}^{n}$. Let $\boldsymbol{p}$,$\boldsymbol{p}^{ \prime } \in \boldsymbol{P}$ and $0<\lambda < 1$. 

We need to show that $\lambda \boldsymbol{p} + (1-\lambda)\boldsymbol{p}^{\prime} \in \boldsymbol{P}$. First note that 
\begin{equation*}
    0< \lambda p_{1} + (1-\lambda)p_{1}^{\prime} < \ldots < \lambda p_{n} + (1-\lambda)p_{n}^{\prime}
\end{equation*}
so we only need to show that $d_{i}(\lambda \boldsymbol{p} + (1-\lambda)\boldsymbol{p}^{\prime}) > 0$ for all $i=1,\ldots,n$. 
Let $1\leq i \leq n-1$. Because $d_{i}(\boldsymbol{p})>0$ and $d_{i}(\boldsymbol{p^{\prime}})>0$ we have $F\genfrac{(}{)}{}{}{p_{i +1} -p_{i}}{q_{i +1} -q_{i}} - F\genfrac{(}{)}{}{}{p_{i} -p_{i -1}}{q_{i} -q_{i -1}}>0$ and $F\genfrac{(}{)}{}{}{p^{\prime}_{i +1} -p^{\prime}_{i}}{q_{i +1} -q_{i}} - F\genfrac{(}{)}{}{}{p^{\prime}_{i} -p^{\prime}_{i -1}}{q_{i} -q_{i -1}}>0$. Strict monotonicity of $F$ on its support implies $\frac{p_{i +1} -p_{i}}{q_{i +1} -q_{i}}>\frac{p_{i} -p_{i -1}}{q_{i} -q_{i -1}}$ and $\frac{p^{\prime}_{i +1} -p^{\prime}_{i}}{q_{i +1} -q_{i}} > \frac{p^{\prime}_{i} -p^{\prime}_{i -1}}{q_{i} -q_{i -1}}$. Hence, 
\begin{equation*}
    \frac{\lambda p_{i +1} + (1-\lambda)p^{\prime}_{i +1} -(\lambda p_{i} + (1-\lambda)p^{\prime}_{i})}{q_{i +1} -q_{i}} >  \frac{\lambda p_{i} + (1-\lambda)p^{\prime}_{i} -(\lambda p_{i-1} + (1-\lambda)p^{\prime}_{i-1})}{q_{i } -q_{i-1}}.
\end{equation*}
Using again the strict monotonicity of $F$ we conclude that 
\begin{equation*}
    F\left ( \frac{\lambda p_{i +1} + (1-\lambda)p^{\prime}_{i +1} -(\lambda p_{i} + (1-\lambda)p^{\prime}_{i})}{q_{i +1} -q_{i}} \right) -  
    F\left( \frac{\lambda p_{i} + (1-\lambda)p^{\prime}_{i} -(\lambda p_{i-1} + (1-\lambda)p^{\prime}_{i-1})}{q_{i } -q_{i-1}} \right ) > 0.
\end{equation*}
That is, $d_{i}(\lambda \boldsymbol{p} + (1-\lambda)\boldsymbol{p}^{\prime}) > 0$. Similarly we can show that  $d_{n}(\lambda \boldsymbol{p} + (1-\lambda)\boldsymbol{p}^{\prime}) > 0$. Thus, $\boldsymbol{P}$ is a convex set.

Because $d_{i}(\boldsymbol{p})$ is continuous on $\boldsymbol{P}$ for all $1 \leq i \leq n$, it is immediate that the set  $\boldsymbol{P}$ is an open set in $\mathbb{R}^{n}$.

\textbf{Step 2}. The function $\psi$ is strictly convex on $\boldsymbol{P}$. 

\textbf{Proof of Step 2.} We claim that $\nabla \psi$ is strictly monotone on $\boldsymbol{P}$, i.e., for all $\boldsymbol{p} =(p_{1} ,\ldots  ,p_{n})$ and $\boldsymbol{p}^{\prime } =(p^{ \prime }_{1} ,\ldots ,p^{ \prime }_{n})$ that belong to $\boldsymbol{P}$ and satisfy $\boldsymbol{p} \neq \boldsymbol{p}^{ \prime }$, we have 
\begin{equation*}\left \langle \nabla \psi (\boldsymbol{p}) -\nabla \psi (\boldsymbol{p}^{ \prime }) ,\boldsymbol{p} -\boldsymbol{p}^{ \prime }\right \rangle  > 0
\end{equation*} 
where $\left \langle \boldsymbol{x} ,\boldsymbol{y}\right \rangle  : =\sum _{i =1}^{n}x_{i}y_{i}$ denotes the standard inner product between two vectors $\boldsymbol{x}$ and $\boldsymbol{y}$ in $\mathbb{R}^{n}$. 
Because $\boldsymbol{P}$ is a convex set it is well known that $\nabla \psi$ is strictly monotone on $\boldsymbol{P}$ if and only if $\psi$ is strictly convex on $\boldsymbol{P}$. 

Let $\boldsymbol{p}$,$\boldsymbol{p}^{ \prime } \in \boldsymbol{P}$ and assume that $\boldsymbol{p} \neq \boldsymbol{p}^{ \prime }$. 

 Because $g$ is strictly increasing in $p_{i}$, $k$ is a positive function,  and  $\phi(B_{i})>0$, the supply function  $s_{i}(p_{i}) =  p_{i}^{1/\alpha }\int _{B_{i}}k(x)^{-1/\alpha} \phi (dx)$ is strictly increasing in the price $p_{i}$. Thus, $s_{i}(p_{i}) >s_{i}(p_{i}^{ \prime })$ if and only if $p_{i} >p_{i}^{ \prime }$. Combining the last inequality with the fact that $\boldsymbol{p} \neq \boldsymbol{p}^{ \prime }$ implies
 \begin{equation*}\sum _{i =1}^{n}(p_{i} -p_{i}^{ \prime })(s_{i}(p_{i}) -s_{i}(p_{i}^{ \prime })) >0.
\end{equation*}

Let $p_{0}=p^{\prime}_{0}=0$. We have 

\begin{align*} \sum _{i =1}^{n}(p_{i} -p_{i}^{ \prime })(d_{i}(\boldsymbol{p}) -d_{i}(\boldsymbol{p}^{ \prime })) &  = \sum _{i =1}^{n-1}(p_{i} -p_{i}^{ \prime }) \left (F\left ( \frac{p_{i +1} -p_{i}}{q_{i +1} -q_{i}} \right ) -F \left ( \frac {p_{i} -p_{i -1}}{q_{i} -q_{i-1}} \right ) \right )  \\
& -  \sum _{i =1}^{n-1}(p_{i} -p_{i}^{ \prime })\left (F\left ( \frac{p^{\prime}_{i +1} -p^{\prime}_{i}}{q_{i +1} -q_{i}} \right ) -F \left ( \frac {p^{\prime}_{i} -p^{\prime}_{i -1}}{q_{i} -q_{i-1}} \right ) \right ) \\ 
& + (p_{n}-p^{\prime}_{n})\left (F\genfrac{(}{)}{}{}{p^{\prime}_{n} -p^{\prime}_{n -1}}{q_{n} -q_{n-1}} -F\genfrac{(}{)}{}{}{p_{n} -p_{n -1}}{q_{n} -q_{n-1}}\right ) \\
 & =    \sum _{i =1}^{n}(p_{i} - p_{i-1} -(p_{i}^{ \prime }- p_{i-1}^{\prime}))\left (F\left ( \frac{p^{\prime}_{i} -p^{\prime}_{i-1}}{q_{i} -q_{i-1}} \right ) -F \left ( \frac {p_{i} -p_{i -1}}{q_{i} -q_{i-1}} \right ) \right )  \\ 
& \leq 0.
 \end{align*}
The last inequality follows from the monotonicity of $F$. Thus, 
\begin{align*} \left \langle \nabla \psi (\boldsymbol{p}) -\nabla \psi (\boldsymbol{p}^{ \prime }) ,\boldsymbol{p} -\boldsymbol{p}^{ \prime }\right \rangle & =\sum _{i =1}^{n}(s_{i}(p_{i}) -d_{i}(\boldsymbol{p})) -(s_{i}(p_{i}^{ \prime }) -d_{i}(\boldsymbol{p}^{ \prime }))(p_{i} -p_{i}^{ \prime }) \\
 & =\sum _{i =1}^{n}(p_{i} -p_{i}^{ \prime })(s_{i}(p_{i}) -s_{i}(p_{i}^{ \prime })) - \sum _{i =1}^{n}(p_{i} -p_{i}^{ \prime })(d_{i}(\boldsymbol{p}) -d_{i}(\boldsymbol{p}^{ \prime })) \\ & > 0. \end{align*}
We conclude that $\nabla \psi$ is strictly monotone on the convex set $\boldsymbol{P}$. Hence, $\psi$ is strictly convex on $\boldsymbol{P}$.

\textbf{Step 3.} $(I,\boldsymbol{p})$ is implementable if and only if $\boldsymbol{p}$ is the unique minimizer of $\psi$.   

\textbf{Proof of Step 3.} Suppose that $(I,\boldsymbol{p})$ is implementable where $I =\{B_{1} ,B_{2} ,\ldots  ,B_{n}\}$ and  $\boldsymbol{p}=(p(B_{1}),\ldots,p(B_{n}))$. Let $D =\{D_{I}(B_{i} ,\boldsymbol{p})\}_{B_{i} \in I}$, $S =\{S(B_{i},p(B_{i})\}_{B_{i} \in I}$, and $\lambda  =\{\lambda _{B_{i}}\}_{B_{i} \in I}$ be an equilibrium under $(I,\boldsymbol{p})$.

 Because $(I,\boldsymbol{p})$ is implementable we have $p(B_{i})>0$ for all $B_{i} \in I$ and 
\begin{equation*}
D_{I}(B_{i} ,\boldsymbol{p})=S_{I}(B_{i} ,p(B_{i})) = \int _{B_{i}}g(x ,p(B_{i}))\phi (dx) >0 
\end{equation*}
where the last inequality follows because $g$ is positive (see the proof of Lemma \ref{Lem: expected sellers}) and $\phi(B_{i})>0$. We can assume without loss of generality that $\mathbb{E}_{\lambda _{B_{1}}}(X) < \ldots  < \mathbb{E}_{\lambda _{B_{n}}}(X)$. To see this, note that
 if $\mathbb{E}_{\lambda _{B_{i}}}(X) = \mathbb{E}_{\lambda _{B_{j}}}(X)$ for some $i<j$ then $\min \{D_{I}(B_{i} ,\boldsymbol{p}), D_{I}(B_{j} ,\boldsymbol{p})\}=0$ which contradicts the implementability of  $(I,\boldsymbol{p})$. Thus, relabeling if needed, we can assume $\mathbb{E}_{\lambda _{B_{i}}}(X) < \mathbb{E}_{\lambda _{B_{j}}}(X)$ for all $i<j$. This implies that $p(B_{i}) < p(B_{j})$ for all $i<j$. Thus, $\boldsymbol{p}$ belongs to $\boldsymbol{P}$. Hence, $\nabla  \psi(\boldsymbol{p})=0$ for some $\boldsymbol{p} \in \boldsymbol{P}$. Because $\psi$ is strictly convex on the convex set $\boldsymbol{P}$, there is at most one $\boldsymbol{p} \in \boldsymbol{P}$ such that $\nabla \psi(\boldsymbol{p})=0$.  We conclude that for every information structure $I \in \mathbb{I}(I_{o})$ there exists at most one pricing function $\boldsymbol{p}$ such that $(I,\boldsymbol{p})$ is implementable. 
 
 Furthermore, because the set $\boldsymbol{P}$ is an open set, we have $\nabla  \psi(\boldsymbol{p})=0$ if and only if $\boldsymbol{p}$ is the unique minimizer of the strictly convex function $\psi$ on $\boldsymbol{P}$. We conclude that    
$(I,\boldsymbol{p})$ is implementable if and only if $\boldsymbol{p}$ is the unique minimizer of $\psi$.     
\end{proof}



\begin{proof}[Proof of Theorem \ref{Thm: info structure1}]
We show that $\mathcal{C}^{Q}$ is regular. Then, Theorem \ref{Theorem: Main} implies that the optimal menu is $1$-separating, and hence, the optimal information structure consists of one set of sellers. We proceed with the following steps:

\textbf{Step 1.} Let $\{B\}$ be a $1$-separating information structure and let  $\{(p(B),\mathbb{E}_{\lambda _{B}}(X))\} \in \varphi ^{Q}(\{B\})$. Then for every $p>0$ we have $S_{\{B\}}(B ,p) \geq D_{\{B\}}(B,p)$ if and only if $p \geq p(B)$.

\textbf{Proof of Step 1.} Assume in contradiction that $p(B) > p>0$ and $S_{\{B\}}(B ,p) \geq D_{\{B\}}(B,p)$. Recall that the sellers' expected quality $\mathbb{E}_{\lambda _{B}}(X)$ does not depend on the price (see Lemma \ref{Lem: expected sellers}). We have
\begin{align*}  
 1 -F\left ( \frac {p}{\mathbb{E}_{\lambda _{B}}(X)} \right ) & = D_{\{B\}}(B,p) \leq S_{\{B\}}(B,p)  \\
& = \int _{B}g(x ,p)\phi (dx) \\
 & < \int _{B}g(x ,p(B))\phi (dx) \\
 &  = 1 -F\left ( \frac {p(B)}{\mathbb{E}_{\lambda _{B}}(X)} \right ) \end{align*}
 which is a contradiction to the fact that $F$ is increasing. The strict inequality follows because $g$ is strictly increasing in the price and $\phi(B)>0$ (see the proof of Lemma \ref{Lem: expected sellers}). The last equality follows from the fact that $\{(p(B),\mathbb{E}_{\lambda _{B}}(X))\} \in \varphi ^{Q}(\{B\})$. This proves that $S_{\{B\}}(B ,p) \geq D_{\{B\}}(B,p)$ implies $p \geq p(B)$. The other direction is proven in a similar manner.

\textbf{Step 2.} Suppose that $(\{B\},p(B))$  induces a menu that is maximal in $\mathcal{C}_{1}^{Q}$. Then $B \in I_{o}=\{A_{1},\ldots,A_{l}\}$. 


\textbf{Proof of Step 2.} Let $I=\{B\}$ be a $1$-separating information structure and assume that $B \neq A_{i}$ for all $A_{i} \in I_{o}$. Thus, $B$ is a union of at least two elements of $I_{o}$. Let $k$ be highest index among these elements. Hence, $\mathbb{E}_{\lambda _{A_{j}}}(X) \leq \mathbb{E}_{\lambda _{A_{k}}}(X)$ for all $A_{j} \subseteq B$, $A_{j} \in I_{o}$. We have   
\begin{align*}\mathbb{E}_{\lambda _{B}}(X)  & =\frac{\int _{B}x(k(x))^{ -1/\alpha}\phi (dx)}{\int _{B}(k(x))^{ -1/\alpha}\phi (dx)} \\ 
& =  \frac{\sum _{A_{i}:A_{i} \subseteq B, A_{i} \in I_{o}} \int _{A_{i}}x(k(x))^{ -1/\alpha}\phi (dx)}{\sum _{A_{i}:A_{i} \subseteq B, A_{i} \in I_{o}} \int _{A_{i}}(k(x))^{ -1/\alpha}\phi (dx)} \\
& \leq \frac{\int _{A_{k}}x(k(x))^{ -1/\alpha}\phi (dx)}{\int _{A_{k}}(k(x))^{ -1/\alpha}\phi (dx)} \\
& = \mathbb{E}_{\lambda _{A_{k}}}(X) .
\end{align*}
The first and last equalities follow from Lemma \ref{Lem: expected sellers}. The inequality follows from the elementary inequality $\sum _{i=1}^{n} x_{i} / \sum _{i=1}^{n} y_{i} \leq \max _{1\leq i \leq n} x_{i}/y_{i}$ for positive numbers $x_{1},\ldots,x_{n}$ and $y_{1},\ldots,y_{n}$. 

Assume that $(I,p(B))$ is implementable and that it induces the menu \{$(p(B),\mathbb{E}_{\lambda _{B}}(X))$\}. Then the arguments above imply $\mathbb{E}_{\lambda _{B}}(X) \leq \mathbb{E}_{\lambda _{A_{k}}}(X)$.   

We claim that $p(B) < p(A_{k})$ where $p(A_{k})$ is the (unique) equilibrium price under the information structure $\{A_{k}\}$ (the existence of this equilibrium price follows from the arguments in Step 3). To see this, note that 
\begin{align*}  S_{A_{k}}(B,p(B)) & = \int_{A_{k}} \left ( \frac{p(B)}{k(x)} \right)^{1/\alpha}\phi (dx)  \\
& < \int_{B} \left ( \frac{p(B)}{k(x)} \right)^{1/\alpha} \phi (dx)  \\
& = S_{I}(B,p(B)) = D_{I}(B,p(B)) \\
& =1-F\left( \frac{p(B)}{\mathbb{E}_{\lambda _{B}}(X)} \right ) \\ 
& \leq  1-F\left( \frac{p(B)}{\mathbb{E}_{\lambda _{A_{k}}}(X)} \right ) \\
& = D_{A_{k}}(B,p(B)).
\end{align*}
The first inequality follows from the facts that  $k$ is a positive function, $B \supseteq A_{k}$, and $\phi(B \setminus A_{k})>0$. The second inequality follows from the fact that $F$ is increasing. Hence, the demand exceeds the supply under the price $p(B)$. From Step 1 we have $p(B) < p(A_{k})$. Thus, the information structure-price pair  $(\{B\},p(B))$ does not induce a menu that is maximal in $\mathcal{C}^{Q}_{1}$.

\textbf{Step 3.} $\mathcal{C}^{Q}$ is regular.

\textbf{Proof of Step 3.} Let $(I ,\boldsymbol{p})$ be implementable where $I =\{B_{1} ,B_{2} ,\ldots  ,B_{n}\}$. Let $$C=\{(p(B_{1}),\mathbb{E}_{\lambda _{B_{1}}}(X)),\ldots,(p(B_{n}),\mathbb{E}_{\lambda _{B_{n}}}(X))\}$$ be the menu that is induced by $(I ,\boldsymbol{p})$. Suppose that $(D ,S ,\lambda )$ implements $(I,\boldsymbol{p})$. We can assume that $D(B_{i} ,\boldsymbol{p}) >0$ for all $B_{i} \in I$ and $0 <p(B_{1}) < \ldots  < p(B_{n})$ (see the proof of Proposition \ref{Prop: unique equilibrium}). Note that $D(B_{i} ,p) >0$ for $B_{i} \in I$ implies $0 <\mathbb{E}_{\lambda _{B_{1}}}(X) < \ldots  < \mathbb{E}_{\lambda _{B_{n}}}(X)$. 

Consider the $1$-separating information structure $I^{ \prime } =\{B_{n}\}$. 

We claim that there exists a $ p^{eq}(B_{n}) \geq p(B_{n})$ such that $(I^{ \prime } ,p^{eq}(B_{n}))$ is implementable and $(I^{ \prime } ,p^{eq}(B_{n}))$ induces the menu $\{(p^{eq}(B_{n}) ,\mathbb{E}_{\lambda _{B_{n}}}(X))\}$.

 From Step 1 in the proof of Theorem \ref{Theorem: Main}, we have $D_{I^{ \prime }}(B_{n} ,p(B_{n})) =1 -F \left ( \frac {p(B_{n})}{\mathbb{E}_{\lambda _{B_{n}}}(X)} \right )$. Note that there exists a $\overline{p} >p(B_{n})$ such that $D_{I^{ \prime }}(B_{n} ,\overline{p}) =0$ (for example we can choose $\overline{p} =\mathbb{E}_{\lambda _{B_{n}}}(X)b$). 

Define the excess demand function $\tau  :[p(B_{n}) ,\overline{p}] \rightarrow \mathbb{R}$ by $\tau ( \cdot ) =D_{I^{ \prime }}(B_{n} , \cdot ) -S_{I^{ \prime }}(B_{n}, \cdot )$. From the definition of $\overline{p}$ we have $\tau (\overline{p}) <0$.

 Note that 
\begin{align*}\tau (p(B_{n})) &  =D_{I^{ \prime }}(B_{n},p(B_{n})) -S_{I^{ \prime }}(B_{n} ,p(B_{n})) \\
 &  =D_{I^{ \prime }}(B_{n} ,p(B_{n})) -S_{I}(B_{n} ,p(B_{n})) \\
 &  \geq D_{I}(B_{n} ,\boldsymbol{p}) -S_{I}(B_{n} ,p(B_{n})) =0\end{align*}
 The first equality follows from the definition of $\tau $. The second equality follows from the fact that $S_{I}(B_{n} ,p(B_{n})) =S_{I^{ \prime }}(B_{n} ,p(B_{n})) =\int _{B_{n}}g(x ,p(B_{n}))\phi (dx)$, i.e., seller $x$'s optimal quantity decision does not change when the information structure changes. The inequality follows from the definition of the demand function. The last equality follows from the fact that $(I ,\boldsymbol{p})$ is implementable.

Because the distribution function $F$ and the optimal quantity function $g$ are continuous in the price, the excess demand function $\tau $  is continuous on $[p(B_{n}) ,\overline{p}]$. Thus, from the intermediate value theorem, there exists a $ p^{eq}(B_{n})$ in $[p(B_{n}) ,\overline{p}]$ such that $\tau (p^{eq}(B_{n})) =0$. We conclude that $(I^{ \prime } ,p^{eq}(B_{n}))$ is implementable and that $ p^{eq}(B_{n}) \geq p(B_{n})$. Thus, the menu $\{(p^{eq}(B_{n}) ,\mathbb{E}_{\lambda _{B_{n}}}(X))\}$ is a $1$-separating menu that belongs to $\mathcal{C}^{Q}_{1}$ and condition (i) of Definition \ref{Def: marketN} holds. 

Condition (ii) of Definition \ref{Def: marketN} immediately follows from using Step 2 to conclude that $B^{H} \in I_{o}$, and applying Step 1 to the information structure $\{B^{H}\}$. Thus, $\mathcal{C}^{Q}$ is regular. 

Theorem \ref{Theorem: Main} implies that the optimal $1$-separating menu is maximal. Combining this with Step 2 imply that the optimal $1$-separating information structure-price pair induces a menu that is maximal in $\mathcal{C}^{Q}_{1}$ and $ B^{\ast} \in I_{o}=\{A_{1},\ldots,A_{l}\}$ where $ I^{ \ast } := \{B^{\ast} \} $ is the optimal information structure. This concludes the proof of the Theorem. 
 \end{proof}
 
 \begin{proof}[Proof of Proposition \ref{prop:localModel1}]
 Let $(I ,\boldsymbol{p})$ be implementable where $I =\{B_{1} ,B_{2} ,\ldots  ,B_{n}\}$. Let $$C=\{(p(B_{1}),\mathbb{E}_{\lambda _{B_{1}}}(X)),\ldots,(p(B_{n}),\mathbb{E}_{\lambda _{B_{n}}}(X))\}$$ be the menu that is induced by $(I ,\boldsymbol{p})$.


(i) From the proof of Theorem \ref{Thm: info structure1} there exists a menu $C''= \{ (p^{eq}(B_{n}),\mathbb{E}_{\lambda _{B_{n}}}(X))\}$ with $p^{eq}(B_{n}) \geq p(B_{n} )$ that is feasible. Inequality (\ref{Ineq: supply}) together with Step 1 in the proof of Theorem \ref{Thm: info structure1} imply that $p^{eq}(B_{n}) \leq p^{M}(B_{n})$. Hence, the conditions of Proposition \ref{Proof of Prop: localB} are satisfied and the result follows. The proof of part (ii) is analogous. 
\end{proof}

 \begin{proof} [Proof of Proposition \ref{Prop:SupplyImbalnce1}]
Let $(I ,\boldsymbol{p})$ be implementable where $I =\{B_{1} ,B_{2} ,\ldots  ,B_{n}\}$. Let $$C=\{(p(B_{1}),\mathbb{E}_{\lambda _{B_{1}}}(X)),\ldots,(p(B_{n}),\mathbb{E}_{\lambda _{B_{n}}}(X))\}$$ be a menu that is induced by $(I ,\boldsymbol{p})$. From the proof Theorem \ref{Theorem: Main} the menu $ \{ (p(B_{n}), \mathbb{E}_{\lambda _{B_{n}}}(X)) \} $ yields more revenue than the menu $C$. 
Consider the $1$-separating information structure $I^{ \prime } =\{A_{l}\}$. Then from the Theorem's assumption $(I',p^{M}(A_{l})$ is implementable. We have 
$$ \pi^{Q} (I ,\boldsymbol{p}) \leq p(B_{n}) \left (1- F \left ( \frac{p(B_{n})} {\mathbb{E}_{\lambda _{B_{n}}}(X)} \right ) \right ) \leq p(B_{n}) \left (1- F \left ( \frac{p(B_{n})} {\mathbb{E}_{\lambda _{A_{l}}}(X)} \right ) \right ) \leq p^{M}(A_{l}) \left (1- F \left ( \frac{p^{M}(A_{l})} {\mathbb{E}_{\lambda _{A_{l}}}(X)} \right ) \right )  $$
which proves that there is a $1$-separating information structure that yields more revenue than the menu $C$. 
\end{proof}

\subsection{Proofs of Section \ref{Sec: Two-sided2 prices}}

\begin{proof}[Proof of Theorem \ref{thm: Bertrand}]
Let $I=\{B_{1},\ldots,B_{n}\}$ be an information structure and let $L(I)=\{G_{1},\ldots,G_{n}\}$. 

 (i) Suppose that $C \in \varphi ^{P}(I)$. Let $\boldsymbol{p}=(p(B_{1}),\ldots,p(B_{n}))$ be the equilibrium price vector that is associated with the menu $C$. We claim that $p(B_{i})=c(G_{i})$. 

If $p(B_{i}) < c(G_{i})$ then for every seller $x \in B_{i}$ we have $\overline{U}(x ,H(B_{i}) ,p(B_{i}) ,\boldsymbol{p}) < 0$ so the mass of sellers that participate in the platform equals to $0$ which contradicts the implementability of $I$. 
If $p(B_{i}) > c(G_{i})$ then the sellers' pricing decisions are not optimal. Sellers in $G_{i} \subseteq B_{i}$ can decrease their price and increase their utility. Thus, $I$ is not implementable. We conclude that $p(B_{i})=c(G_{i})$ for all $B_{i} \in I$. 

Let $B_{i} \in I$. Because $c(A_{i}) < c(A_{j})$ whenever $i<j$ we have $\overline{U}(x ,H(B_{i}) ,p(B_{i}) ,\boldsymbol{p}) < 0$ for sellers $x \in B_{i} \setminus G_{i}$ under the equilibrium price vector $\boldsymbol{p}=(c(G_{1}),\ldots,c(G_{n}))$. Thus, sellers in $B_{i} \setminus G_{i}$ do not participate in the platform and only the sellers in $G_{i} \subseteq B_{i}$ participate in the platform. This completes the proof of part (i).

(ii)  First note that $D_{\{B_{n}\}}(B_{n} ,c(G_{n})) \geq D_{I}(B_{n} ,(c(G_{1}),\ldots,c(G_{n}))) > 0$ (see the proof of Theorem \ref{Thm: info structure1}). Furthermore, under the price $c(G_{n})$, it is optimal for all the sellers in $G_{n} \subseteq B_{n}$ to participate in the platform and for all the sellers in $B_{n} \setminus G_{n}$ to not participate in the platform. So $\mathbb{E}_{\lambda _{G_{n}}}(X)$ is the sellers' expected quality given the sellers' optimal entry decisions and the price  $c(G_{n})$. Also, it is easy to see that the price $c(G_{n})$ maximizes the participating sellers' utility. From the quantity allocation function $h_{I}$ it follows immediately that the market clearing condition is satisfied. We conclude that $\{(c(G_{n}), \mathbb{E}_{\lambda _{G_{n}}}(X))\}  \in \varphi ^{P}(\{B_{n}\})$.

(iii) From part (i) we have $C_{o} =\{(c(A_{1}) ,\mathbb{E}_{\lambda _{A_{1}}}(X)) ,\ldots  ,(c(A_{l}) ,\mathbb{E}_{\lambda _{A_{l}}}(X))\}$. Let $C \in \varphi ^{P} (I)$. Then part (i) implies that $C =\{(c(G_{1}) ,\mathbb{E}_{\lambda _{G_{1}}}(X)) ,\ldots  ,(c(G_{n}) ,\mathbb{E}_{\lambda _{G_{n}}}(X))\}$. Thus $C \in 2^{C_{o}}$. We conclude that $\mathcal{C}^{P} \subseteq 2^{C_{o}}$. Now consider a menu  $C' = \{(c(A_{\mu_{1}}) ,\mathbb{E}_{\lambda _{A_{\mu_{1}  }}}(X)) ,\ldots  ,(c(A_{\mu_{j}}) ,\mathbb{E}_{\lambda _{A_{\mu_{j}}}}(X))\}  \in 2^{C_{o}}$ for sum increasing numbers $\{ \mu_{k} \}_{k=1}^{j}$.  Consider the information structure $I'=\{A_{\mu_{1}} , \ldots , A_{\mu_{j}} \} $. Because $I_{o}$ is implementable we have $D_{I'}(A_{\mu _{i}} ,(c(A_{ \mu _{1} }),\ldots,c(A_{ \mu _{j} }))) \geq D_{I_{o}}(A_{\mu _{i}}  ,(c(A_{1}),\ldots,c(A_{l}))) > 0$ for all $A_{\mu _{i} } \in I'$. An analogous argument to the argument in part (ii) shows that $I'$ is implementable and $ C' \in \varphi ^{P}(I')$. That is, $2^{C_{o}}  \subseteq \mathcal{C}^{P}$. We conclude that $2^{C_{o}}  = \mathcal{C}^{P}$ which proves part (iii). 
\end{proof}

\end{document}